\documentclass[10pt]{article}
\usepackage{graphicx}
\usepackage{graphics, cancel}
\usepackage{color}
\usepackage{amsmath, amsthm, slashed,bbold,upgreek}
\usepackage{amssymb,enumitem}
\usepackage{rotating}
\usepackage{epsfig}
\usepackage{verbatim}
\usepackage[margin=1.4in,dvips]{geometry}
\usepackage{rotating}
\usepackage{graphicx}
\usepackage[affil-it]{authblk}
\newtheorem{definition}{Definition}[section]
\newtheorem{theorem}{Theorem}[section]
\newtheorem{conjecture}{Conjecture}[section]
\newtheorem{corollary}{Corollary}[section]
\newtheorem{proposition}{Proposition}[section]
\newtheorem{lemma}{Lemma}[section]
\newtheorem{remark}{Remark}[section]

\newtheorem*{theorem*}{Theorem}
\newtheorem*{conjecture*}{Conjecture}
\DeclareMathOperator{\arccot}{arccot}
\title{A scattering theory construction \\ of dynamical vacuum black holes}

\author[1,4]{Mihalis Dafermos}
\author[2,4]{Gustav Holzegel}
\author[3,4]{Igor Rodnianski}
\affil[1]{\small University of Cambridge, Department of Pure Mathematics and Mathematical Statistics, Wilberforce~Road, Cambridge CB3 0WA, United Kingdom}

\affil[2]{Imperial College London, Department of Mathematics, South Kensington Campus, London~SW7~2AZ, United Kingdom}

\affil[3]{Massachusetts Institute of Technology, Department of Mathematics, 77~Massachusetts~Avenue,~Cambridge,~MA~02139, United States}

\affil[4]{Princeton University, Department of Mathematics, Fine Hall, Washington Road, Princeton,~NJ~08544, United States}

\begin{document}

\maketitle

\begin{abstract}
We construct a large class of dynamical vacuum black hole spacetimes whose exterior geometry asymptotically settles down to a fixed Schwarzschild or Kerr metric. The construction proceeds  by solving a backwards scattering problem for the Einstein vacuum
equations with characteristic data prescribed on the event horizon and (in the limit) at null infinity. The class admits the full ``functional'' degrees of freedom for the vacuum equations,
and thus our solutions will in general possess no geometric or algebraic symmetries. 
It is essential, however, for the construction that the scattering data (and  the resulting solution
spacetime) converge to 
 stationarity exponentially fast,  in advanced and retarded time,  their rate of decay  
 intimately related to the surface gravity of the event horizon. This can be traced back to the celebrated redshift effect, which in the context of backwards evolution  is seen as a blueshift.
\end{abstract}

\tableofcontents

\section{Introduction}

The question of the dynamical stability of vacuum black holes is a fundamental open problem in
classical  general relativity:
\begin{conjecture*}[Nonlinear stability of Kerr]
For all vacuum Cauchy data
 sufficiently
``near'' the data corresponding to a subextremal ($|a_0|<M_0$) 
Kerr metric $g_{a_0,M_0}$~\cite{Kerr}, 
the maximal vacuum 
Cauchy development spacetime~\cite{Geroch}
$(\mathcal{M},g)$
\begin{enumerate}
\item
possesses a complete null infinity $\mathcal{I}^+$ (cf.~\cite{Chrmil}) whose
past $J^-(\mathcal{I}^+)$ is bounded in the future
by a smooth affine complete event horizon $\mathcal{H}^+\subset\mathcal{M}$, 
\item
stays globally close to $g_{a_0,M_0}$ in $J^-(\mathcal{I}^+)$, 
\item 
asymptotically settles down in $J^-(\mathcal{I}^+)$
 to a nearby subextremal member of
the Kerr family $g_{a,M}$ with  parameters $a\approx a_0$ and $M\approx M_0$.
\end{enumerate}
\end{conjecture*}

The problem poses a considerable challenge.
The \emph{Einstein vacuum  equations} 
\begin{equation}
\label{vacEq}
{\rm Ric}(g)=0
\end{equation}
constitute a complicated system of nonlinear
hyperbolic equations for an unknown $3+1$-dimensional Lorentzian metric $g$.
In the asymptotically flat setting,  
non-linear stability under the evolution of $(\ref{vacEq})$
is known only for the trivial solution
 \emph{Minkowski space} $\mathbb R^{3+1}$, as was proven in monumental
work of Christodoulou and Klainerman~\cite{ChristKlei}.
Resolving the above conjecture
 would require understanding the long time dynamics of
 general solutions of the Cauchy problem in 
 a neighbourhood of the highly non-trivial \emph{Kerr} family of vacuum metrics.
Moreover, due to the ``supercriticality'' of the nonlinearity implicit in equations
$(\ref{vacEq})$, statements 1,~2 and~3 of the  conjecture 
 are in fact strongly coupled, and one thus  expects to have to  prove
 them  all simultaneously. 

Recent work on this problem has been concentrated in two directions: 
One direction has been to indeed look at  the fully \emph{non-linear} Einstein equations,
further coupled\footnote{In $3+1$ dimensions under spherical symmetry,
in view of Birkhoff's theorem~\cite{Wald},
it is \emph{necessary} in this context to couple $(\ref{vacEq})$
with matter fields
so as to yield non-trivial dynamics. The simplest choice
to consider is then the \emph{Einstein--scalar field} 
system: ${\rm Ric}(g)=\nabla\psi \otimes\nabla \psi$. 
In $4+1$ dimensions, however, the vacuum equations $(\ref{vacEq})$ themselves admit
 a more interesting
Bianchi IX symmetry assumption~\cite{Bizon}, compatible
with asymptotic flatness, which is governed
by two dynamic degrees of freedom
satisfying a system of $1+1$-dimensional pde's.
The stability of $4+1$-dimensional
Schwarzschild  has been resolved in this symmetric
setting  in~\cite{HolBi9}. In the
present paper, in what follows we shall always consider the classical context of $3+1$ dimensions.}
to additional matter fields, but \emph{under spherical symmetry} \cite{ChrMToGC, DafRod}, so
as for the problem to effectively reduce to a system of $1+1$-dimensional hyperbolic
pde's.
This effectively breaks the supercriticality\footnote{One can understand this as follows:
In spherical symmetry, suitable Einstein-matter systems become subcritical with respect to the conserved quantity given by the flux of
Hawking mass, \emph{provided that the area radius function $r$ is bounded away from $0$}. 
The latter bound indeed holds in the domain of outer
communications to the future of any hypersurface
containing a \emph{marginally trapped surface}, in particular, for small perturbations
of Schwarzschild.
Using the aforementioned conserved quantity,
the existence of a complete null infinity $\mathcal{I}^+$
can then be proven \emph{without proving asymptotic 
stability}~\cite{Mihali1}. In contrast, the  stability problem for
Minkowski space
retains its ``supercriticality'' even under spherical symmetry, in view of the presence of the 
regular centre $r=0$. Nonetheless, that problem remains much easier (see~\cite{Christodoulou}) than its non-spherical symmetric
version~\cite{ChristKlei}!}
 of the problem and decouples 1--3 above. 
 In this context, the analogue of the above conjecture 
 is completely understood for the simplest matter models, but at the expense of changing the very
 nature of the analysis and suppressing (via the symmetry assumption) several
 of the most interesting phenomena associated to black holes.
 
The second main direction of study
has restricted itself to
the \emph{linear} problem associated with (\ref{vacEq}),
where the background metric $g$ is fixed,
but where the linear fields considered are now \emph{without} symmetry. See~\cite{mihalisnotes}
for an extensive review. 
In the latter direction, even the case of the Cauchy problem for the linear \emph{scalar} wave equation
\begin{equation}
\label{scalarwave}
\Box_g\psi=0
\end{equation}
on a fixed subextremal Kerr exterior background $(\mathcal{M},g_{a,M})$--what 
can reasonably be viewed as 
a ``poor man's'' linearisation of $(\ref{vacEq})$, completely neglecting the tensorial
structure--has only recently been understood: first for the Schwarzschild case $a=0$ in \cite{KayWald, DafRod2, Sterbenz}, then for the slowly rotating case $|a| \ll M$ in \cite{mihalisnotes, dafrodsmalla, Toha2, AndBlue} and finally for the full case $|a|<M$ in \cite{dafrodlargea, SRT}. Note also \cite{Whiting, Finster}. The extremal case $|a|=M$, on the other hand, is subject to the recently discovered \emph{Aretakis instability} \cite{Aretakis, Aretakis2, Reallextreme, Aretakis3}, hence its exclusion from our formulation of  the conjecture.

Beyond  the above two directions, the
study of the dynamical stability of black holes  is still terra incognita.
Better understanding of  the full non-linear Kerr stability conjecture
is hampered by the fact that a much more basic question has
not yet been answered:
\begin{quotation}
\noindent{\it 
Are there \underline{any} examples of dynamical vacuum black hole
spacetimes which radiate  nontrivially for all time  to both a complete event horizon 
$\mathcal{H}^+$ and to
a complete null infinity $\mathcal{I}^+$ and asymptotically settle down
to Schwarzschild or Kerr?}
\end{quotation}

The purpose of the present paper is to answer the above question in the affirmative, in fact,
to address
both the issue of existence\footnote{The closest previously known examples exhibiting such dynamics are the so-called Robinson-Trautmann
solutions, which we shall discuss in Section \ref{sec:RobTraut}.}
of such examples and how a suitably general such class can be effectively 
parametrized. 
Our main result can be summarised:
\begin{theorem*}
For all $|a|\le M$,
there exist  smooth vacuum black hole spacetimes,  
parametrized by ``scattering data'' on a complete event horizon
$\mathcal{H}^+$ and a complete
null infinity $\mathcal{I}^+$ (with the
full functional degrees of freedom), 
which asymptotically settle down  to the Kerr metric $g_{a,M}$. 
\end{theorem*}

We have here stated the result somewhat loosely.
A  precise version of the theorem in the Schwarzschild case 
$a=0$ is stated as Theorem \ref{theo:full} of Section \ref{sec:maintheorem}, which in turn follows from Theorems \ref{theo1} and \ref{theo2}. The case of the general Kerr family  is considered  
in Section~\ref{sec:thekerrcase}. As we shall discuss below, the Kerr case produces no additional conceptual difficulties but is computationally more involved and less explicit, hence the
pedagogical advantage of treating first Schwarzschild in detail. We also remark that by ``smooth" above, we in fact mean solutions of arbitrary prescribed finite regularity, cf.~Remark \ref{rem:Pdep}.
Finally, we note explicitly that the above theorem includes the extremal case $|a|=M$, despite
its exclusion from the stability of Kerr conjecture.
We will return to this surprising fact later on.

Let us note that a general notion of vacuum spacetimes asymptotically settling down to Schwarzschild was in fact first introduced in~\cite{Holzegelspin2}. Though~\cite{Holzegelspin2}
did not produce examples of such spacetimes, it proved that \emph{given such a spacetime},
certain higher order energies (associated with the propagating degrees of freedom of the vacuum equations $(\ref{vacEq})$) decayed quantitatively with 
respect to a suitable foliation. The point was that the decay bounds were estimable from 
initial data augmented only by certain
global a priori decay estimates for a finite number of lower order energies. This type of 
result  was motivated by its possible usefulness in the
context of the stability problem. 
Our  theorem above in particular gives the first non-trivial
examples of  spacetimes satisfying the assumptions of~\cite{Holzegelspin2}.

While still far from resolving the stability conjecture with which we set out,
the above theorem confirms that the qualitative  picture of dynamical vacuum
black holes radiating their ``dynamical degrees of freedom'' to the event horizon and to 
(a complete) null infinity,  finally asymptotically settling down 
to the Kerr family,       is indeed reflected by
a non-trivial class of spacetimes with no geometric or algebraic symmetries.
We will give further discussion of the relation of what we have proven
with the full stability problem at various points in the remainder of this introduction.

\subsection{Brief overview} \label{sec:mainpoints}

In this section, we give a brief overview of the main ideas behind our construction, beginning
with the basic setup. In our discussion, we will assume some
familiarity with general relativity and the geometry of the Schwarzschild
and Kerr families in particular. 
We refer the reader to the textbook~\cite{Wald} for  background, the lecture notes~\cite{mihalisnotes} for a mathematical discussion of
black holes with emphasis on wave propagation, and the monumental works of Christodoulou
and Klainerman~\cite{ChristKlei} and Christodoulou~\cite{formationofbh}
on the global analysis of the vacuum equations $(\ref{vacEq})$, works which provide the
ultimate source for many ideas which 
will be used here.

\subsubsection{The setup}
\label{thesetuphere}
Scattering data for the Einstein vacuum equations (\ref{vacEq}) are posed on (what will be) the event horizon $\mathcal{H}^+$ and null infinity $\mathcal{I}^+$ of our spacetime.
The geometry defined by the scattering data on $\mathcal{H}^+$ and $\mathcal{I}^+$ will  approach \emph{exponentially}\footnote{Following Christodoulou~\cite{formationofbh},
null characteristic data can be parametrized by the geometry on a fixed
sphere and a seed ``function'' representing the
conformal geometry of the cone. Exponential approach
can here be characterized by suitable exponential decay of the seed function.}
 the geometry of a Kerr event horizon and a Kerr null infinity, in 
appropriately defined advanced and retarded time respectively, for 
some fixed parameters $|a|\le M$.
The necessity of assuming exponential approach will be discussed in
Section~\ref{whyexp}  below.

Our objective is then to construct a  spacetime $\left(\mathcal{M},g\right)$ solving $(\ref{vacEq})$,
whose domain of existence can be pictured by the shaded region of the Penrose diagram below: 
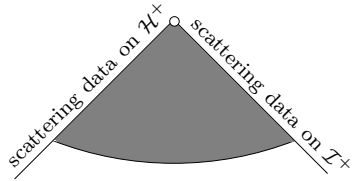
\[
\input{intro1.pstex_t}
\]
which  attains the prescribed  scattering data as induced geometry
on the horizon $\mathcal{H}^+$ and null infinity
$\mathcal{I}^+$, the latter as an asymptotic limit. The data on null infinity $\mathcal{I}^+$
then encodes gravitational radiation to far away observers.

The actual construction is taken via the limit of an associated finite problem; see
the figure below:
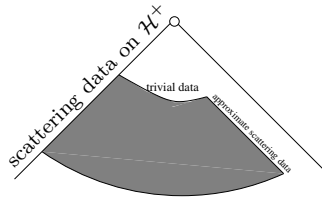
\[
\begin{picture}(0,0)%
\includegraphics{intro2.pstex}%
\end{picture}%
\setlength{\unitlength}{1816sp}%
\begingroup\makeatletter\ifx\SetFigFont\undefined%
\gdef\SetFigFont#1#2#3#4#5{%
  \reset@font\fontsize{#1}{#2pt}%
  \fontfamily{#3}\fontseries{#4}\fontshape{#5}%
  \selectfont}%
\fi\endgroup%
\begin{picture}(4464,2462)(3949,-3671)
\put(4126,-3286){\rotatebox{45.0}{\makebox(0,0)[lb]{\smash{{\SetFigFont{8}{9.6}{\rmdefault}{\mddefault}{\updefault}{\color[rgb]{0,0,0}scattering data on $\mathcal{H}^+$}%
}}}}}
\end{picture}%

\]
Here the scattering data on the horizon  $\mathcal{H}^+$ is ``cut off'' at late finite  advanced
time to  match to Kerr data while the data on null infinity 
$\mathcal{I}^+$ is approximated by data on an ingoing null cone, again cut off at late retarded time,
and these data
are further supplemented by trivial Kerr data on a spacelike hypersurface connecting
the two. (See Section \ref{sec:local} for a detailed discussion of setting up the data.) 
This now defines a mixed Cauchy-characteristic initial value problem whose local well posedness in the smooth category\footnote{In spaces of finite differentiability, there is an inherent loss
of derivatives in the characteristic initial value problem, even for the linear wave equation
(see~\cite{formationofbh}).
In our setting, this loss is  for instance already
reflected in the numerology of the table of Section~\ref{sec:redsrole}.\label{alosshere}} essentially follows from \cite{ChoquetBruhat} and \cite{Rendall} (as applied in \cite{formationofbh}). The global estimates on solutions of this finite problem (to be discussed below)
will indeed allow one to infer existence 
of a solution to the original limiting problem, \emph{not, however, uniqueness}. 
For the latter, see the discussion
in Section~\ref{difrences} and again in Section \ref{sec:uniq}.
The estimates will moreover show that the spacetime indeed possesses a complete
null infinity and approaches the Kerr metric
uniformly in the shaded region with respect to a suitable foliation.

We will return to these global estimates for the finite approximate problem 
in Section~\ref{THEMAIN} later in this introduction.   To set the stage for these, 
we shall first introduce in Section~\ref{renor} 
the double null gauge with its 
 corresponding 
Ricci connection coefficients and null-frame curvature components,
renormalised by subtracting out a background Kerr solution.
The analytic content of the Einstein equations $(\ref{vacEq})$ is then captured
by the null structure equations and Bianchi equations, whose schematic
form under our renormalisation is discussed in Section~\ref{MOSTBASICSCHEM}.  
We then turn to a discussion of the two fundamental
issues which govern our setup:
first,  in Section~\ref{NSBE}, the problem
of capturing the ``null condition'' at null infinity  $\mathcal{I}^+$ (a difficulty familiar
from~\cite{ChristKlei}), 
and then, in Section~\ref{whyexp}, the role of the
red-shift effect near $\mathcal{H}^+$, which in the context of our problem
will in fact appear as a blue-shift and which is the essential  origin
of our strong exponential decay assumptions.
Given the latter assumption and a proper understanding of the null condition,
the global estimates of the proof described in Section~\ref{THEMAIN}
will in fact be relatively straightforward.

\subsubsection{A renormalised double null gauge}
\label{renor}
As in~\cite{formationofbh, KlaiNic},
we will capture the  content of the Einstein equations in
the form of the structure equations associated to a \emph{double null foliation},
which in our case will
cover the shaded
region above.
(This formalism  is particularly suited for our problem in view of the geometry of
the Schwarzschild and Kerr black hole exteriors, which, \emph{in contrast to 
the case of Minkowski space $\mathbb R^{3+1}$}, can
be \emph{globally} covered by such foliations without degeneration.)
The leaves of this foliation are null hypersurfaces defined as level sets of functions $u$ and $v$
which we identify with \emph{retarded} and \emph{advanced time}, respectively,
and will intersect in $2$-spheres.  The event horizon $\mathcal{H}^+$ will correspond to $u=\infty$ and null infinity
$\mathcal{I}^+$ with $v=\infty$, while the finite approximation to null infinity will be
a null hypersurface $v=v_\infty$.
This foliation defines two useful null frames $(\ref{frefdef})$ and $(\ref{kruskfr})$  (differing in normalisation), Ricci coefficients $(\ref{RicC})$,
for instance the outgoing shear $\hat\chi$ or expansion $tr\chi$, and components of
the Riemann curvature tensor $(\ref{curvdefins})$, 
for instance the components $\alpha_{AB}$ and $\rho$,
which satisfy the so-called \emph{null structure
equations} and the  \emph{Bianchi equations}.  
It is these two sets of equations which we shall
use to estimate solutions.  

A novelty with respect to previous work is that we shall in fact \emph{renormalise} these
two sets of equations by subtracting the contribution due to the background Schwarzschild or more generally Kerr  metric
to which the solution is to asymptotically approach. 
This background conveniently also defines the differentiable structure of the ambient manifold
(with its coordinates $u$, $v$) on which
everything is then to be defined.

In the simpler case where the background is to be Schwarzschild, the
ambient differential structure
is introduced in Section~\ref{sec:manifold} and
 the renormalised null structure and Bianchi equations are first introduced
in Section~\ref{therebne}. 
The former equations concern quantities including for example the \emph{outgoing shear} and
\emph{renormalised expansion},
\begin{equation}
\label{forexampcon}
\hat \chi(u,v,\theta,\phi) ,\qquad (tr\chi- tr\chi_{\circ}) (u,v,\theta,\phi),
\end{equation}
respectively (see formulas $(\ref{RicC})$, $(\ref{eq:ex0})$), whereas
the latter concern quantities including for example
the (renormalised) Riemann components
\begin{equation}
\label{forexambian}
\alpha_{AB}(u,v,\theta,\phi), \qquad 
 \qquad (\rho-\rho_\circ)(u,v,\theta,\phi)
\end{equation}
 (see formulas~$(\ref{curvdefins})$).
 The second example in each of $(\ref{forexampcon})$, $(\ref{forexambian})$ has been
nontrivially  renormalised by subtracting
a Schwarzschild contribution (denoted by a subsript $\circ$), whereas 
in the first, the
Schwarzschild contribution would vanish.
Concretely, $tr\chi_\circ(u,v,\theta,\phi)=2r^{-1}(1-2Mr^{-1})$, where
$r(u,v,\theta,\phi)$ is the Schwarzschild area-radius function, a function
we shall use in what follows to understand decay properties towards null infinity $\mathcal{I}^+$.
Notice that everything is manifestly defined
as functions or tensors of ambient coordinates  $(u,v,\theta,\phi)$,
and this is what fundamentally allows comparison of say $tr\chi$ and $tr\chi_\circ$, etc.

In the Kerr case (Section~\ref{sec:thekerrcase}), the renormalisation is more 
involved, since in our double null foliation gauge
all Ricci and curvature components  appear non-trivially\footnote{In contrast, the null frame
which makes the algebraically special property of Kerr manifest is non-integrable.} 
for the background Kerr solution, and moreover, as functions of  ambient coordinates they
are only given implicitly  and  depend  nontrivially
on two variables, cf.~Section \ref{sec:thekerrcasere}. 
In both cases, it can be verified that the renormalised  null structure and Bianchi equations 
share a common ``schematic'' form characterized by a certain hierarchial structure in 
the non-linear interactions and their decay properties. 
The subsequent analysis only depends on this form; in particular, the above
complications aside,
\emph{the Kerr case yields no additional conceptual difficulties over Schwarzschild}.

\subsubsection{The schematic form of the equations}\label{MOSTBASICSCHEM}
Before proceeding further, let us make 
 some of our above comments  more explicit 
by giving a first glimpse  of the schematic from of the
renormalised null structure and Bianchi 
equations. (The form will be further elaborated in Section~\ref{NSBE}.)

We will denote by 
 $\Gamma$  an arbitrary renormalised Ricci coefficient and by
 $\psi$ an arbitrary renormalised      Riemann curvature component.
(These are defined in Section~\ref{sec:Gamdef} in the case of a Schwarzschild background,
and in Section~\ref{sec:thekerrcase} for the general Kerr case.)
Thus, quantities $(\ref{forexampcon})$
are examples of a $\Gamma$ 
while $(\ref{forexambian})$ are examples of $\psi$.
The (renormalised) Riemann curvature components $\psi$ can in turn be grouped
into so-called \emph{Bianchi pairs}, which we will denote $(\uppsi, \uppsi')$.
(For instance, $(\uppsi, \uppsi') = (\alpha,\beta )$ , $\big((\rho-\rho_\circ,\sigma),\underline\beta\big)$ are examples of Bianchi pairs in the Schwarzschild case.) We note that the same component can appear in the role
of $\uppsi$ or $\uppsi'$ in distinct Bianchi pairs.

The content of the evolution aspect\footnote{Besides $(\ref{introforgamma})$, there are additional null structure relations
linking Ricci coefficients $\Gamma$ and curvature components $\psi$ 
by \emph{elliptic} equations on spheres.
These equations are given in Appendix~\ref{appendix}.
These additional relations 
can be interpreted as constraint equations which, if satisfied initially, are satisfied
subsequently as a consequence of the remaining equations. We shall thus
be able to estimate solutions of $(\ref{vacEq})$ without ever invoking explicitly the equations
of Appendix~\ref{appendix} (although we shall of course need to specifically
impose the latter
to constuct initial data (see Section~\ref{sec:hozco})!). 
We note that by not directly exploiting these additional null structure equations, 
 we obtain less sharp results with respect to regularity, but this is nonetheless
 sufficient for our purposes. These elliptic equations will, however, make a brief appearance in
 Section~\ref{difrences}; see footnote~\ref{footdif}.\label{longfoo}} of
the Einstein vacuum equations is given by the \emph{null structure
equations} for $\Gamma$, which can all be written schematically in the form
\begin{equation}
\label{introforgamma}
\slashed{\nabla}_3\Gamma = f \Gamma+\Gamma\cdot\Gamma+\psi, \qquad
\slashed{\nabla}_4\Gamma = f \Gamma+ \Gamma\cdot \Gamma+\psi
\end{equation}
and the \emph{Bianchi equations} for  Bianchi pairs $(\uppsi, \uppsi')$, which can
all be written schematically as
\begin{equation}
\label{introbianchi}
\slashed{\nabla}_3\uppsi =  \slashed{\mathcal{D}}\uppsi' +f\psi +f\Gamma +\Gamma\cdot \psi,
\qquad
\slashed{\nabla}_4\uppsi'= \slashed{\mathcal{D}}\uppsi+f\psi +f\Gamma +\Gamma\cdot \psi.
\end{equation}
Here $f$ denotes a known function or tensor (arising from the background Schwarzschild or more
generally Kerr metric), $f\Gamma$, $\Gamma\cdot\Gamma$, $\Gamma\cdot\psi$ etc., denote in 
fact \emph{sums} over various
contractions of the product of known functions and elements of $\Gamma$, etc.
The operators  $\slashed{\nabla}_3$, $\slashed{\nabla}_4$ are  appropriate
first order differential operators acting in the directions of the null vectors $e_3$,
$e_4$, tangential to the constant-$v$- and constant-$u$ null hypersurfaces
respectively, 
whereas $\slashed{\mathcal{D}}$ are first order differential operators on the
spheres of intersection of the null cones.

The essential hyperbolicity of the Einstein equations $(\ref{vacEq})$ is encoded in the Bianchi equations $(\ref{introbianchi})$,
which can be controlled by \emph{energy estimates}. 
(The latter can be derived via the Bel Robinson tensor, see~\cite{ChristKlei},
but alternatively more directly upon multiplication of each couple of equations for Bianchi pairs by
$\uppsi$, $\uppsi'$, respectively, and integration by
parts,  exploiting the divergence structure in the angular
operators $\slashed{\mathcal{D}}$.)
The null structure equations $(\ref{introforgamma})$ on the other hand are  here
estimated (solely, cf.~footnote~\ref{longfoo}) as transport equations. 
As is clear from above,  equations $(\ref{introforgamma})$
and $(\ref{introbianchi})$ are
coupled and must be estimated together.

An  important feature of the coupling in $(\ref{introforgamma})$ is that the components of curvature
$\psi$ appearing on the right hand side  are such that,
 upon 
integration along the appropriate null hypersurface, they can be estimated by a flux associated 
to the energy estimates for $(\ref{introbianchi})$. 
To obtain the latter, however, one needs to estimate \emph{higher} $L^q$ norms of
$\Gamma$ and $\psi$. 
These are in turn obtainable from higher order $L^2$ 
estimates via Sobolev inequalities. 
Thus,
one must also commute equations  $(\ref{introforgamma})$--$(\ref{introbianchi})$
with suitable differential operators, and derive higher order estimates.
It turns out that the schematic structure of  $(\ref{introforgamma})$--$(\ref{introbianchi})$
 is preserved
under appropriate commutation. We will return to this point in Section~\ref{compreserves}, after we have
further elaborated on this structure.

\subsubsection{Asymptotics towards $\mathcal{I}^+$ and the ``null condition''}
\label{NSBE}
The  level of structure exhibited by
$(\ref{introforgamma})$, $(\ref{introbianchi})$
just discussed, though already non-trivial,  is as such
sufficient only to prove \emph{local} estimates for $(\ref{vacEq})$.

The first global aspect that must be addressed is how to obtain
uniform estimates \emph{up to null infinity}
 $\mathcal{I}^+$, say at first instance only for a finite interval of
retarded time $u$. In analogy already with the semi-linear wave equation
\begin{equation}
\label{thesemil}
\Box\psi = Q(\nabla \psi,\nabla\psi)
\end{equation}
on $\mathbb{R}^{3+1}$, one can only hope to obtain such uniform control
by at the same time capturing \emph{decay} properties of the solutions.
In the case of $(\ref{vacEq})$, more specifically, for this one must\footnote{Problems (i)--(ii) are of course interesting in themselves! Moreover, in 
view of the fact that in our setup, we are imposing \emph{data} on $\mathcal{I}^+$, we need
at the very least
some basic understanding of the asymptotics just to say that the solution
indeed attains the scattering data. 
But more fundamentally,  (i)--(ii)  will be necessary simply to ensure the existence
of solutions up to $\mathcal{I}^+$.}
\begin{itemize}
\item[(i)]
guess the 
 correct hierarchy of asymptotics towards $\mathcal{I}^+$ for the various Ricci coefficients $\Gamma$
 and curvature components $\psi$,
 and 
  \item[(ii)] show that the non-linear structure of the interactions in
  $(\ref{introforgamma})$ and $(\ref{introbianchi})$ indeed allows
 for the propagation of this hierarchy, at least locally in retarded time $u$. 
\end{itemize}

The study of (i) for the Einstein vacuum equations $(\ref{vacEq})$ has a long history in connection with understanding
gravitational radiation. 
The pioneering works in the subject are due to Pirani, Trautman, Bondi and then
Penrose (see for instance~\cite{Penroseasymp}), 
who,
\emph{imposing} various \emph{a priori}  basic assumptions on the asymptotics, derived from these 
a specific hierarchy of decay rates for various curvature components, known as \emph{peeling}.
At the time, however, it remained completely unclear whether
there were any non-trivial spacetimes that satisfied these assumptions,
 much less whether they held for general solutions of the Cauchy
problem  arising from asymptotically flat data.

It was only  with the monumental work of Christodoulou and Klainerman~\cite{ChristKlei}
on the  stability of Minkowski space
that  the door opened to a definitive understanding of the
question of asymptotic structure for
$(\ref{vacEq})$. In the context of the global stability of Minkowski space,
it turned out that
 it was more natural to propagate (cf.~(ii))
a slightly weaker version of the original  peeling 
hierarchy of~\cite{Penroseasymp}. 
This propagation
 was ensured in~\cite{ChristKlei} by
 a careful
 analysis of the (null-decomposed) error-terms in the energy estimate for the Bianchi equations
arising from the Bel-Robinson tensor applied to appropriate vector field multipliers.
One can think of the consistency of the decay of all nonlinear error-terms with 
the hierarchy (i)
 as an elaborate version
 of the ``null condition''~\cite{Klainull}.\footnote{The difficulties of problem (ii)
are familiar from the semi-linear example $(\ref{thesemil})$. Not all non-linearities
$Q$ admit existence results for $\psi$ up to $\mathcal{I}^+$ even for finite $u$; a sufficient condition is that $Q$ satisfies the celebrated
null condition~\cite{Klainull}. As is well known however,
when cast in the the form of non-linear
wave equations by imposing the harmonic
gauge $g^{\mu\nu}\Gamma^{\lambda}_{\mu\nu}=0$,
the Einstein equations $(\ref{vacEq})$
do \emph{not} satisfy the null 
condition of~\cite{Klainull}; this is what made the problem of
stability of Minkowski space so difficult! This in turn is related to the fact that in such coordinates the 
asymptotics of $g_{\mu\nu}$ do not correspond to the
asymptotics of free waves in view of the logarithmic divergence of the 
light cones. Only much more recently was the harmonic gauge successfully used
to give
a new proof~\cite{Igor2} of a version of stability of Minkowski space, 
by propagating a weaker hierarchy of asymptotics.}

 Christodoulou in fact subsequently showed 
 in~\cite{mgrome} that 
 generic physically interesting Cauchy data \emph{never} satisfy the original full peeling
 hierarchy of~\cite{Penroseasymp}, the obstructions having the interpretation of moments encoding
 the past history of the system, explicitly calculable in 
 the post-Newtonian
 approximation.\footnote{The analysis of~\cite{ChristKlei} was localised 
 near null infinity in~\cite{KlaiNic} in a double null gauge
similar to the one applied here, where it was
shown that the full peeling hierarchy could be propagated if it was assumed on a given
outgoing null cone.
 See also~\cite{Friedrich}.}

Here, we are going to make the analogue of the ``null condition'' used
in the present work  manifest in a 
slightly more direct way,  at a level more readily read off from the null-decomposed
null-structure $(\ref{introforgamma})$ and 
Bianchi  equations $(\ref{introbianchi})$ 
themselves. This approach systematizes 
and extends observations made in~\cite{Holzegelspin2, LukRod}
and may be useful for other problems.

\paragraph{The $p$-index notation.}
To see our version of the ``null condition'' in the systematic form of the equations $(\ref{introforgamma})$, $(\ref{introbianchi})$,
we first must introduce some additional 
notation.

We will assign (see  Section \ref{sec:equations})
 to each renormalised\footnote{The considerations of this section 
 could also be applied to the \emph{un}renormalised equations; for this
all instances of  $f_p$ would be replaced
by $\Gamma_p$ or $\psi_p$.} connection coefficient $\Gamma$
and   curvature component $\psi$, as well as to each known tensor $f$ arising
from the background,  
 a characteristic weight $p$,
which we will denote by a subscript, which will reflect the fact that $r^p\|f_p\|$ has, and
$r^p\| \Gamma_p\|$,
$r^p\| \psi_p\|$ are expected to have, a finite (possibly zero\footnote{That is to say,
the weight may in fact be weaker than the actual decay rate in the solutions we finally
construct.}) trace on $\mathcal{I}^+$. Here, $r$ is the
Schwarzschild area-radius function referred to already in Section~\ref{renor}, i.e.~it is a known 
function of the ambient coordiantes $(u,v)$.
This procedure thus encodes our
``guess'' (i) above.

For $\Gamma$, we will further distinguish components with a ${}^{(3)}$, ${}^{(4)}$ 
superscript according
to whether the
relevant component satisfies an equation in the ingoing or outgoing null direction.

For instance, with this notation, the ingoing shear $\hat{\underline\chi}$ (which is unrenormalised
in the case of Schwarzschild background) can be written as 
\[
\hat{\underline\chi}=\overset {(3)}{\Gamma}_1,
\]
indicating   that this quantity satisfies an equation 
\[
\slashed{\nabla}_3\hat{\underline\chi}=\ldots
\]
and that $r^1\|\hat{\underline\chi}\|$ is expected to have a finite trace on $\mathcal{I}^+$.

Other examples of our use of notation in the case where the background is Schwarzschild are
\[
tr\chi-tr\chi_\circ = \overset{(4)}{\Gamma}_2, \qquad 
\underline\alpha=\psi_1.
\]

\paragraph{The null structure equations for $\Gamma$.}\label{NSEFGa}
With this notation, we may elaborate the structure of  
the equations $(\ref{introforgamma})$  (see Proposition~\ref{uneq})
as follows:
\begin{align}
\label{moreprecisely0}
&\slashed{\nabla}_3 \overset {(3)}{\Gamma}_p  
&=& \,O_p  \\ 
\label{moreprecisely1}
&\slashed{\nabla}_4 \left(r^{2c[\overset {(4)}{\Gamma}_p]}  \overset {(4)}{\Gamma}_p  \right) 
&=&\,
{r^{2c[ \overset {(4)}{\Gamma}_p]} } \left(\sum_{p_1+p_2=p+1}f_{p_1}\overset{(3)}\Gamma_{p_2}
+{O}_{p+\frac32}\right) &=&\,{r^{2c[ \overset {(4)}{\Gamma}_p]} }O_{p+1}
\end{align}
where by $O_z$ we mean a sum of products of $f$, $\Gamma$, $\psi$, whose total decay as measured
in the above sense is $r^{-z}$ or faster. Here, $c[\Gamma_p]$ is a weight factor
defined in Proposition~\ref{uneq}.

From the above, one sees immediately  that the decay assumptions on $\Gamma_p$
are \emph{consistent}, in the sense that the decay
assumptions on the differentiated quantities on the left hand side
of $(\ref{moreprecisely0})$, $(\ref{moreprecisely1})$
 are ``retrieved'' by inserting the decay assumptions of the hiearchy
 on the 
right hand side and integrating $(\ref{moreprecisely0})$--$(\ref{moreprecisely1})$
as transport equations for finite
affine retarded time in the $\slashed{\nabla}_3$ direction 
\begin{equation}
\label{fromheresmall}
\int_{u}^{u+\epsilon}{ O_p}\lesssim \epsilon O_p
\end{equation}
and for infinite affine advanced time in the 
$\slashed\nabla_4$ direction,
\begin{equation}
\label{borderlinet}
\int_{v}^{\infty}{r^{2c[ \overset {(4)}{\Gamma}_p]}O_{p+1}} dv \lesssim r^{2c[ \overset {(4)}{\Gamma}_p]}O_{p},
\end{equation}
where we are exploiting  the extra decay in $r$ to integrate,
noting  also that for large $r$, $dv\sim dr$.

In the above computation, however, inequality $(\ref{borderlinet})$ fails
to yield a smallness parameter. Thus, in itself, the above computation does
not allow 
to prove estimates. 
As is apparent, however, from the precise structure of the middle term of $(\ref{moreprecisely1})$,
the \emph{borderline} terms, i.e.~those which are $O_{p+1}$ and no better (as measured in the $p$-subscript notation), are 
of the form 
\[
f_{p_1}\overset{(3)}\Gamma_{p_2}
\]
where
$\overset{(3)}\Gamma_{p_2}$ can be estimated by $(\ref{moreprecisely0})$.
Thus, at least when considering only \emph{local} evolution in retarded time,
a smallness parameter $\epsilon$ can be retrieved  for this term
by integrating equation $(\ref{moreprecisely0})$
in view of 
$(\ref{fromheresmall})$,\footnote{One can draw a comparison at this point with the reductive structure in~\cite{formationofbh}.} while for 
 the non-borderline terms, we have
\begin{equation}
\label{nonborderlinet}
\int_{v}^{\infty}{O_{p+\frac32}} dv \lesssim \epsilon O_{p},
\end{equation}
where $\epsilon\to 0$ as $r(v)\to \infty$.

One can view the special structure of the non-linear interactions just discussed as representing
a ``null condition'' at the
level of the null structure equations, which (contingent also on the considerations
for $\psi$ to which the null structure equations are coupled--see Section~\ref{Herethebianchi} below!) in principle permits propagation of the
$p$-hierarchy for $\Gamma$ (cf.~(ii) above), at least locally in retarded time.
In the more difficult context of our global estimates, which require
solving for infinite retarded time, we will see how this is done
in practice in the discussion of Section~\ref{THEMAIN}.

\paragraph{The Bianchi equations for $\psi$.}\label{Herethebianchi}
 To propagate the $p$-hierarchy for $\psi$, we must also 
identify a ``null condition'' at the level of the Bianchi equations.

We first elaborate $(\ref{introbianchi})$    with regards to our $p$-hierarchy
 by rewriting the equations:
\begin{align}
\label{moreprecbian1}
&\slashed{\nabla}_3\uppsi_p =  \slashed{\mathcal{D}}\uppsi'_p +
O_p,\\
\label{moreprecbian2}
&\slashed{\nabla}_4\uppsi'_{p'}+\gamma_4(\uppsi'_{p'})tr\chi \uppsi'_{p'}= \slashed{\mathcal{D}}\uppsi_p + O_{p'+\frac32}
\end{align}

With respect to the $p$-decay of the last terms, we see a similar structure 
to that of $(\ref{moreprecisely0})$--$(\ref{moreprecisely1})$ for $\Gamma$. 
Now, however, the equations
are to be estimated not as transport equations, but 
 with \emph{weighted} energy estimates proven
by multiplying  each pair $(\ref{moreprecbian1})$, $(\ref{moreprecbian2})$,
respectively, by $r^q \uppsi_p$, $r^{q}\uppsi'_{p'}$,
and then integrating by parts, for a well-chosen 
weight $q(\uppsi_p)$.

The choice of the weight $r^q$ serves
so as to eliminate the contribution of the
$\gamma_4(\uppsi'_{p'})tr\chi \uppsi'_{p'}$ term in the divergence identity, which would 
be borderline with respect to decay (cf.~the weight factor $r^{2c[\overset {(4)}{\Gamma}_p]} $ 
in $(\ref{moreprecisely1})$). 
The significance of the presence of $O_{p'+\frac32}$ on the right
hand side of $(\ref{moreprecbian2})$, as opposed to $O_p$ on the right hand side of $(\ref{moreprecbian1})$, enters because
 the $\uppsi_p$ terms appear in energy fluxes through 
outgoing null cones, while the $\uppsi'_{p'}$ appear in energy fluxes
through ingoing null cones. Thus, extra decay is required in $(\ref{moreprecbian2})$
upon multiplication by $\uppsi'_{p'}$,
to ensure integrability as $r\to \infty$ for terms where this direction is not represented
by a flux.
In contrast to the case of $(\ref{moreprecisely1})$, there are no ``borderline'' terms
arising from this procedure,  as the decay of the extra
terms in $(\ref{moreprecbian2})$ is strictly greater that $p'+1$.
Concretely,
in the context of our global estimates, 
one can view the
final manifestation of this ``null condition'' as 
represented by the estimate  $(\ref{grontype})$ in Section~\ref{THEMAIN}.

\paragraph{Commutation.}\label{compreserves}
As discussed already in Section~\ref{MOSTBASICSCHEM}, one must derive higher order estimates
so as to close via Sobolev inequalities. This, however, means that
suitable decay must then 
also be captured for 
higher order quantities.
As will be shown in Section~\ref{sec:commute},
the $p$-hierarchical structure of the schematic form of the equations $(\ref{moreprecisely0})$--$(\ref{moreprecisely1})$
and $(\ref{moreprecbian1})$--$(\ref{moreprecbian2})$
is preserved under
arbitrary commutations with respect to a suitable set of differential operators $\mathfrak{D}^k$.
(These operators include tangential operators to the spheres which are not however
the usual ``angular momentum operators'' but 
in fact raise the type of tensors. See in particular the discussion of 
footnote~\ref{footnotecom}.)
This will indeed  allow for higher order estimates for $\mathfrak{D}^k\Gamma$ and 
$\mathfrak{D}^k\psi$.

\subsubsection{The blue-shift at $\mathcal{H}^+$ and the necessity of exponential decay}
\label{whyexp}
Having discussed the analysis at null infinity $\mathcal{I}^+$, we turn to 
considerations regarding the horizon
$\mathcal{H}^+$.
When solving the Einstein equations $(\ref{vacEq})$ \emph{backwards}, we immediately
meet a fundamental obstacle: {\bf the celebrated
red-shift effect on $\mathcal{H}^+$ is now seen as a \emph{blue-shift effect}}:
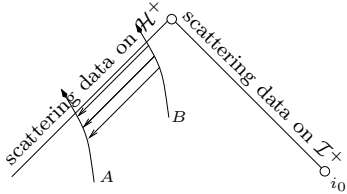
\[
\input{intro7.pstex_t}
\]
This effect should be understood in the geometric optics approximation  as follows:
If two observers $A$ and $B$ cross the horizon $\mathcal{H}^+$ as depicted, the frequency of a signal that $B$ sends to $A$ \emph{backwards} in time will be received by $A$ exponentially blue-shifted (in the difference in retarded horizon-crossing time between the two observers). 

For stationary black holes, the exponential factor is determined by the so-called \emph{surface gravity} $\kappa$ of the horizon~\cite{mihalisnotes}. In the Kerr
case, this is given explicitly by
\[
\kappa=\frac{\sqrt{M^2-a^2}}{2M^2+2M\sqrt{M^2-a^2}}.
\]
Note that in the Schwarzschild case $\kappa=(2M)^{-1}$, while for  $M$ fixed,
$\kappa$ decreases as $|a|\to M$, 
vanishing in the extremal case $|a|=M$.

As the red-shift/blue-shift is an effect of geometric optics, it is also present in the
 context of the linear wave equation $(\ref{scalarwave})$ 
 on a fixed Schwarzschild background, where everything can be made
 very concrete, in view
 of the existence of a $\partial_t$-energy scattering theory due
 to Dimock and Kay~\cite{Dimock1, Dimock2}.\footnote{There is a large
 literature concerning scattering theory for $(\ref{scalarwave})$ on black holes.
 We mention also the monograph~\cite{Matzner} the works~\cite{Bachelot3, Nicolas2, Haefner2}
 and the very recent~\cite{baskin}.}
Recall from~\cite{Dimock1} that for scattering data
$\psi|_{\mathcal{H}^+}$ and $r\psi|_{\mathcal{I}^+}$ 
assumed only to be of finite $\partial_t$-energy,
one can associate a unique solution $\psi$ of $(\ref{scalarwave})$ in the domain of outer
communications of Schwarzschild, realising the scattering data,
such that $\psi$ has finite 
$\partial_t$-energy on $t=0$.
 It can then be explicitly shown that, in accordance with the above geometric optics effect,
\emph{generic} smooth scattering data  $\psi|_{\mathcal{H}^+}$ and  $r\psi|_{\mathcal{I}^+}$ 
decaying slower than exponential will lead to a
solution $\psi$ which, though smooth in the black hole exterior,
will fail to be regular on the event horizon.\footnote{Recall that
the $\partial_t$ energy does not control transversal derivatives
to $\mathcal{H}^+$.}
Thus, to ensure regularity assuming only  decay assumptions for
  $\psi|_{\mathcal{H}^+}$ and  $r\psi|_{\mathcal{I}^+}$, one must 
impose  that this decay is
exponential.

For the more complicated Einstein vacuum equations $(\ref{vacEq})$, one sees the role
of the blue-shift  in our schematic equations
in the \emph{sign} of certain $f_{p_1}$ terms in equations
$(\ref{moreprecisely0})$ and $(\ref{moreprecbian1})$, as these drive exponential growth
of certain $\Gamma$ and $\psi$ when solving
backwards, if the final parameters are to satisfy $|a|<M$. 
In view of the above remarks, then similarly with the case
of the wave equation $(\ref{scalarwave})$, for $(\ref{vacEq})$ one is 
again led naturally
to the imposition of exponential decay.\footnote{In fact, to ensure each order
of higher regularity at the horizon one is led to impose a faster exponential decay
rate. 
This is because the strength of the red-shift is more and more enhanced 
at each order of commutation for the 
commuted equations described above in Section~\ref{compreserves}.
See Section~\ref{sec:redshiftrole} and Remark~\ref{rem:Pdep}.}
The challenge is then to show that indeed (issues of asymptotics towards $\mathcal{I}^+$--just
discussed in Section~\ref{NSBE}--aside)
the quantities $\Gamma$ and $\psi$ grow \emph{at most} exponentially when solving backwards. 
This is precisely what will be achieved in our global estimates outlined in  Section~\ref{THEMAIN}.

We  defer a discussion of the possibility of constructing singular
solutions of  $(\ref{vacEq})$, like those of $(\ref{scalarwave})$ discussed immediately above, to
Section~\ref{sec:weaknull}. Finally, we will discuss
in 
 Section~\ref{instabconj}
 the significance of these
remarks for the relation of the spacetimes which we construct here
with generic solutions of the forward problem,
for which, as we shall see, exponential decay is \emph{not} expected to hold.

Let us point out explicitly that in the extremal case $|a|=M$, in view of the degeneration
of $\kappa$, one is not in fact ``forced'' to impose exponential decay, and this case will play
thus a special role in the discussion of~Sections~\ref{sec:weaknull} and~\ref{instabconj}.
In the present work, however, we will not attempt to exploit this but
will make  uniform assumptions on data
for all $|a|\le M$.

\subsubsection{The global estimates}
\label{THEMAIN}
Having given a preview of all the essential aspects of the analysis, let us now return to 
the proof proper and describe
how these elements
enter into the main global estimates.

Recall the setup of Section~\ref{thesetuphere} for the associated finite approximate problem.
In addition to the double null foliation discussed above, we will foliate the ambient
manifold (which we now denote by $\mathcal{M}(\tau_0,\tau,v_{\infty})$)  on which our solution will be defined 
 by (what will be) spacelike
hypersurfaces
$\Sigma_{\tau_*}$, $\tau_0\le \tau_*\le \tau$,
terminating on the event horizon $\mathcal{H}^+$ and the finite approximation
$v=v_{\infty}$ to null infinity $\mathcal{I}^+$.
For each such  $\tau_*$, we may also consider the
subregion $\mathcal{M}(\tau_*,\tau,v_{\infty})$, depicted as the darker
shaded region below.
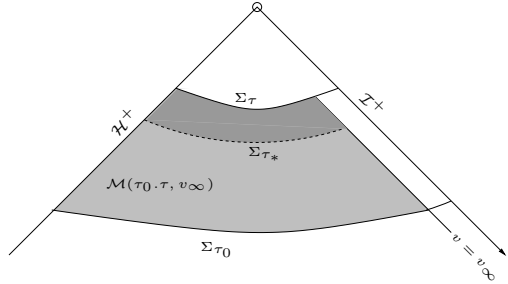
\[
\input{locally24.pstex_t}
\]
As is usual for non-linear problems, existence of the solution and estimates up to $\Sigma_{\tau_0}$
must be proven simultaneously, by a continuity argument
in $\tau_*$. A further complication
arises from the fact that 
for existence,  we appeal to general well-posedness
theory for $(\ref{vacEq})$ in the smooth category, whereas
we estimate solutions using $(\ref{introforgamma})$--$(\ref{introbianchi})$ 
in the gauge described in Section~\ref{renor}.
A  framework for 
dealing with this has been given in complete detail in~\cite{formationofbh}.
See our treatment in Section~\ref{sec:logic}. 
For the present discussion, 
let us for convenience assume that we are \emph{given} the existence of a spacetime in the region
in question and concentrate only on the issue of proving global \emph{estimates}
independent of $\tau$ and $v=v_\infty$ (as $\tau\to \infty ,v_\infty\to\infty$).

\paragraph{The case of the linear scalar wave equation $\Box\psi=0$.}
\label{warmupcase}
As a warm-up, let us first see how one obtains uniform estimates for solutions $\psi$
of the linear scalar wave equation 
$(\ref{scalarwave})$  in the  region $\mathcal{M}(\tau_0,\tau,v_{\infty})$, 
with analogously prescribed (approximated)
scattering data, and where $g$ is now the fixed Kerr metric. Namely, let 
 us prescribe
trivial data $\psi=0$, $\nabla\psi=0$ on $\Sigma_\tau$, 
exponentially decaying (in view again of the considerations of Section~\ref{whyexp}) data $\psi|_{\mathcal{H}^+}$ 
on $\mathcal{H}^+$, cut off to vanish in the future
of $\Sigma_\tau$, and exponentially decaying approximate
scattering data $r\psi|_{v=v_\infty}$ on $v=v_\infty$, again cut off.  
 
As we shall see, the argument is extremely straightforward. 
One defines an energy 
\begin{equation}
\label{rweightshere}
\mathcal{E}[\psi](\tau_*)=\int_{\Sigma_{\tau_*}} (1-2M/r)^{-1} r^{-h} |\partial_u \psi|^2+r^2|\partial_v\psi|^2+r^2|\slashed\nabla\psi|^2,
\end{equation}
which is non-degenerate at the horizon and incorporates positive $r$-weights at null infinity related to the expected decay
hierarchy 
\begin{equation}
\label{hierarc}
|\partial_u\psi| \lesssim r^{-1}, \qquad |\partial_v\psi| \lesssim r^{-2},
\qquad  |\slashed\nabla\psi| \lesssim r^{-2}.
\end{equation}
(The negative $r$-weight $r^{-h}$ arises
from the geometry of $\Sigma_\tau$ itself, for $1<h< 2$, from the choice $(\ref{hereh!})$.
The integral is to be taken with respect to the induced volume form.)
We can then derive an energy identity in the region $\mathcal{M}(\tau_*,\tau,v_\infty)$
of the form
\begin{align}
\label{thebulkterforwave}
\nonumber
\mathcal{E}[\psi] (\tau_*)+&
\int_{\mathcal{M}(\tau_*,\tau,v_\infty)}
 {\rm Bulk\  term}\\
&=  {F}_{\mathcal{H}^+\cap J^+(\Sigma_{\tau_*})}[\psi] 
 +{F}_{{v=v_\infty}\cap J^+(\Sigma_{\tau_*})}[\psi]+\mathcal{E}[\psi] (\tau),
\end{align}
where ${F}$ denote the flux terms corresponding to the above energy quantity.

Our assumptions on (approximate) scattering data are
 \begin{equation}
 \label{ofthefluxes}
\mathcal{E}[\psi] (\tau)=0, \qquad
 {F}_{\mathcal{H}^+\cap J^+(\Sigma_{\tau_*})}[\psi] 
 +{F}_{{v=v_\infty}\cap J^+(\Sigma_{\tau_*})}[\psi]\lesssim e^{-P\tau_*}.
\end{equation}
It follows that $(\ref{thebulkterforwave})$ yields
\begin{align}
\label{thebulkterforwave2}
\mathcal{E}[\psi] (\tau_*)+
\int_{\mathcal{M}(\tau_*,\tau,v_\infty)}
 {\rm Bulk\  term}
\lesssim e^{-P\tau_*}.
\end{align}
On the other hand, one can show  (see~\cite{DafRodInf})
\begin{equation}
\label{BASIKO}
\int_{\mathcal{M}(\tau_*,\tau,v_\infty)}
 |{\rm Bulk\  term}|
 \lesssim \int_{\tau_*}^\tau \mathcal{E}[\psi] (\tau') d\tau'.
 \end{equation}
 Were both integrals restricted to a unformly bounded $r$-range $r\le R$, the  inequality $(\ref{BASIKO})$
 would be essentially trivial,
 requiring only the fact that the background $g$ (Kerr in our case)
 admits a time-translation invariant
timelike vectorfield. In view, however, of the fact that 
$\sup_{\mathcal{M}(\tau_*,\tau,v_\infty)} r\to\infty$
as $v\to v_\infty$,
the fact that no positive $r$ weights appear on the right hand side of
$(\ref{BASIKO})$ is
nontrivial, and reflects the relation of the weights of $(\ref{rweightshere})$
with the hierarchy $(\ref{hierarc})$, the fact that Kerr is 
 asymptotically flat,
the properties of the geometry of $\Sigma_\tau$, and the existence
a timelike Killing field near infinity.\footnote{In fact, the existence
of time-translation invariance and Killing field can be weakened
to the existence of vector fields whose deformation tensor 
is uniformly bounded in a suitable $r$-weighted sense.}

Thus,  $(\ref{thebulkterforwave2})$ yields 
\begin{align}
\label{thebulkterforwave3}
\mathcal{E}[\psi] (\tau_*)\lesssim
\int_{\tau_*}^\tau \mathcal{E}[\psi] (\tau') d\tau'
+ e^{-P\tau_*},
\end{align}
which by Gronwall's inequality
gives
\begin{equation}
\label{whatitgives}
\mathcal{E}[\psi] (\tau_*)\lesssim e^{-P\tau_*},
\end{equation}
provided that  $P$ is chosen larger than
the implicit constant in the $\lesssim$ symbol of $(\ref{thebulkterforwave3})$. 

We may see concretely the role of the blue-shift effect discussed in Section~\ref{whyexp}
in  the identity $(\ref{thebulkterforwave})$: If $|a|<M$ then the
\emph{sign} of the bulk term 
in a neighbourhood of the horizon $\mathcal{H}^+$
is  \emph{negative} and comparable
to the integrand of $(\ref{thebulkterforwave})$, with constant
of proportionality related to the surface gravity (see~\cite{mihalisnotes}).
$P$ is thus in particular constrained
by the strength of the surface gravity.\footnote{Let us note that
the sign of the bulk term is also necessarily negative near $\mathcal{I}^+$ due
to the weights. The precise analysis of~\cite{DafRodInf} shows, however, that there is hierarchial structure
relating the $r$-weights of boundary and bulk terms. This can be used to
show that considerations near $\mathcal{I}^+$ only constrain the energy $\mathcal{E}$ 
to grow polynomially.
We shall not pursue this further here.}

Let us note explicitly that our argument above appeals
neither 
to a conserved non-negative energy nor to a
Morawetz-type estimate (integrated local energy decay).
We see thus that the assumption of exponential decay of the scattering-data fluxes
$(\ref{ofthefluxes})$
has
absolved us of the arduous task of understanding either superradiance or the structure of trapped null-geodesics, so fundamental in establishing decay for the forward problem (see Section 4.1 of \cite{mihalisnotes}). Thus, the Schwarzschild and Kerr problems are at the exact
same level of difficulty, and in
particular, the extremal limit $|a| \to M$, where the problems of superradiance and trapping are strongly coupled (see Section 1.5 of \cite{dafrodlargea}) and moreover the redshift degenerates, is entirely unproblematic here, indeed, in the extremal case $|a|=M$
one can obtain in principle a better estimate\footnote{For this, however, one would indeed have to understand the  issue of 
trapping, etc.
Note that in the extremal case, the solutions $\psi$ of this finite problem 
constructed above (and
thus also the limiting solution with the data imposed on $\mathcal{H}^+$ and 
$\mathcal{I}^+$) manifestly have vanishing Aretakis constants (see~\cite{Aretakis2}).} 
than $(\ref{thebulkterforwave3})$.

For the linear equation $(\ref{scalarwave})$,
there is in fact considerable flexibility as to the
$r$-weights  in $(\ref{rweightshere})$. 
 For instance, we could have removed completely
 the $r^2$ weights from $\mathcal{E}$ and still would obtain $(\ref{BASIKO})$
 (which would now indeed be essentially trivial).
 More generally, we could have replaced $r^2$ with $r^p$ for $0\le p\le 2$ 
 (see again~\cite{DafRodInf}).
 This would lead of course to less precise uniform estimates
 on the solutions, but in the limit $v_\infty\to \infty$ would nonetheless
 still allow us to give some meaning to 
 a solution of the scattering problem.
 
If, on the other hand, we were to pass to  the analogous problem replacing 
now the linear equation $(\ref{scalarwave})$ with the
 the semilinear wave equation $(\ref{thesemil})$, then simply
to obtain the analogue of $(\ref{BASIKO})$, including the correct
$r$-weights in $(\ref{rweightshere})$ is now essential, and moreover, obtaining
 $(\ref{BASIKO})$ will depend on the validity of a null condition on $Q$.
In anticipation of our
argument, the reader may wish at this point to work out the analogue of our theorem
for the case of $(\ref{thesemil})$ for a $Q$
satisfying an appropriate null condition. 
Here, let us however proceed directly to $(\ref{vacEq})$.

\paragraph{Energy estimates for curvature $\psi$.}\label{EEfCu}
We  return thus to the vacuum equations $(\ref{vacEq})$ and the problem at hand.

We define an energy
$\mathcal{E}[\mathfrak{D}^3\Psi] (\tau_*)$
through a $\Sigma_{\tau_*}$ leaf (see the related notation in Section~\ref{sec:normse}), 
which is a sum of energy-type quantities containing 
each  renormalised Bianchi
component, commuted up to order $3$, i.e.~$\mathfrak{D}^k\psi$ for $|k|\le 3$,
incorporating the weights $r^q(\uppsi)$ as discussed in Section~\ref{NSBE}:
\begin{equation}
\label{herethedefi}
\mathcal{E}[\mathfrak{D}^3\Psi] (\tau_*)=\int_{\Sigma_{\tau_*}} \sum_{\psi, |k|\le 3}( r^{2q}
{\rm\ or\ }r^{2q-h} )
 |\mathfrak{D}^k
\psi|^2.
\end{equation}
Some components will naturally have $r^{-h}$ weights due to the geometry of
$\Sigma_\tau$ (cf.~$(\ref{rweightshere})$).
As discussed already in Section~\ref{MOSTBASICSCHEM}, 
considering a higher order energy is necessary for the estimates to close,
and the precise numerology is imposed by the structure of the  
 nonlinear terms  and 
the Sobolev inequality used.

By the procedure described in Section~\ref{NSBE} using (the $\mathfrak{D}$-commuted
version of) formulas 
$(\ref{moreprecbian1})$--$(\ref{moreprecbian2})$, 
we will derive an energy identity for the  quantity $\mathcal{E}$ in the region 
$\mathcal{M}(\tau_*,\tau,v_\infty)$, for each $\tau_0\le \tau_*\le\tau$,
 from which we derive a relation,
\begin{align}
\label{thebulkter}
\nonumber
\mathcal{E}[\mathcal{D}^3\Psi] (\tau_*)+&
\int_{\mathcal{M}(\tau_*,\tau,v_\infty)}
 {\rm Bulk\  term}\\
 &\le  {F}_{\mathcal{H}^+\cap J^+(\Sigma_{\tau_*})}[\mathfrak{D}^3\Psi] 
 +{F}_{{v=v_\infty}\cap J^+(\Sigma_{\tau_*})}[\mathfrak{D}^3\Psi]+\mathcal{E}[\mathfrak{D}^3\Psi] (\tau),
\end{align}
where the quantities on the right hand side are flux terms
completely determined by scattering data.
In view of the set-up of our approximate problem and the
exponential approach assumption on the data discussed in   Section~\ref{whyexp}, we have 
for the last term on the right hand side of $(\ref{thebulkter})$
\[
\mathcal{E}[\mathfrak{D}^3\Psi] (\tau)=0,
\]
while for the first two terms we have
\[
{F}_{\mathcal{H}^+\cap J^+(\Sigma_{\tau_*})}[\mathfrak{D}^3\Psi] \lesssim e^{-P\tau_*},
\qquad {F}_{{v=v_\infty}\cap J^+(\Sigma_{\tau_*})}[\mathfrak{D}^3\Psi] \lesssim e^{-P\tau_*}.
\]
Compare with $(\ref{ofthefluxes})$. (Here it is important that our 
$\mathcal{E}(\mathfrak{D}^3\Psi)$, $F(\mathfrak{D}^3\Psi)$ are constructed from  the renormalised quantities 
$\psi$  with the Kerr contribution subtracted out.\footnote{We note that the loss of derivatives
in the characteristic initial value problem, discussed in footnote~\ref{alosshere}, arises
when relating these fluxes to tangential
derivatives to the null hypersurfaces $\mathcal{H}^+$ and $v=v_\infty$.})

We thus have in analogy with $(\ref{thebulkterforwave2})$:
\begin{equation}
\label{thebulkter2}
\mathcal{E}[\mathfrak{D}^3\Psi] (\tau_*)
 \lesssim e^{-P\tau_*} +
\int_{\mathcal{M}(\tau_*,\tau,v_\infty)}
 |{\rm Bulk\  term}|.
\end{equation}

Let us denote by ${F}_v[\mathfrak{D}^3\Psi](\tau_*)$ the associated flux of the above
energy identity on the constant-$v$ hypersurface
intersected with $J^+(\Sigma_{\tau_*})\cap \mathcal{M}(\tau_*,\tau,v_\infty)$,
and by ${F}_u[\mathfrak{D}^3\Psi](\tau_*)$ the analogous quantity on a constant-$u$ 
hypersurface (cf.~the notation of Section~\ref{sec:normse}). 
We obtain a similar estimate to $(\ref{thebulkter2})$
bounding $F_v$ and $F_u$, and the three can be combined to yield
\begin{equation}
\label{thebulkter3}
\mathcal{E}[\mathfrak{D}^3\Psi] (\tau_*)+
F_u[\mathfrak{D}^3\Psi](\tau_*)+F_v[\mathfrak{D}^3\Psi](\tau_*)
 \lesssim e^{-P\tau_*} +
\int_{\mathcal{M}(\tau_*,\tau,v_\infty)}
 |{\rm Bulk\  term}|.
\end{equation}
It turns out that having bounded above
also the null flux terms 
will be useful in Section~\ref{HGFGa} below.

The bulk terms  on the right hand side 
of $(\ref{thebulkter3})$ are cubic and higher in their combined dependence on  background terms $f$,
renormalised curvature
$\mathfrak{D}^3\psi$ and renormalised Ricci coefficients $\mathfrak{D}^3\Gamma$:
\begin{equation}
\label{Bterm}
 {\rm Bulk\  term}=\sum_{k+l\le 3}
r^q f\cdot\mathfrak{D}^{k}\psi\cdot \mathfrak{D}^{l}\psi +
   r^qf\cdot\mathfrak{D}^{k}\Gamma \cdot \mathfrak{D}^{l}\psi
+
r^q\mathfrak{D}^{k}\Gamma\cdot \mathfrak{D}^{l_1}\psi\cdot \mathfrak{D}^{l_2}\psi +
 \cdots.
\end{equation}
Recall how these terms arise from the right hand side of $(\ref{moreprecbian1})$--$(\ref{moreprecbian2})$
after multiplication by $r^q\uppsi$, $r^q\uppsi'$, everything having been appropriately commuted by $\mathfrak{D}$.
Note, however, that the first two terms displayed above are only quadratic in the  quantities
$(\mathfrak{D}^k\psi, \mathfrak{D}^k\Gamma)$. 
As in our discussion of the linear wave equation $(\ref{scalarwave})$ in Section~\ref{warmupcase},
some of the constituents of the first term will be negative (and proportional
to the surface gravity), and this  reflects
concretely the 
blue-shift effect of Section~\ref{whyexp}.
The second term
arises from the
presence of the ``non-homogeneous terms'' $f\Gamma$
in renormalised Bianchi $(\ref{introbianchi})$.

We can now see
the significance of the null condition discussed briefly in Section~\ref{Herethebianchi}.
The appropriate choice of weights defining the energy
in $(\ref{herethedefi})$ allows us to estimate:
\begin{align}
\nonumber
\label{grontype}
\int_{\mathcal{M}(\tau_0,\tau,v_\infty)}
 |{\rm Bulk\  term} |\lesssim &
   \int_{\tau_0}^{\tau} \mathcal{E}[\mathfrak{D}^3\Psi](\tau_*) d\tau_*\\
   \nonumber
   &+\int_{\mathcal{M}(\tau_0,\tau,v_\infty)} f_2 r^{2p-2} \|\mathfrak{D}^{k}\Gamma_p \|_{L^2(u,v)}^2\\
  &+
\left( \sup_{\Gamma ,u,v} r^{p}\|\mathfrak{D}^{1} \Gamma_p\|_{L^\infty(u,v)}+\cdots\right)
  \int_{\tau_0}^{\tau} \mathcal{E}[\mathfrak{D}^3\Psi](\tau_*) d\tau_*.
\end{align}
The first two lines of the right hand side above  arise from the first two terms on the
right hand side of $(\ref{Bterm})$. Note that the first line is 
 comparable to the estimate $(\ref{BASIKO})$ for the linear
wave equation $(\ref{scalarwave})$.
The third line (familiar from the analogous energy
estimates for quasilinear wave equations)
arises
from the cubic terms in $(\ref{Bterm})$. We have omitted  from the prefactor
other higher $L^q$ norms (involving also  $\psi$)
which can be treated similarly to the one we have included.
The remarkable point to notice is that all weights are consistent
with our $p$-hierarchy.

The weighted $L^\infty$ norms in the last term on the right hand side of $(\ref{grontype})$
above can be controlled
\begin{equation}
\label{sobolevineq}
r^p\|\mathfrak{D}^1\Gamma_p\|_{L^\infty (u,v)}
\lesssim \sum_{0\le k\le 3}
r^{p-1}\|\mathfrak{D}^{k}{\Gamma}_p\|_{L^2(u,v)},
\end{equation}
by  applying
suitable Sobolev 
inequalities
 on the $2$-spheres of constant $(u,v)$.
Thus, to close, we must couple the estimate $(\ref{thebulkter})$ with estimates
that allow for control of $r^{p-1}\|\mathfrak{D}^{k}\Gamma_p\|_{L^2(u,v)}$ 
for $0\le k\le 3$.
Specifically,  we shall show in Section~\ref{HGFGa} below that
 \begin{equation}
 \label{givenithere}
r^{p-1}\| \mathfrak{D}^k\Gamma_p\|_{L^2(u,v)} (\tau_*) \lesssim e^{-P\tau_*/2}.
\end{equation}
Here, $\tau_*=\tau_*(u,v)$ denotes the $\tau_*$
value such that $(u,v)\in \Sigma_{\tau_*}$.  
 
It is clear that, \emph{given $(\ref{givenithere})$}, then by Gronwall's inequality,
the estimate
$(\ref{grontype})$ (together with the Sobolev inequality $(\ref{sobolevineq})$)
immediately yields\footnote{In view of the omitted terms
in $(\ref{grontype})$, the actual story 
is slightly more complicated, in that we will need $L^\infty$ estimates
for curvature $\psi$.  This will require introducing the inequalities $(\ref{ourgoalhere})$
as boostrap assumptions and subsequently \emph{improving} the relevant constants by the
estimate described. Including $(\ref{ourgoalhere})$ as boostrap assumptions
 will also be necessary because control
of the Sobolev constants themselves require some basic geometric
input. See Section~\ref{HGFGa} below where we shall explicitly
introduce a bootstrap assumption in the context of obtaining $(\ref{givenithere})$ and Definition~\ref{def:Adef} for the precise bootstrap setup.
\label{talkingboot}}
\begin{equation}
\label{ourgoalhere}
\mathcal{E}[\mathfrak{D}^k\Psi] (\tau_*) \lesssim e^{-P\tau_*}, \qquad
F_{u,v}[\mathfrak{D}^k\Psi] (\tau_*) \lesssim e^{-P\tau_*},
\end{equation}
in analogy with $(\ref{whatitgives})$, for $P$ sufficiently 
large.\footnote{In view of the comments after $(\ref{Bterm})$, one sees that $P$ is
in particular constrained by the surface gravity of the horizon.}
The estimates $(\ref{givenithere})$ and $(\ref{ourgoalhere})$ (which
can easily be extended to all higher order) would then together represent
our desired uniform estimates in the region $\mathcal{M}(\tau,\tau_0,v_\infty)$.

We turn finally to obtaining $(\ref{givenithere})$.

\paragraph{Global transport estimates for $\Gamma$.}\label{HGFGa}
As discussed already, to obtain 
$(\ref{givenithere})$ we will
integrate the ($\mathfrak{D}$-commuted version of the) transport equations
$(\ref{moreprecisely0})$--$(\ref{moreprecisely1})$, as described in  
Section~\ref{NSEFGa},
along constant-$u$ and constant-$v$ hypersurfaces.

In view of our exponential decay assumption, our data can be taken to satisfy
\begin{equation}
\label{Gdata1}
r^{p}\|\mathfrak{D}^{k}\Gamma|_{\mathcal{H}^+}\|_{L^2(\infty,v)}\lesssim e^{-P\tau_*/2}, \qquad
r^{p}\|\mathfrak{D}^{k}\Gamma|_{v=v_\infty} \|_{L^2(u,v_\infty)} \lesssim e^{-P\tau_*/2},
\end{equation}
\begin{equation}
\label{Gdata2}
r^{p}\|\mathfrak{D}^{k}\Gamma|_{\Sigma_\tau}\|_{L^2(u,v)}=0,
\end{equation}
consistent with $(\ref{givenithere})$. Moreover, 
by a continuity argument in $\tau_*$, we can in fact \emph{assume} the inequality
$(\ref{givenithere})$ itself as a bootstrap assumption, provided that our
estimates improve the relevant constants (see footnote~\ref{talkingboot}  
and  Definition~\ref{def:Adef}).
In particular, given $(\ref{givenithere})$, we indeed
 have by Section~\ref{EEfCu}, the estimates $(\ref{ourgoalhere})$.

We now integrate  $(\ref{moreprecisely0})$--$(\ref{moreprecisely1})$,
 backwards in time \emph{after integration in $L^2$ over the constant-$(u,v)$ spheres},
starting at the 
hypersurfaces $\mathcal{H}^+$, $v=v_\infty$, $\Sigma_{\tau}$ where data are defined.
The contributions of data are then
bounded precisely by $(\ref{Gdata1})$, $(\ref{Gdata2})$.

Let us first briefly discuss the 
integrals  of the $\psi$ terms
which appear linearly
 on the right hand side of $(\ref{introbianchi})$.
 Having first integrated in $L^2$ of the spheres, upon integration in the $u$
and $v$ directions these terms
can be estimated
from the null curvature fluxes $F_v$ and $F_u$ (which have in
turn been estimated in $(\ref{ourgoalhere})$!), after applying
Cauchy-Schwarz. 
Thus, at least from the point of view of regularity, the estimates can in principle
close.

From the point of view of $r^p$-weights,
we have already discussed in Section~\ref{Herethebianchi}
 the structure that in principle   will allow
a smallness factor to arise when evolving locally in retarded time $u$. 
In our global context, instead of shortness of $u$-interval (exploited to yield the
$\epsilon$ of $(\ref{fromheresmall})$),
we now exploit the exponential factor $e^{-P\tau/2}$ of our bootstrap
assumption $(\ref{givenithere})$.

Without giving the details of the estimates here, let us note simply that
integration in the $u$ direction will now
produce factors
\begin{equation}
\label{integral1}
 \sup_{\theta,\phi}\left|\int_{u}^\infty e^{-P\tau_*} \Omega^2_{\mathcal{EF}} (u_*,v,\theta,\phi) du_*\right|
\end{equation}
whereas integration 
in the $v$
direction will produce factors
\begin{equation}
\label{integral2}
\int_{v}^\infty e^{-P\tau_*} r^{-p-1} dv_*.
\end{equation}
Here $\Omega^2_{\mathcal{EF}}$ is the conformal factor of the metric in our gauge
with respect to an Eddington-Finkelstein normalised null frame (see $(\ref{maing2})$).

Functions  $u$, $v$, $\tau$, and $r$ are
 all fixed to our differential structure
 and satisfy
\[
dv \lesssim r^h d\tau
\]
whereas, 
for $r\ge R$ we have 
\[
\partial_vr \le \frac12.
\]
Moreover,
our bootstrap assumptions allows us to show that the conformal factor
of the metric $\Omega^2_{\mathcal{EF}}$
satisfies
\[
\Omega^2_{\mathcal{EF}} du_*\sim d\tau.
\]
(This behaviour can be understood already from the exactly Schwarzschild case.)

From these relations, one sees immediately that  $(\ref{integral1})$ yields  a factor
\begin{equation}
\label{epshere}
\epsilon e^{-P\tau_*(u,v)},
\end{equation}
whereas $(\ref{integral2})$, using also $(\ref{epshere})$ for the borderline terms, yields
\begin{equation}
\label{epsthere}
\epsilon r^{-p}e^{-P\tau_*(u, v)},
\end{equation}
where
 the $\epsilon$ arises from choice of a suitably large $P$.
The $\epsilon$ factors in $(\ref{epshere})$ and $(\ref{epsthere})$ are the global
analogue of those in the local
naive computations $(\ref{fromheresmall})$ and $(\ref{nonborderlinet})$
and allow one to improve 
 the bootstrap assumption  $(\ref{givenithere})$, as desired.
 See Section~\ref{sec:ricciimprove} for details.

This concludes our discussion of the proof of global uniform
estimates for the approximate finite
problem.

\subsubsection{Differences of solutions and convergence}
\label{difrences}
Since our main theorem does not assert uniqueness (recall the discussion of
Section~\ref{thesetuphere}), one can  infer the existence 
of a solution of the limiting problem (with data at $\mathcal{H}^+$ and $\mathcal{I}^+$)
from the above estimates for the finite problem simply
by taking a subsequential limit
via Arzela-Ascoli (cf.~\cite{formationofbh}).

For future applications, we prefer to understand the convergence more quantitatively.
For this, one must estimate differences of two solutions $g$, $g^\dagger$. 
We 
shall consider 
\[
{\boldsymbol \Gamma}=\Gamma-\Gamma^\dagger,\qquad
{\boldsymbol\psi}=\psi-\psi^\dagger
\]
and derive
a system of equations for these quantities analogous to  the null structure and Bianchi equations.
It turns out that the general structure 
as described in Sections~\ref{MOSTBASICSCHEM}--\ref{NSBE}
 is again reflected in this system.\footnote{\label{footdif}We note only one additional feature:
We will separate out a subcollection ${\bf G}$ of the commuted
${\bf \Gamma}$, for which we shall
appeal to the additional
 elliptic null structure equations (cf.~footnote~\ref{longfoo}) to improve their regularity. 
 We note that we will 
 also appeal to this extra structure
 in the context of the estimates of Section~\ref{EEfCu} in the Kerr case $a\ne 0$
 (see~Section~\ref{sec:thekerrcase}).
We stress that in neither  cases is our appeal to this elliptic structure truly fundamental for the
argument--we do this simply to avoid applying
an additional commutation beyond the $\mathfrak{D}^3$ required
so as for $(\ref{grontype})$ to close.}
See Section~\ref{sec:convergence}.
The resulting estimates then 
show that our limit converges strongly. Moreover, these equations can in fact be used
to assert more generally the uniqueness of our solution in the class of
solutions a priori assumed to 
settle down to the Kerr family at a suitably fast exponential rate.

\subsection{Outline of the paper}
Having given an overview of the main ideas of our proof, we give here a very brief 
outline of the structure of the paper for the convenience of the reader.

In Section~\ref{sec:setting}, we fix the ambient manifold on which both our finite approximations
(as a subset) and our final spacetime will be defined. 
We then introduce the class of metrics to be considered, in appropriate gauge,
and give the renormalised Bianchi and null structure equations in this context.

We then proceed in Section~\ref{sec:equations} to discuss our systematic reformulation of the renormalised equations. Furthermore, the set of commutation operators
$\mathfrak{D}$ is defined which 
allows commuting the equations an arbitrary number of times while preserving the fundamental structure required for our estimates.

Section~\ref{sec:norms} will introduce the basic norms which will be relevant both in understanding
the conditions imposed on data, and for controlling the solutions.

Defining a notion of scattering data and associated data for an approximate finite problem is the content of Section \ref{sec:local}. This is completed by an appropriate well posedness statement, Theorem \ref{theolocal}.

The main theorems of the paper in the Schwarzschild case will then be given in Section~\ref{sec:maintheorem}. The main theorem is Theorem \ref{theo:full}.
Theorem \ref{theo1} is a statement of the uniform control of solutions to the ``finite" problem, while Theorem \ref{theo2} addresses the issue of convergence of the approximation procedure. The latter two will imply Theorem \ref{theo:full}.

Section \ref{sec:bootstrap} is devoted to the proof of Theorem \ref{theo1}, while
Theorem \ref{theo2} is proven in Section \ref{sec:convergence}.

The slight variation in the setup which is necessary to treat the Kerr case (including the case of extremality) is briefly discussed in Section \ref{sec:thekerrcase}.

The paper ends with a discussion of the Robinson-Trautman metrics and some formulae
which are collected in Appendix A.

\subsection{Future directions and other comments}
We end this introduction with some comments on open directions 
suggested by our results.

\subsubsection{The event horizon as a weak null singularity} \label{sec:weaknull}
We have discussed in Section~\ref{whyexp} the ``necessity'' of imposing exponential decay
on our scattering data on $\mathcal{H}^+\cup \mathcal{I}^+$,
in view of the constraints given by the horizon blueshift.
Let us comment here in more detail on what   happens if one indeed tries to solve the problem with slower decay assumed on the scattering data.

Recall that, according to Section~\ref{whyexp}, for solutions of the linear scalar
wave equation~$(\ref{scalarwave})$ on a fixed Schwarzschild background, 
one sees that generic polynomially decaying scattering  data lead to 
a solution regular in the black hole exterior but 
singular  on the horizon. 
For a general non-linear wave equation, this would suggest
that it would be simply impossible to provide any meaningful solution to 
the backwards problem, as the nonlinearities could propagate the 
singular behaviour from the horizon to the exterior, invalidating any existence theory.
Experience from the vacuum equations $(\ref{vacEq})$, on the other hand,
suggests that they may exhibit precisely that special 
structure\footnote{The recent~\cite{LukRod, LukRod3} may be especially relevant for this.} necessary
to preserve the localised singular behaviour of the linear wave equation $(\ref{scalarwave})$.
Motivated by the latter, we thus conjecture:
\begin{conjecture} \label{conj:inst2}
For smooth scattering vacuum data
as in the main theorem but now assumed
to settle down to a Kerr solution on $\mathcal{H}^+$ and $\mathcal{I}^+$
only at a (suitably fast) inverse polynomial rate, there again exists a vacuum
spacetime $(\mathcal{M},g)$ ``bounded by'' $\mathcal{H}^+$ and $\mathcal{I}^+$,
attaining the data,
regular away from $\mathcal{H}^+$. However,
for generic such data with asymptotic parameters $|a|<M$,  the
Christoffel symbols of the resulting metric 
(specifically, for instance, the ingoing null shear $\underline{\hat{\chi}}$)
fail to be locally  square integrable  near the horizon.
\end{conjecture}

The statement that the solution attains the data implies in particular that the metric and various tangential
derivatives thereof extend continuously to the boundary.
The horizon $\mathcal{H}^+$ would then correspond precisely to a \emph{weak null singularity},
analogous to phenomena well known from the context of black hole 
interiors~\cite{Dafermosinterior,Dafermosnospacelike} (cf.~Section~\ref{BHI} below).\footnote{The
name ``weak'' null singularity is traditional, but this singularity is in fact quite strong, in particular,
the Einstein equations $(\ref{vacEq})$ would \emph{not} be satisfied at $\mathcal{H}^+$ 
in the weak sense.} See also~\cite{Kozameh}.

The restriction to the subextremal case is related precisely to the degeneration of the surface
gravity when $|a|=M$. In the extremal case, the spacetimes of the above conjecture
may be in fact regular.

If Conjecture~\ref{conj:inst2} is indeed true, it would furthermore be interesting to understand the threshold governing the rate of polynomial decay sufficient for solving backwards.

\subsubsection{Comparison with the forward problem} \label{sec:instabcon}
\label{instabconj}

Conjecture~\ref{conj:inst2}    can be interpreted as the statement that our result--in its general outlines--is sharp, in the sense that,
if one were to impose generic weaker-than-exponentially decaying 
scattering data on $\mathcal{H}^+$ and
$\mathcal{I}^+$, then this scattering data will \emph{not} arise from regular Cauchy data.

This might at first suggest that the generic perturbations of Kerr initial data (as in
the statement of the non-linear black hole stability conjecture)
should give rise to spacetimes which exhibit precisely the type of decay we impose
on $\mathcal{H}^+\cup\mathcal{I}^+$.
The situation, however, turns out to be far more complicated!

\emph{Even for the linear scalar wave equation $(\ref{scalarwave})$},
on asymptotically flat spacetimes with non-zero mass (including, thus,
the Schwarzschild and Kerr case), generic Cauchy  data  are expected
to lead to polynomial tails on both $\mathcal{H}^+$ and $\mathcal{I}^+$,
no matter how smooth and localised the Cauchy data are imposed to be. This phenomenon originates in the scattering of low frequencies.\footnote{See Section 4.6 of \cite{mihalisnotes} for further discussion and also the more recent \cite{tohaneanu, DonnSchlag}, the latter showing that a $t^{-3}$ decay rate is indeed optimal on Schwarzschild for compactly supported data.}

On the other hand, turning to the simplest model nonlinear problem, that of
a general semi-linear wave equation $(\ref{thesemil})$ on Minkowski space, again
solutions of generic Cauchy data 
can be expected to develop polynomial tails on $\mathcal{I}^+$.
 See for instance \cite{BizonYM}. 
 
 Thus,  generic perturbations
 of Kerr initial data on a Cauchy hypersurface  can be expected
 to   have non-trivial contributions to their asymptotics generated from \emph{both} the linearisation and the non-linearities\footnote{For the collapse of the self-gravitating scalar field in spherical symmetry, one can show upper bounds on decay rates up to (within $\epsilon$) the obstructions given by the linear theory \cite{DafRod}. In the absense of symmetry, it is unclear whether the linear tails retain their relevance because they may be dwarfed by tails generated by the non-linearity. For instance, only much slower polynomial decay is known in the context of \cite{ChristKlei}.} of the Einstein equations $(\ref{vacEq})$.
Whatever the precise behaviour may be, this already  motivates:
\begin{conjecture} \label{conj:inst3}
For
 ``generic'' smooth Cauchy data as in the statement of the Non-linear stability of Kerr Conjecture,
the resulting vacuum spacetime
(and thus its scattering data on $\mathcal{I}^+\cup\mathcal{H}^+$) will settle
down to the Kerr family  only at an inverse polynomial rate and no better. Thus, the entire family of
spacetimes constructed in our main theorem
arise from a non-generic set of initial data for the Cauchy problem.
\end{conjecture}

It is a difficult problem
to identify and understand any non-trivial 
sufficient \emph{compatibility} assumption on polynomially decay on $\mathcal{I}^+$
and $\mathcal{H}^+$ scattering data
that would allow for solving backwards to obtain  Cauchy data  regular
at the
horizon. This problem is interesting even in the
 case of $(\ref{scalarwave})$ on a fixed Schwarzschild background.

\subsubsection{Wave operators and reference dynamics}
In this paper, we view ``scattering theory'' simply as the map
from asymptotic data on $\mathcal{H}^+ \cup \mathcal{I}^+$ to solution of $(\ref{vacEq})$.
In particular, we never refer explictly to a global \emph{comparison dynamics}
with an auxiliary reference solution.

For the reader wishing to make that connection, 
let us consider for convenience the scattering problem for
 the spherically symmetric linear wave equation $\Box_g\psi=0$ on a Schwarzschild
 background. We may rewrite this equation
in terms of $r\psi$ as
 \begin{equation}
 \label{rewriteit}
-(\partial_{t}^2+ \partial_{r_*}^2)(r\psi)= V\cdot r\psi
\end{equation}
where $r^*$ is the Regge--Wheeler coordinate taking values in $(-\infty,\infty)$
and $V(r^*)$ is a  potential. We may define comparison dynamics 
associated to
\begin{equation}
\label{Free}
-(\partial_{t}^2+ \partial_{r_*}^2)(r\widetilde\psi)=0,
\end{equation}
and wave operators $W_\pm$, etc., by conjugation of $(\ref{rewriteit})$
with the evolution of $(\ref{Free})$.
The analogue of our scattering map (i.e.~the map from data
on $\mathcal{H}^+\cup\mathcal{I}^+$ to solution of $(\ref{rewriteit})$
restricted to $t=0$ say)
 can then be understood in the standard way in the language
of these wave operators.

It may be useful to remark that in this language, writing the scattering map
as a conjugation typically requires understanding \emph{decay} for the reference comparison
dynamics
in solving forward 
while only a \emph{boundedness} statement for the actual dynamics  in solving 
backwards.\footnote{More generally, one requires that solutions of the 
actual problem grow suitably more slowly than the decay of the reference problem.}
It is the latter
which is in general unavailable 
in the black hole context, in view of the blue-shift effect (unless the whole theory
can be formulated with respect to a degenerate energy, as in the
$\partial_t$-scattering theory~\cite{Dimock1, Dimock2} for $(\ref{scalarwave})$ on 
Schwarzschild\footnote{This theory is analogous to the usual scattering 
setting where ``boundedness'' is an immediate consequence of unitarity.}, 
discussed in Section~\ref{whyexp}). From this point of view,
the imposition of exponential decay on scattering data allows one to relax
the requirement of boundedness above with that of at most exponential growth. One sees thus
that this  assumption
can in principle be used to construct solutions in many 
other scattering problems for which boundedness  statements
(either forward or backward)
are not available, for instance $(\ref{scalarwave})$ with a time-dependent
potential.

\subsubsection{Uniqueness} \label{sec:uniq}
As explained already in Section~\ref{thesetuphere}, our theorem does \emph{not} assert the uniqueness of the vacuum spacetime $\left(\mathcal{M},g\right)$ realising the given scattering data at $\mathcal{H}^+ \cup \mathcal{I}^+$, \emph{even in the case of trivial data}, for which our theorem of course yields simply Kerr as solution. 
The method of proof of Theorem \ref{theo2} can be used to prove uniqueness in the class of solutions which are assumed a priori to 
exponentially settle down to Kerr in
the domain of outer communications. It would already be interesting to obtain uniqueness if it is only assumed that $\left(\mathcal{M},g\right)$ is uniformly close to Kerr.\footnote{Cf.~with the recent work \cite{Alexakis3} on the uniqueness of the Kerr family in the class of \emph{stationary} solutions under similar assumptions of uniform closeness.} This problem exhibits some of the difficulties of the stability problem but is in principle easier.

\subsubsection{Asymptotically flat hypersurfaces terminating at spatial infinity}
In the present paper, we have only done the analysis up to a spacelike hypersurface
$\Sigma_{\tau_0}$ intersecting null infinity $\mathcal{I}^+$,
but, provided scattering data along $\mathcal{I}^+$ are imposed with suitable fall-off towards spacelike infinity $i_0$, one can seek to obtain
a spacetime admitting an asymptotically flat Cauchy hypersurface:
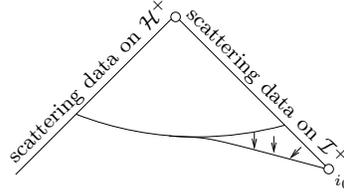
\[
\input{intro3.pstex_t}
\]
This would be interesting to do explicitly using the set-up developed here (cf.~the comments in Section 14 of \cite{Mihalisbourbaki}). See also \cite{Li}.

\subsubsection{The black hole interior}\label{BHI}
We have been referring to the spacetimes of our main Theorem throughout as
 ``black hole'' spacetimes. This is justified!
Indeed, one can attach an honest black hole region to our spacetimes by extending the data suitably through the event horizon $\mathcal{H}^+$
and solving a characteristic problem, as indicated in the figure below. By the results of \cite{luk2}, there indeed exists a solution in a neighborhood of the entire event horizon (the lightly shaded region). In view of \cite{Dafermosinterior}, one may conjecture that, \emph{provided the final parameters are not Schwarzschild},
i.e.~$a\ne 0$, the solution will exist in a small uniform strip (the dark shaded region below) and will admit as future boundary a Cauchy horizon $\mathcal{CH}^+$ emanating from $i^+$ through which the metric extends continuously to a larger spacetime but for which generically the Christoffel symbols will fail to be square integrable. This latter property is essential to ensure the validity of Christodoulou's formulation of strong cosmic censorship \cite{formationofbh}. However, a resolution of these issues will require techniques beyond those of the present paper.
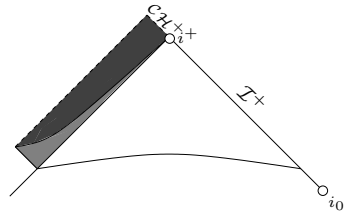
\[
\input{intro6.pstex_t}
\]

\subsubsection{$\Lambda\ne0$}
Finally, the problem posed in this paper can also be considered in the case 
of the Einstein equations with cosmological constant
$\Lambda\ne 0$:
\[
{\rm Ric}=\Lambda g
\]

In the case $\Lambda>0$, the corresponding problem would be to
construct  
spacetimes asymptotically settling down to Kerr-de Sitter in the region between
the cosmological horizon $\mathcal{C}^+$ and black hole horizon $\mathcal{H}^+$:
\[
\input{intro5.pstex_t}
\]
See \cite{Haefner, DafRoddS, Melrose, Dyatlov1, VasyDyatlov, sbzworski, Schlue} for the study of the wave equation (\ref{scalarwave}) on such backgrounds. Here, in principle, it is not necessary to capture a special ``null condition'' in the non-linearity of (\ref{vacEq}), and one can imagine an
alternative approach
\emph{based entirely in harmonic coordinates} or suitable modifications thereof as in \cite{Igor2, Ringstrom,SpeckRod}. See also \cite{FangWang}.

In the case $\Lambda<0$, then the question would  be to construct spacetimes
which asymptotically settle down to Schwarzschild- or Kerr-AdS.
\[
\input{intro4.pstex_t}
\]
See \cite{HawkingReall, HolzegelAdS, gs:decay, GHHCMW, Gannot, oberwolf,
newGustJac} for the study of the wave equation (\ref{scalarwave}) on such backgrounds. Recall that $\mathcal{I}^+$ is now timelike and the associated spacetimes are not globally hyperbolic.
For the usual boundary conditions imposed at null-infinity $\mathcal{I}^+$, 
 one expects that the free scattering data are prescribed only at the horizon $\mathcal{H}^+$. 
(See, however,~\cite{AndChruDel}.) 
This requires on the other hand a good understanding of the well-posedness issue for the full Einstein equations in the presence of a timelike boundary at infinity. See \cite{FriedrichAdS,gs:lwp}.

Note that in the asymptotically AdS case the potential contrast with generic solutions of the ``forward'' problem is even more stark than that described in Section~\ref{instabconj} above.
On Kerr--AdS, it has now been proven that
generic solutions of the wave equation decay logarithmically~\cite{gs:decay} 
and no better~\cite{oberwolf, Gannot, newGustJac}, suggesting that 
perhaps Kerr--AdS spacetimes are  unstable as solutions of the vacuum Einstein
equations~\cite{gs:decay}.\footnote{Compare with the case of pure AdS, where
solutions of the linear
wave equation do not decay, leading to the conjecture
in~\cite{eguchiha} that pure AdS is
 nonlinearly unstable (see also \cite{Anderson}). 
This instability has subsequently been discovered  numerically by Bizon and
Rostworowski~\cite{BizonAdS},
who have also understood better its heuristic basis. The black hole
case may turn out to be quite subtle; see~\cite{Marolfstab}.}
Thus, the method of the present paper could
provide an indispensable tool for constructing dynamical black hole spacetimes with negative cosmological constant.\footnote{One can draw an analogy here with our main theorem applied to the extremal case $|a|=M$ where the forward problem is subject to the Aretakis instability along the horizon \cite{Aretakis, Aretakis2, Reallextreme}.}

\subsection{Acknowledgements}
The authors thank Stefanos Aretakis and Martin Taylor for many comments on the 
manuscript.
M.D.~is supported in part by a grant from the European Research Council. G.H. and I.R. acknowledge support through NSF grants DMS-1161607 and DMS-0702270 respectively.
\section{The geometric setting} \label{sec:setting}
We henceforth  specialise our 
discussion to the Schwarzschild $a=0$ case of the main theorem, to be
stated precisely in Section~\ref{sec:maintheorem}.
We shall consider the general Kerr case
in Section~\ref{sec:thekerrcase}.

In this section, we will consider 
a \emph{given} spacetime in a particular double null gauge,
and our goal will be to derive the null structure and Bianchi
equations, discussed in Section~\ref{renor}. The assumptions on the underlying
manifold will be introduced in stages, begining in Section~\ref{sec:manifold} 
simply with its differential structure and associated coordinates and frames.
The differential
structure is chosen so as for the manifold
to admit naturally the Schwarzschild metric of mass $M$ (see Section~\ref{sec:ssv}).  
Once we impose (in Section~\ref{therebne}) the vacuum equations,
the existence of a single
example other than Schwarzschild
satisfying our assumptions will 
only become explicit
in the context of the proof of Theorem~\ref{theo1} in Section~\ref{sec:bootstrap}.  
Without first deriving
these equations, however, it would be impossible to understand the setup of initial data
in Section~\ref{sec:local}.

\subsection{The manifold} \label{sec:manifold}
Define the four-dimensional manifold with boundary
\[
\mathcal{M} = \mathcal{D} \times S^2 = \left(-\infty,0\right] \times \left(0, \infty\right) \times S^2  \, .
\]
We will denote a point in $\mathcal{M}$ by its coordinates $\left(U,V,\theta^1,\theta^2\right)$ with the implicit understanding that two coordinate charts are required on $S^2$. The boundary $\mathcal{H}^+:=\{0\} \times \left(0,\infty\right) \times S^2$ of $\mathcal{M}$ will, for reasons that will become apparent later, be called the \emph{horizon}. 

Fix a constant $M>0$. This constant will correspond to the mass of an auxiliary Schwarzschild metric on $\mathcal{M}$ to be defined in Section \ref{sec:ssv} (and
will represent finally in Theorem~\ref{theo:full} the mass of the Schwarzschild
metric to which our spacetime will settle down to). Besides the $\left(U,V\right)$ coordinate system for the base space $\mathcal{D}$, which we shall call \emph{Kruskal coordinates}, we define so-called \emph{Eddington-Finkelstein coordinates} $\left(u,v\right)$ related to $\left(U,V\right)$ by the transformations
\begin{equation} \label{EFK}
U = - e^{-\frac{u}{2M}} \textrm{ \ \ \ \ and  \  \ \ \ \ } V =  e^{\frac{v}{2M}} \, .
\end{equation}
Eddington-Finkelstein coordinates cover $\mathcal{M}$ except for its boundary $\mathcal{H}^+$ which formally corresponds to $u = \infty$. We denote the spheres at fixed $\left(U,V\right)$ by $S^2_{U,V}$ (or $S^2_{u,v}$). 

Let the function $r : \mathcal{M} \rightarrow \left[2M,\infty\right) \subset \mathbb{R}^+$ be defined implicitly as the solution of
\begin{align} \label{rdef}
e^{\frac{v-u}{2M}} = \frac{1-\frac{2M}{r}}{\frac{2M}{r} e^{-\frac{r}{2M}}}   \, .
\end{align}
Note that the left hand side is manifestly non-negative and that $r=2M$ holds on the horizon. Defining $x= \frac{r}{2M}$, we see that the right hand side is equal to $e^x \left(x-1\right)$, which is strictly monotonically increasing on $[1,\infty)$. Hence $r$ is well-defined and it is also seen to be smooth on $\mathcal{M}$. Besides the function $r$, we define a smooth function $t : \mathcal{M} \setminus \partial \mathcal{M} \rightarrow \mathbb{R}$ by 
\[
t\left(u,v,\theta^1,\theta^2\right) = u+v \, 
\]
and a smooth function $r^\star : \mathcal{M} \setminus \partial \mathcal{M}  \rightarrow \mathbb{R}$ by
\[
r^\star \left(u,v,\theta^1,\theta^2\right) = v-u \, .
\]
Note that fixing an $r>2M$ and a $t \in \mathbb{R}$ determines uniquely a pair $\left(u,v\right)$ and vice versa. Therefore, we could also use $\left(t,r,\theta^1, \theta^2\right)$ or equivalently $\left(t,r^\star,\theta^1, \theta^2\right)$ as coordinates for the interior of $\mathcal{M}$, as there is a simple smooth invertible relation between $r^\star$ and $r$ following from (\ref{rdef}).
\begin{figure}
\[
\input{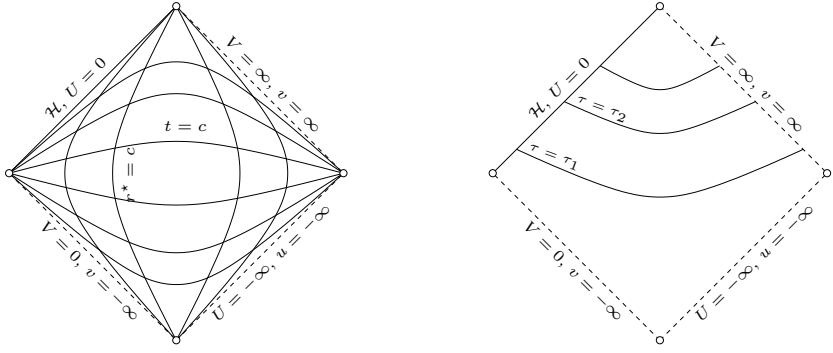}
\]
\caption{The manifold $\mathcal{D}$} \label{figure}
\end{figure}

With this in mind, we define a piecewise smooth function $\tau\left(u,v\right)$ 
as follows \cite{dafrodlargea}. Interpreting $r$ as a function of $r^\star$ we fix some $1< h \leq 2$ and define
\begin{equation}
\label{hereh!}
f\left(r^\star\right) = \int_0^{r^\star} \sqrt{1-\frac{(2M)^h}{(r\left(\tilde{r}^\star\right))^h}\left(1-\frac{2M}{r\left(\tilde{r}^\star\right)}\right)} d\tilde{r}^\star \, .
\end{equation}
Next, setting $D =  f\left(r^\star\left(r=9/4M\right)\right) +  f\left(r^\star\left(r=8M\right)\right)$, we define
\begin{equation} \label{fpiece}
f_{piece} \left(r^\star\right)  = \left\{
\begin{array}{ll}
f\left(r^\star\left(u,v\right)\right)  &\text{\ \ \ if } r\left(u,v\right) < 9/4 M ,\\
f\left(r^\star\left(r=9/4M\right)\right)  &\text{\ \ \ if } 9/4M \leq r\left(u,v\right) \leq 8M,\\
-f\left(r^\star\left(u,v\right)\right) + D &\text{\ \ \ if } r\left(u,v\right) > 8M.
\end{array} \right.
\end{equation}
The function $f_{piece}$ is piecewise smooth on $\mathcal{M} \setminus \partial \mathcal{M}$ and can be turned into a smooth function by mollifying it in a small neighbourhood of $r^\star\left(r=\frac{9}{4}M\right)$ and $r^\star\left(r=8M\right)$ respectively. We denote the function thus obtained by $f_{smooth}$.\footnote{The introduction of $f_{smooth}$ is for the aesthetic convenience of dealing with smooth hypersurfaces. The only property of the mollification required later (once a metric has been introduced) is that it preserves (up to a small constant) global bounds on the norm of the gradient of $\tau=u+v+f_{piece}$ when comparing with the gradient of $\tau=u+v+f_{smooth}$. See Lemma \ref{sec:slices}.} Finally,
\begin{equation*}
\tau \left(u,v\right) := u+v + f_{smooth} \left(r^\star\left(u,v\right)\right) \, .
\end{equation*}
The function $\tau$ extends smoothly to a monotonically increasing function along the horizon $\mathcal{H}^+$ and we therefore consider $\tau$ as a function on all of $\mathcal{M}$. The hypersurfaces of constant $\tau$ define a family of $3$-dimensional connected smooth submanifolds (with boundary being a $2$-sphere) on $\mathcal{M}$, which we will call $\Sigma_\tau$. In view of these considerations, we obtain an additional coordinate system on $\mathcal{M}$: the $\left(\tau, r, \theta^1, \theta^2\right)$-coordinates. We define the region
\[
\mathcal{M} \left(\tau_1, \tau_2\right) := \bigcup_{\tau_1 \leq \tau \leq \tau_2} \Sigma_\tau \, .
\]
For later purposes it will be convenient to introduce also the region
\begin{equation} \label{finregion}
\mathcal{M} \left(\tau_1, \tau_2,v_\infty\right) = \mathcal{M} \left(\tau_1, \tau_2\right) \cap \{v \leq v_\infty \}
\end{equation}
for a given $v_\infty \gg \tau_2$. Associated with such a region are the quantities
\begin{align} \label{ufutdef}
{u}_{fut}\left(v\right) := \sup_{\left(\tilde{u},v\right) \in \mathcal{M} \left(\tau_1,\tau_2, v_\infty\right)} \tilde{u} \qquad , \qquad {v}_{fut} \left(u\right):= \sup_{\left(u,\tilde{v}\right)  \in \mathcal{M} \left(\tau_1,\tau_2, v_\infty\right)} \tilde{v} \, .
\end{align}
Note that $u_{fut}=\infty$ if the constant $v$-hypersurface intersects the horizon, and $v_{fut}=v_\infty$ if the constant $u$-hypersurface reaches $v=v_\infty$. We will sometimes use the shorthand notation $u_{fut}$ if the $v$-hypersurface involved is clear from the context and similarly $v_{fut}$ if the $u$-hypersurface is unambiguous.
\subsection{The class of metrics}
So far we have not introduced any metric on $\mathcal{M}$. We now proceed to equip the manifolds $\mathcal{M}\left(\tau_1,\tau_2,v_\infty\right)$ with smooth metrics of the form (in Kruskal coordinates)
\begin{align} \label{maing}
g = - 2 \Omega_{\mathcal{K}}^2 \left(dU \otimes dV + dV \otimes dU\right) + \slashed{g}_{AB} \left(d\theta^A - b_{\mathcal{K}}^A dV\right) \otimes \left(d\theta^B - b_{\mathcal{K}}^B dV\right)  ,
\end{align}
where $A,B=1,2$ and
\begin{equation}
\begin{split}
&\Omega_{\mathcal{K}} : \mathcal{M}\left(\tau_1,\tau_2,v_\infty\right) \rightarrow \mathbb{R}^+ \qquad \textrm{is a function,}  \\
&\slashed{g}_{AB} \left(U,V,\theta^1,\theta^2\right) d\theta^A d\theta^B \qquad \textrm{\ is a Riemannian metric on $S^2_{U,V}$,} \\
& b^A_{\mathcal{K}} \ \ \  \textrm{ is a vectorfield taking values in the tangent space of  $S^2_{U,V}$.} 
\end{split}
\end{equation}
In addition, we shall require $b^A_\mathcal{K}=0$ on the horizon $\{U=0\} \cap \mathcal{M}\left(\tau_1,\tau_2\right)$.\footnote{This condition can be viewed as a gauge condition, which is useful as it turns the constraint equations into ODEs, cf.~Section \ref{sec:local}.}  All quantities are assumed to be smooth. Defining $\Omega^2_{\mathcal{E}\mathcal{F}}$ and $b_{\mathcal{E}\mathcal{F}}$ via the relations
\begin{align} 
 \Omega_{\mathcal{K}}^2 = \Omega_{\mathcal{E}\mathcal{F}}^2 \frac{2M}{V} \frac{2M}{-U} \ \ \ \ , \ \ \ b_{\mathcal{E}\mathcal{F}}^A = b_{\mathcal{K}}^A \frac{V}{2M} \, , \nonumber
\end{align}
we have in Eddington-Finkelstein coordinates
\begin{align} \label{maing2}
g = - 2 \Omega_{\mathcal{E}\mathcal{F}}^2 \left(du \otimes dv + dv \otimes du\right) + \slashed{g}_{AB} \left(d\theta^A - b_{\mathcal{E}\mathcal{F}}^A dv\right) \otimes \left(d\theta^B - b_{\mathcal{E}\mathcal{F}}^B dv\right) .
\end{align}
Observe that hypersurfaces of constant $u$ and $v$ (or $U$ and $V$ respectively) are null hypersurfaces of the metric $g$. Note also that since the metric (\ref{maing}) is regular, then necessarily $\lim_{u \rightarrow \infty} \Omega^2_{\mathcal{E}\mathcal{F}} \left(u,v\right) = 0$ for all $v$. This illustrates the failure of the Eddington-Finkelstein coordinate system on the horizon.

As we will state explicitly in Section \ref{sec:local}, any metric arising in the context of a suitable initial value formulation for the Einstein vacuum equations can be expressed (a-priori only locally) in the form (\ref{maing}) by choosing a double null foliation as the coordinate gauge and Lie-transporting the angular coordinates fixed on a particular sphere on the horizon in the null-directions.  See \cite{formationofbh}.

For now, we continue by defining geometric objects such as the frame, the Ricci-coefficients and the curvature for a given $\left(\mathcal{M} \left(\tau_1, \tau_2,v_\infty\right),g\right)$ of the form (\ref{maing}). \newline

{\bf Convention:} Because the Eddington-Finkelstein coordinates are more convenient in computations, we agree on the following convention. We may drop the subscript ${\mathcal{E}\mathcal{F}}$ for $b_{\mathcal{E}\mathcal{F}}$ and $\Omega_{\mathcal{E}\mathcal{F}}$. Hence $b^A=b^A_{\mathcal{E}\mathcal{F}}$ and $\Omega^2 = \Omega^2_{\mathcal{E}\mathcal{F}}$. On the other hand, the Kruskal-quantities will always retain their subscript $\mathcal{K}$.

\subsection{Frames} \label{sec:frames}
We define the \emph{regular normalized Eddington Finkelstein (${\mathcal{E}\mathcal{F}}$) frame}:
\begin{align} \label{frefdef}
e_3^{\mathcal{E}\mathcal{F}} = \frac{1}{\Omega^2} \partial_u \textrm{ \ \ \ , \ \ \ } e_4^{\mathcal{E}\mathcal{F}}= \partial_v + b^A e_A \ \ \ , \ \ \ e_A = \partial_{\theta^A} \, .
\end{align}
One easily checks that $g\left(e_3^{\mathcal{E}\mathcal{F}},e_4^{\mathcal{E}\mathcal{F}}\right)=-2$, $g\left(e_A,e_B\right) = \slashed{g}_{AB}$ and that all other combinations are zero, Moreover, $\nabla_{e_3^{\mathcal{E}\mathcal{F}}} e_3^{\mathcal{E}\mathcal{F}}= 0$.
To see that this frame is indeed regular consider the normalized Kruskal frame
\begin{align}
\label{kruskfr}
e_3^{\mathcal{K}} = \frac{1}{\Omega_{\mathcal{K}}} \partial_U \textrm{ \ \ \ , \ \ \ } e_4^{\mathcal{K}} = \frac{1}{\Omega_{\mathcal{K}}} \left(\partial_V + b_\mathcal{K}^A e_A \right) \, .
\end{align}
We have the following relation between the two frames:
\begin{align} \label{framerel}
e_3^{\mathcal{K}} = \frac{\Omega_{\mathcal{K}} V}{2M} e^{\mathcal{E}\mathcal{F}}_3   \textrm{ \ \ \ \ ,  \ \ \ \ } e_4^{\mathcal{K}} = \frac{2M} {\Omega_{\mathcal{K}} V} e^{\mathcal{E}\mathcal{F}}_4 \, ,
\end{align}
which shows that indeed the above ${\mathcal{E}\mathcal{F}}$-frame is regular.

\begin{remark} \label{stupre}
The frame $\frac{1}{\Omega_{\mathcal{E}\mathcal{F}}} \partial_u$, $\frac{1}{\Omega_{\mathcal{E}\mathcal{F}}} \left(\partial_v +b_{\mathcal{E}\mathcal{F}}^A e_A\right)$ would be the ``natural" null-coordinate frame arising from Eddington Finkelstein coordinates; however, this frame is not regular on the horizon $\mathcal{H}^+$.
\end{remark}

With respect to the regular ${\mathcal{E}\mathcal{F}}$-frame we define the Ricci coefficients
\begin{align} \label{RicC}
\chi_{AB} &= g \left(\nabla_A e_4^{\mathcal{E}\mathcal{F}} ,e_B\right) \, , \qquad  \qquad \qquad \underline{\chi}_{AB} = g \left(\nabla_A e_3^{\mathcal{E}\mathcal{F}},e_B\right) \, ,
\nonumber \\
\eta_{A} &= -\frac{1}{2} g \left(\nabla_{e_3^{\mathcal{E}\mathcal{F}}} e_A,e_4^{\mathcal{E}\mathcal{F}}\right) \, , \qquad  \qquad \underline{\eta}_{A} = -\frac{1}{2} g \left(\nabla_{e_4^{\mathcal{E}\mathcal{F}}}e_A,{e_3^{\mathcal{E}\mathcal{F}}}\right) \, ,
\nonumber \\
\hat{\omega} &= \frac{1}{2} g \left(\nabla_{e_4^{\mathcal{E}\mathcal{F}}} {e_3^{\mathcal{E}\mathcal{F}}},e_4^{\mathcal{E}\mathcal{F}}\right) \, , \qquad \qquad \hat{\underline{\omega}} =  \frac{1}{2}g \left(\nabla_{e_3^{\mathcal{E}\mathcal{F}}} e_4^{\mathcal{E}\mathcal{F}},e_3^{\mathcal{E}\mathcal{F}}\right) = 0 \, ,  \nonumber \\
\zeta_A &= \frac{1}{2} g \left(\nabla_A e_4^{\mathcal{E}\mathcal{F}},e_3^{\mathcal{E}\mathcal{F}}\right) \, .
\end{align}
The following relations can be established using the properties of the frame:\footnote{In the (irregular) double-null frame of Remark \ref{stupre}, ($\frac{1}{\Omega}\partial_u$, $\frac{1}{\Omega}\left(\partial_v + b^A e_A\right)$), we would have the familiar $2\zeta = \eta - \underline{\eta}$ from \cite{formationofbh}.}
\begin{align}
\zeta = - \underline{\eta} \ \ \ \ \textrm{,} \ \ \  2 \slashed{\nabla} \left(\log \Omega \right) = \eta + \underline{\eta} \ \ \ \ \textrm{,} \ \ \ \hat{\omega} = \frac{e_4^{\mathcal{E}\mathcal{F}}\left( \Omega^2\right)}{\Omega^2}\, .
\end{align}
We also define the null-decomposed curvature components
\begin{align}
\label{curvdefins}
\alpha_{AB} &= R \left(e_A, e_4^{\mathcal{E}\mathcal{F}}, e_B, e_4^{\mathcal{E}\mathcal{F}}\right) \ , \qquad \ \  \qquad \  \underline{\alpha}_{AB} = R \left(e_A, e_3^{\mathcal{E}\mathcal{F}}, e_B, e_3^{\mathcal{E}\mathcal{F}}\right) \nonumber \\
2\beta_{A} &=  R \left(e_A, e_4^{\mathcal{E}\mathcal{F}}, e_3^{\mathcal{E}\mathcal{F}}, e_4^{\mathcal{E}\mathcal{F}}\right) \, , \qquad \qquad \ 2\underline{\beta}_{A} = R \left(e_A, e_3^{\mathcal{E}\mathcal{F}}, e_3^{\mathcal{E}\mathcal{F}}, e_4^{\mathcal{E}\mathcal{F}}\right) \nonumber \\
4\rho &=  R\left(e_4^{\mathcal{E}\mathcal{F}}, e_3^{\mathcal{E}\mathcal{F}},e_4^{\mathcal{E}\mathcal{F}}, e_3^{\mathcal{E}\mathcal{F}}\right) \, ,  \qquad \qquad  4\sigma = {}^\star R\left(e_4^{\mathcal{E}\mathcal{F}},e_3^{\mathcal{E}\mathcal{F}},e_4^{\mathcal{E}\mathcal{F}},e_3^{\mathcal{E}\mathcal{F}}\right)
\end{align}
with ${}^\star R$ denoting the Hodge dual of $R$. Cf.~\cite{ChristKlei,formationofbh}.

\subsection{The Schwarzschild values} \label{sec:ssv}
With the help of the spacetime function $r\left(u,v\right)$ defined in (\ref{rdef}), we can define a notion of  the \emph{Schwarzschild values} (indicated by a $\circ$) of the metric quantities:
\begin{align} \label{eq:exr}
\Omega^2_{\circ}\left(u,v\right) = 1- \frac{2M}{r\left(u,v\right)} \ \ \ , \ \ \
\slashed{g}^{\circ}_{AB} =  r^2 \left(u,v\right) \gamma_{AB} 
\ \ \ \ , \ \ \ 
b^\circ_A = 0 \, ,
\end{align}
where $\gamma_{AB} d\theta^A d\theta^B$ is the round metric on the unit sphere. The choice (\ref{eq:exr}) of $\Omega^2, \slashed{g}_{AB}$ and $b^A$ in (\ref{maing2}) defines an auxiliary Schwarzschild metric of mass $M$ on $\mathcal{M}\left(\tau_1,\tau_2\right)$.
 From the above one can compute
\begin{align}
\label{eq:ex0}
\rho_{\circ} = \frac{-2M}{r^3} \  \ \ \ \ , \ \ \ \hat{\omega}_{\circ} = \frac{2M}{r^2} \textrm{ \ \ \  \  , \ \ \ \ } \hat{\underline{\omega}}_{\circ} = 0 \, ,
\end{align}
\begin{align} \label{eq:exs}
tr \left( \chi_{\circ}\right) =  2 \frac{1-\frac{2M}{r}}{r} \textrm{ \ \ \ \ and \ \ \ \ \ } tr \  \underline{\chi}_{\circ} = - \frac{2}{r} \, ,
\end{align}
while all other Ricci-coefficients and curvature components vanish identically. Note that while for Schwarzschild the function $r$ indeed has the geometric interpretation as the area radius of the spheres $S^2_{u,v}$, this is not true in general, as $r$ was defined without reference to any metric.

Using the relation (\ref{framerel}), one can also compute all these quantities (\ref{eq:exr})--(\ref{eq:exs}) with respect to the Kruskal frame. Here we only note 
\begin{align} \label{Kom}
\left(\Omega^2_{\mathcal{K}}\right)_\circ = 4M^2 \frac{2M}{r} e^{\frac{-r}{2M}} \, .
\end{align}
\subsection{The renormalised Bianchi and null structure equations}
\label{therebne}
We are going to impose that $\left(\mathcal{M}\left(\tau_1,\tau_2,v_\infty\right),g\right)$ satisfies the vacuum Einstein equations (\ref{vacEq}). It follows that the Ricci-coefficients satisfy the so-called \emph{null structure equations} and that the Riemann curvature tensor satisfies the \emph{Bianchi equations} on $\left(\mathcal{M}\left(\tau_1,\tau_2,v_\infty\right),g\right)$. See \cite{formationofbh, KlaiNic}. These equations contain the full content of the Einstein equations and will be used to estimate solutions.\footnote{Note however that in the logic of the proof of Section \ref{sec:logic} we \emph{construct} solutions of (\ref{vacEq}) by appealing directly to the general existence theorems of \cite{ChoquetBruhat, Geroch, Rendall}, not by extracting a closed subsystem of the equations here.}

In the following subsections, we express the Bianchi and null-structure equations on $\left(\mathcal{M}\left(\tau_1,\tau_2,v_\infty\right),g\right)$ as equations for the null-decomposed quantities of Section \ref{sec:frames}. These quantities are interpreted as \emph{$S^2_{u,v}$-tensors} of rank $2$, $1$ and $0$ and the equations themselves will be expressed in terms of the \emph{$S^2_{u,v}$-projected covariant derivatives}: $\slashed{\nabla}_3 = \slashed{\nabla}_{\frac{1}{\Omega^2} \partial_u}$ and $\slashed{\nabla}_4 = \slashed{\nabla}_{\partial_v + b^A e_A}$ acting on these tensors.
 We refer the reader to \cite{ChristKlei, formationofbh} for detailed derivations of these formulae and the notational conventions but recall at least the following Lemma, which relates the projected covariant derivatives to the coordinate derivatives.
\begin{lemma} \label{teco}
Let $\xi_{A_1...A_N}$ be a covariant $S^2_{u,v}$-tensor of rank $N$.
We have
\[
\left(\slashed{\nabla}_3 \xi\right)_{A_1\ldots A_N}  = \frac{1}{\Omega^2}\partial_u \left(\xi_{A_1...A_N} \right) -\sum_{k=1}^N \left(\slashed{g}^{-1}\right)^{BC} \underline{\chi}_{A_k B} \xi_{A_1 \ldots A_{k-1} C A_{k+1} \ldots A_N}
\]
and
\begin{align}
\left(\slashed{\nabla}_4 \xi\right)_{A_1\ldots A_N} = \left[ \partial_v + b^C \partial_{\theta^C} \right] \left(\xi_{A_1\ldots A_N} \right) - \sum_{k=1}^N \left(\slashed{g}^{-1}\right)^{BC}{\chi}_{A_k B} \xi_{A_1\ldots A_{k-1} C A_{k+1}\ldots A_N} \nonumber \\
-\sum_{k=1}^N b^D \slashed{\Gamma}^{C}_{ \ DA_k}  \xi_{A_1 \ldots A_{k-1} C A_{k+1}\ldots A_N} \nonumber
\end{align}
\begin{align} \label{eq:tec}
\left(\slashed{\nabla}_B \xi\right)_{A_1...A_N} = \partial_{\theta^B} \left(\xi_{A_1\ldots A_N}\right) - \sum_{k=1}^N \slashed{\Gamma}^{C}_{ \ BA_k}  \xi_{A_1\ldots A_{k-1} C A_{k+1} \ldots A_N}
\end{align}
with $\slashed{\Gamma}^{C}_{ \ AB} = \frac{1}{2} \left(\slashed{g}^{-1}\right)^{CD} \left(\slashed{g}_{AD,B} + \slashed{g}_{DB,A} - \slashed{g}_{AB,D}\right)$ the Christoffel symbols of the metric $\slashed{g}$.
\end{lemma}

We also recall the following operations on $S^2_{u,v}$-tensors:
The left Hodge-dual
\[
{}^\star \xi_A = \slashed{\epsilon}_{AB} \xi^B  \ \ \ and \ \ \ {}^\star \xi_{AB} = \slashed{\epsilon}_{AC} \xi^C_{\phantom{C}B}
\]
with $\slashed{\epsilon}$ denoting the volume form associated with $\slashed{g}$. The symmetric traceless product of two $1$-covariant $S^2_{u,v}$-tensors
\[
\left(\xi \widehat{\otimes} \varphi\right)_{AB} = \xi_A \varphi_B + \xi_B \varphi_A - \slashed{g}_{AB} \left(\left(\slashed{g}^{-1}\right)^{CD}\xi_C \varphi_D\right) \, 
\]
and the (anti-symmetric) products
\[
\xi \wedge \varphi = \slashed{\epsilon}^{CD} \xi_C \varphi_D \ \ \ \ \ \textrm{or} \ \ \ \ \ \xi \wedge \varphi = \slashed{\epsilon}^{AB} \left(\slashed{g}^{-1} \right)^{CD} \xi_{AC} \varphi_{BD}
\]
for $1$-covariant and $2$-covariant $S^2_{u,v}$-tensors respectively. Finally, we let
\begin{align}
\left(\xi,\phi\right)_{\slashed{g}} = \slashed{g}^{A_1 B_1} ... \slashed{g}^{A_N B_N} \xi_{A_1...A_N} \varphi_{B_1 ... B_N} 
\end{align}
\begin{align} \label{gnorm}
\| \xi \|^2_{\slashed{g}} = \slashed{g}^{A_1 B_1} ... \slashed{g}^{A_N B_N} \xi_{A_1...A_N} \xi_{B_1 ... B_N} \, 
\end{align}
denote the $\slashed{g}$-inner-product and $\slashed{g}$-norm of $N$-covariant $S^2_{u,v}$-tensors, respectively. With $d\mu_{\slashed{g}}$ denoting the induced volume element of $\slashed{g}$ we define also
\begin{align} \label{l2gnorm}
\|\xi \|^q_{L^q_{\slashed{g}}\left(S\right)} :=  \int_{S^2\left(u,v\right)} \|\xi \|_{\slashed{g}}^q d\mu_{\slashed{g}} \, \qquad \textrm{for $q\geq 1$}.
\end{align}
The subscript $\slashed{g}$ shall be conveniently dropped if it is clear from the context that the norm is taken with respect to $\slashed{g}$, i.e.~we write $\|\xi\|=\|\xi\|_{\slashed{g}} $.
Finally, we recall \cite{formationofbh} the musical notation $\xi^{\sharp}$ (indicating raising an index with $\left(\slashed{g}^{-1}\right)$) and the angular derivative operators 
 $\slashed{\mathcal{D}}^\star_1 \left(\pm \rho,\sigma\right) = \mp \slashed{\nabla}_A \rho + \slashed{\epsilon}_{AB} \slashed{\nabla}^B \sigma$ (mapping a pair of scalars to a $1$-form) and $-2\slashed{\mathcal{D}}^\star_2 \xi = \slashed{\nabla}_A \xi_B + \slashed{\nabla}_B \xi_A - \left( \slashed{div} \ \xi \right) \slashed{g}_{AB}$ (mapping a $1$-form to a $2$-form).

As discussed in Section~\ref{renor}, we shall in fact express our equations with respect to 
``renormalised'' curvature components and Ricci coefficients, which are defined simply by subtracting the non-trivial Schwarzschild
values
$(\ref{eq:exr})$--$(\ref{eq:exs})$.
Thus our renormalised quantities are
\[
{\rm curvature \ components:}\qquad \alpha, \beta, \rho-\rho_\circ, \sigma, \underline\beta, \underline \alpha
\]
\[
{\rm Ricci\ coefficients:}\qquad
\hat\omega-\hat\omega_\circ, \underline\eta, \eta, tr\underline\chi-tr\underline\chi_\circ,
 tr\chi-tr\chi_\circ, \hat{\underline\chi},\hat\chi
 \]
 \[
 {\rm metric \ quantities:}\qquad
  \left(\frac{\Omega_\circ^2}{\Omega^2}-1\right),
 b^A, \slashed{g} - \slashed{g}^\circ.
\]
The resulting system of equations to be expressed in the subsections below will be homogeneous in the above quantities,
and Schwarzschild will correspond to the trivial solution.

\subsubsection{Curvature} \label{bisec}
Given (\ref{vacEq}), the null components $\alpha$ and $\underline{\alpha}$ are (in addition to being symmetric) trace free and the Bianchi equations take the following form:
\begin{align}
\slashed{\nabla}_3 \alpha + \frac{1}{2} tr \underline{\chi} \alpha + 2 \underline{\hat{\omega}} \alpha &= -2 \slashed{\mathcal{D}}_2^\star \beta - 3 \hat{\chi} \rho_\circ - 3 \hat{\chi} \left(\rho-\rho_\circ\right) - 3{}^\star \hat{\chi} \sigma  + \left(4\eta + \zeta\right) \hat{\otimes} \beta \nonumber \\
\slashed{\nabla}_4 \beta + 2 tr \chi \beta - \hat{\omega} \beta &= \slashed{div} \alpha + \left(\underline{\eta}^\sharp + 2 \zeta^\sharp\right) \cdot \alpha \nonumber \\
\slashed{\nabla}_3 \beta + tr \underline{\chi} \beta + \underline{\hat{\omega}} \beta &= \slashed{\mathcal{D}}_1^\star \left(-\rho, \sigma\right) + 3 \eta \rho_\circ + 3\eta \left(\rho-\rho_\circ\right) + 3{}^\star \eta \sigma  \boxed{+ 2\hat{\chi}^\sharp \cdot \underline{\beta}} \nonumber \\
\slashed{\nabla}_4(\rho-\rho_\circ)  + \frac{3}{2} tr \chi(\rho-\rho_\circ) &= \slashed{div} \beta + \left(2\underline{\eta} + \zeta, \beta \right) - \frac{1}{2} \left(\underline{\hat{\chi}}, \alpha\right) 
-\frac32\rho_\circ(tr\chi-tr\chi_\circ)
\nonumber \\
\slashed{\nabla}_4 \sigma + \frac{3}{2} tr \chi \sigma &= -\slashed{curl} \beta - \left(2\underline{\eta} + \zeta \right) \wedge \beta  + \frac{1}{2} \underline{\hat{\chi}} \wedge \alpha  
\nonumber 
\end{align}
\begin{align} 
\slashed{\nabla}_3 (\rho-\rho_\circ) + \frac{3}{2} tr \underline{\chi}(\rho-\rho_\circ) &= -\slashed{div} \underline{\beta} - \left(2\eta - \zeta, \underline{\beta}\right) \ \boxed{- \frac{1}{2} \left(\hat{\chi}, \underline{\alpha}\right)}-\frac32\rho_\circ(tr\underline\chi-tr\underline\chi_\circ)\nonumber
\\
&\qquad
-\frac{3}{2} tr \underline{\chi}_\circ \rho_\circ \left(1-\frac{\Omega_\circ^2}{\Omega^2}\right)
\nonumber \\
\slashed{\nabla}_3 \sigma + \frac{3}{2} tr \underline{\chi} \sigma &= -\slashed{curl} \underline{\beta} 
- \left(2\eta - \zeta\right) \wedge \underline{\beta} \ \boxed{ - \frac{1}{2} \hat{\chi} \wedge \underline{\alpha}} \nonumber \\
\slashed{\nabla}_4 \underline{\beta} + tr \chi  \underline{\beta} + \hat{\omega} \underline{\beta} &= \slashed{\mathcal{D}}_1^\star \left(\rho, \sigma\right) - 3 \underline{\eta} \rho_\circ - 3 \underline{\eta} \left(\rho-\rho_\circ\right) + 3{}^\star \underline{\eta} \sigma  + 2 \underline{\hat{\chi}}^\sharp \cdot \beta \nonumber \\
\slashed{\nabla}_3 \underline{\beta} + 2 tr \underline{\chi}  \underline{\beta} - \hat{\underline{\omega}} \underline{\beta} &= - \slashed{div} \underline{\alpha} - \left(\eta^\sharp - 2 \zeta^\sharp\right) \cdot \underline{\alpha} \nonumber \\
\slashed{\nabla}_4 \underline{\alpha} + \frac{1}{2} tr \chi \underline{\alpha} + 2 \hat{\omega} \underline{\alpha} &=  2 \slashed{\mathcal{D}}_2^\star \underline{\beta} - 3 \underline{\hat{\chi}} \rho_\circ - 3  \underline{\hat{\chi}} \left( \rho - \rho_\circ\right) + 3{}^\star \underline{\hat{\chi}} \sigma - \left(4 \underline{\eta} - \zeta\right) \hat{\otimes} \underline{\beta}\nonumber
\end{align}
Here $\xi \cdot \varphi$ denotes the contraction of a symmetric $S^2_{u,v}$-tensor with an $S^2_{u,v}$-vector.
\begin{remark}
The boxed terms are the most weakly decaying towards null-infinity and will be referred to below. Recall also that in the $\mathcal{E}\mathcal{F}$-frame (\ref{frefdef}) we have $\hat{\underline{\omega}} =0$.
\end{remark}

\subsubsection{Ricci-coefficients} \label{sec:rc}
The null-structure equations for the Ricci coefficients take the following form:
\begin{align} \label{3omh}
\frac{1}{2} \slashed{\nabla}_3 \left(\hat{\omega} - \hat{\omega}_\circ\right) &= 2 \left(\eta, \underline{\eta} \right) - | \underline{\eta} |^2 - \left(\rho - \rho_\circ\right) + \rho_\circ \left( \frac{\Omega_{\circ}^2}{\Omega^2} - 1 \right)  
\end{align}
\begin{align} \label{etaeq}
 \slashed{\nabla}_3 \underline{\eta} &= \frac{1}{2} tr \underline{\chi} \left(\eta - \underline{\eta}\right) + \underline{\hat{\chi}} \cdot  \left(\eta - \underline{\eta}\right) + \underline{\beta}\nonumber \\
 \slashed{\nabla}_4  \eta &= \frac{1}{2} tr \chi \left(\underline{\eta} -\eta\right) + \hat{\chi} \cdot \left(\underline{\eta} - \eta\right)  - \beta
 \end{align}
\begin{align} \label{trx3}
\slashed{\nabla}_3 \left(tr \underline{\chi} - tr \underline{\chi}_\circ \right)= - \frac{1}{2}  \left(tr \underline{\chi} + tr \underline{\chi}_\circ \right) \left(tr \underline{\chi} - tr \underline{\chi}_\circ \right) \nonumber \\
 + \frac{1}{2} tr \underline{\chi}_\circ tr \underline{\chi}_\circ  \left( \frac{\Omega_{\circ}^2}{\Omega^2} - 1 \right) - \|\underline{\hat{\chi}} \|^2
\end{align}
\begin{align} \label{trx4}
\slashed{\nabla}_4 \left( tr \chi- tr \chi_\circ \right) =  - \frac{1}{2} \left( tr \chi+ tr \chi_\circ \right)  \left( tr \chi - tr \chi_\circ \right) \phantom{XXXX} \nonumber \\
+\hat{\omega} \left( tr \chi - tr \chi_\circ \right) + tr \chi_\circ \left(\hat{\omega} - \hat{\omega}_\circ\right) 
-  \|{\hat{\chi}} \|^2 
\end{align}
\begin{align} \label{xhat}
\slashed{\nabla}_{3} \underline{\hat{\chi}} &=  - tr \underline{\chi} \ \underline{\hat{\chi}} - \underline{\alpha} \nonumber \\
\slashed{\nabla}_{4} \hat{\chi} &= -tr \chi \ \hat{\chi} + \hat{\omega} \hat{\chi} - \alpha 
\end{align}

\subsubsection{Metric quantities} \label{sec:mq}
The metric quantities satisfy:
\begin{align} \label{4cofa}
 \slashed{\nabla}_4  \left( \frac{\Omega_{\circ}^2}{\Omega^2} - 1 \right) &= - \left(\hat{\omega} - \hat{\omega}_\circ\right)  \frac{\Omega^2_{\circ}}{\Omega^2} =  \left( \hat{\omega}_\circ - \hat{\omega} \right) + \left(\hat{\omega} - \hat{\omega}_\circ\right) \left(1- \frac{\Omega^2_{\circ}}{\Omega^2} \right) 
\end{align}

\begin{align} \label{beq}
\frac{1}{\Omega^2} \partial_u \left(b^A\right) = 2 \left(\eta - \underline{\eta}\right)^A
\end{align}
\begin{align} \label{mei}
\slashed{\nabla}_3 \left(\slashed{g}_{AB} - \slashed{g}^\circ_{AB} \right) = \left(1-\frac{(\Omega_{\circ})^2}{\Omega^2}\right) tr \underline{\chi}_\circ \  \slashed{g}^\circ_{AB}  + \left( tr \underline{\chi} - tr \underline{\chi}_\circ \right)\slashed{g}^\circ_{AB} \nonumber \\
+ 2\hat{\underline{\chi}}_{AB} + \hat{\underline{\chi}}_A^{\phantom{X}C} \left( \slashed{g}^\circ_{CB} - \slashed{g}_{CB}\right) + \hat{\underline{\chi}}_B^{\phantom{X}C} \left( \slashed{g}^\circ_{CA} - \slashed{g}_{CA} \right) 
\end{align}

\subsubsection{The remaining equations}
There are additional null-structure equations, first and foremost elliptic equations on the spheres $S^2_{u,v}$ but also additional transport equations involving curvature. For completeness, we present them in the appendix. These equations will be used in the construction of the data in Section \ref{sec:local}. In the context of our global exponential decay estimates, however, we shall not make use of them, with the exception of (\ref{elliptic}), which we will employ when we estimate differences of solutions. This is not necessary but convenient because it will allow us to close the estimates for differences without having to commute the equations. See also Section \ref{auxquant}. 

\section{A systematic formulation of the equations} \label{sec:equations}
In this section we shall reformulate the above set of equations in a systematic way, so as to reveal and isolate their structure relevant in our present work. 
This reformulation was previewed in Sections~\ref{MOSTBASICSCHEM} and~\ref{NSBE}.
In addition, as 
discussed in Section~\ref{compreserves},
we will select a suitable set of commutation operators which \emph{preserves} this structure of the equations independently of the number of commutations performed.

\subsection{Preliminaries} \label{sec:Gamdef}
Let 
\begin{align} 
\Gamma = \Big\{ \hat{\omega} - \hat{\omega}_\circ,  \eta, \underline{\eta}, tr \underline{\chi} - tr \underline{\chi}_\circ,   tr \chi - tr \chi_\circ,  \hat{\chi} ,  \underline{\hat{\chi}} , \frac{1}{2M} \left[
  \frac{\Omega_{\circ}^2}{\Omega^2} - 1\right] ,  b^A, \slashed{g}_{AB} - \slashed{g}^\circ_{AB} \Big\} \nonumber
\end{align}
and note that all these quantities are zero in Schwarzschild. We decompose $\Gamma$ further into
\begin{align} \label{Gammadeco}
\Gamma_1 &= \overset {(3)}{\Gamma}_1 = \{  \underline{\hat{\chi}} , \slashed{g}_{AB} - \slashed{g}^\circ_{AB} \} \, , \nonumber \\
\Gamma_2 &= \overset {(3)}{\Gamma}_2 \bigcup \overset {(4)}{\Gamma}_2 = \Big\{ \underline{\eta},  tr \underline{\chi} - tr \underline{\chi}_\circ, b \Big\} \bigcup \Big\{  tr \chi - tr \chi_\circ, \hat{\chi},   \frac{1}{2M} \left[
  \frac{\Omega_{\circ}^2}{\Omega^2} - 1\right], \eta \Big\} \nonumber \, , \\ 
\Gamma_3 &= \overset {(3)}{\Gamma}_3 = \{ \hat{\omega} - \hat{\omega}_\circ\} \, .
\end{align}
With a slight abuse of notation, we will denote an arbitrary element of the set $\overset{(.)}\Gamma_p$ again by $\overset{(.)}\Gamma_p$.

The heuristic idea underlying this decomposition is that all elements of the set $\Gamma$ or, in our abuse of notation, all $\Gamma$'s, vanish if the metric were exactly Schwarzschild. Therefore, the $\Gamma$'s will decay in time on the spacetimes that we are about to construct. 
The subscript encodes each quantity's characteristic $r$-decay expected at null-infinity $\mathcal{I}^+$: This simply means that $r^p \|\Gamma_p\|$ will have a finite limit as $r \rightarrow \infty$ for fixed $u$. The superscript $(3)$ or $(4)$ indicates whether the $\Gamma_p$ under consideration satisfies an equation in the null-direction $e_3$ or the $e_4$ direction respectively.\footnote{Some $\Gamma_p$ obey equations in both null-directions, i.e.~strictly speaking the superscript indicates which evolution equation we will exploit in the analysis.}

The reason for the normalization constant $(2M)^{-1}$ for the quantity $\frac{\Omega_{\circ}^2}{\Omega^2} - 1$ has to do with the scaling properties of the various quantities.

Similarly, we introduce the following notation for the curvature components
\begin{align} \label{stabmindec}
\psi_1 = \{ \underline{\alpha} \}  \ \ , \ \ \psi_2 = \{ \underline{\beta} \} \ \ , \ \ \psi_3 = \{ \rho-\rho_\circ, \sigma \} \ \ , \ \ \psi_\frac{7}{2} = \{ \beta \} \ \ , \ \ \psi_4 = \{ \alpha \} \, ,
\end{align}
with the understanding that $r^p \|\psi_p\|$ is expected to have a finite trace (but possibly vanishing) on null-infinity. We call each of the four pairs 
\[
\textrm{$\left(\alpha,\beta\right)$ \ \ , \ \  $\left(\beta, \left(\rho-\rho_\circ,\sigma\right)\right)$ \ \ , \ \  $\left(\left(\rho-\rho_\circ,\sigma\right) , \underline{\beta}\right)$ \ \ , \ \  $\left(\underline{\beta}, \underline{\alpha}\right)$}
\]
 a \emph{Bianchi pair} and denote it by $\left(\uppsi_p,\uppsi^\prime_{p^\prime}\right)$ with $p$ and $p^\prime$ referring to (\ref{stabmindec}).
 
 \begin{remark}
For the spacetimes we construct one actually has that $r^4\|\beta\|$ and $r^5\|\alpha\|$ have finite limits at null-infinity, i.e.~peeling behavior \cite{Wald}. Nevertheless, we have chosen to propagate the weaker decay for $\beta$ and $\alpha$ in (\ref{stabmindec}). The origin of the freedom of which decay bounds to propagate will become clear in the weighted estimates of Section \ref{sec:curvimprove}. Compare with \cite{ChristKlei, Lydia}, and see also \cite{mgrome} for a discussion of the curvature asymptotics arising from physically realistic applications, where it is proven that post Newtonian moments provide a general obstruction to the full peeling behavior on null-infinity.
 \end{remark}

\begin{definition} \label{def:fp}
For $p\geq 0$, we denote by $f_p$ a smooth function depending only on $r$ and satisfying for all $k \in \mathbb{N}_0$ the uniform bound $r^{k+p} | \left(\partial_r \right)^k f_p |\leq C_k$, where $C_k$ depends only on $k$ and $M$. In addition, the weighted smooth $S^2_{u,v}$-tensor $f_p \slashed{g}^\circ_{AB}$ may also be denoted by $f_p$. Note $r^{k+p} \| \left(\partial_r \right)^k f_p\slashed{g}^\circ_{AB} \|_{\slashed{g}^\circ}\leq C_k$ holds.
\end{definition}
Typical examples are
\[
 tr \chi_\circ tr \chi_\circ + \hat{\omega}_\circ = f_2 \ \ \  \textrm{or} \ \ \ 
 \slashed{g}_{AB}^\circ = f_0 \ \ \ or \ \ \ 
 \slashed{g}_{AB}^\circ tr \chi_\circ = f_1 \ \ \  \textrm{or} \ \ \  \rho_\circ = f_3 \, .
 \]
\begin{definition}
We denote by $\Gamma_{p_1} \Gamma_{p_2}$ any finite sum of products of two $\Gamma$'s, with each product being a contraction (with respect to $\slashed{g}$) between two $S^2_{u,v}$-tensors $\Gamma_{p_1}$ and $\Gamma_{p_2}$.\footnote{In the estimates to be proven later we will only employ the formula $\|\Gamma_{p_1} \Gamma_{p_2}\| \leq C \|\Gamma_{p_1}\| \|\Gamma_{p_2}\|$ which allows ourselves to ``forget" about the detailed product structure at this point. It will also be equivalent whether the norms are taken with respect to $\slashed{g}$ or $\slashed{g}^\circ$, as the two metrics will be always at least $C^1$-close in applications. Cf.~bootstrap assumption (\ref{decl4}).} Similarly, we denote by $f_{p_1} \Gamma_{p_2}$ any finite sum of products of an $f_{p_1}$ with some $\Gamma_{p_2}$.
\end{definition}
For instance, using the above notation we may write (\ref{mei}) as
\[
 \slashed{\nabla}_3 \overset {(3)}{\Gamma}_1  = f_1\Gamma_2 + f_0\Gamma_2 + 2\Gamma_1 + \Gamma_1\Gamma_1 =  f_0 \Gamma_2 + f_0 \Gamma_1 + \Gamma_1 \Gamma_1   \, ,
\]
where we used for instance that $2M \cdot tr \chi_\circ=f_1$.
The next proposition expresses all null-structure equations in this schematic form.
\subsection{The null-structure and Bianchi equations}
\label{heretherenorm}
\begin{proposition} \label{uneq}
The null-structure equations of Sections \ref{sec:mq} and \ref{sec:rc} take the following schematic form: 
\begin{align} \label{ns3}
&\slashed{\nabla}_3 \overset {(3)}{\Gamma}_p \phantom{+ c \left[\Gamma_p\right] tr \chi \cdot \Gamma_p \ } = E_3 [\overset {(3)}{\Gamma}_p]   \\ 
&\slashed{\nabla}_4 \overset {(4)}{\Gamma}_p + c [\overset {(4)}{\Gamma}_p ] \ tr \chi \cdot \overset {(4)}{\Gamma}_p  =  E_4 [\overset {(4)}{\Gamma}_p]  \label{ns4}
\end{align}
with
\begin{align} \label{basice3e4}
E_3 [\overset {(3)}{\Gamma}_p] &= \sum_{p_1+p_2 \geq p} \left(f_{p_1} + \Gamma_{p_1}\right) \Gamma_{p_2} + \psi_p \, ,  \\
E_4 [ \overset {(4)}{\Gamma}_p ] &=  \boxed{ \sum_{p_1+p_2= p+1}f_{p_1} \overset {(3)}{\Gamma}_{p_2}} + f_1 \Gamma_3 + f_2 \Gamma_p +  \sum_{p_1+p_2\geq p+2} \Gamma_{p_1} \Gamma_{p_2} + \psi_{\geq p+\frac{3}{2}}  . \label{eq:aux1}
\end{align}
Here the weight factor $c [\overset {(4)}{\Gamma}_p ]$ in (\ref{ns4}) is defined as follows: $c\left[\hat{\chi}\right]=c \left[tr \chi - tr \chi_\circ\right] = 1$, $c\left[\eta\right] = \frac{1}{2}$, $c \left[\frac{\Omega_\circ^2}{\Omega^2} - 1 \right]=0$.  
\end{proposition}
\begin{proof}
Direct inspection of all equations. This also reveals that the quantity $f_p$ in Proposition \ref{uneq} is actually always a function, except for equation (\ref{mei}), where $f_0$ stands for the round metric.
\end{proof}

\begin{remark} \label{rem:main}
As a general principle, the right hand side of (\ref{ns3}) ``preserves" the weight in $r$ while the right hand side of (\ref{ns4}) ``improves" it by a power strictly larger than $1$, except for the boxed anomalous term, which however only involves $\overset {(3)}{\Gamma}_p$. See the following Proposition \ref{rem:anoterm}.
\end{remark}
\begin{remark} \label{rem:reno}
The term containing the weight factor $c [\overset {(4)}{\Gamma}_p ]$  in (\ref{ns4}) can be eliminated by considering the  renormalized equation for $r^{2c[\overset {(4)}{\Gamma}_p]} \overset {(4)}{\Gamma}_p$. In view of
\[
\slashed{\nabla}_4 \overset {(4)}{\Gamma}_p + c [\overset {(4)}{\Gamma}_p] tr \chi \cdot \overset {(4)}{\Gamma}_p = r^{-2c[\overset {(4)}{\Gamma}_p]} \slashed{\nabla}_4 ( r^{2c[\overset {(4)}{\Gamma}_p]} \overset {(4)}{\Gamma}_p ) + c[\overset {(4)}{\Gamma}_p] \left(tr \chi - tr \chi_\circ \right) \overset {(4)}{\Gamma}_p \, , 
\]
we may write (\ref{ns4}) as
\[
 \slashed{\nabla}_4 ( r^{2c[\overset {(4)}{\Gamma}_p]} \overset {(4)}{\Gamma}_p ) = r^{2c[\overset {(4)}{\Gamma}_p] } \cdot E_4 [\overset {(4)}{\Gamma}_p] \, .
\]
\end{remark}

We collect the explicit structure of the anomalous term in the following Proposition:
\begin{proposition} \label{rem:anoterm}
The anomalous boxed term appears only in the equation for $\eta$ and the equation for  $ \left(\frac{\Omega^2_\circ}{\Omega^2}-1\right)$. For $\eta$ it has the explicit form 
\[
  \boxed{ \sum_{p_1+p_2=p+1}f_{p_1} \overset {(3)}{\Gamma}_{p_2}}  = \frac{1}{2} tr \chi_\circ \underline{\eta}  \, ,
\]
as seen from (\ref{etaeq}), while for $ \frac{1}{2M} \left(\frac{\Omega^2_\circ}{\Omega^2}-1\right)$ we have the explicit expression
\[
  \boxed{ \sum_{p_1+p_2=p+1}f_{p_1} \overset {(3)}{\Gamma}_{p_2}}  = - \frac{1}{2M} \left(\hat{\omega}-\hat{\omega}_\circ\right)  \, ,
\]
as seen from equation (\ref{4cofa}). Finally, the term $f_1 \Gamma_3$ in (\ref{eq:aux1}) appears only for the $\slashed{\nabla}_4 \left( tr \chi - tr \chi_\circ\right)$ equation. It has the explicit form $f_1 \Gamma_3 = tr \chi_\circ \left(\hat{\omega} - \hat{\omega}_\circ\right)$. 
\end{proposition}

The general principle of Remark \ref{rem:main} also applies to the Bianchi equations:

\begin{proposition}  \label{uneq2}
With $f_p$ as in Definition \ref{def:fp}, each Bianchi-pair $\left(  \uppsi_p, \uppsi^\prime_{p^\prime} \right)$ satisfies
\begin{align} \label{bi3}
\slashed{\nabla}_3 \uppsi_p  \phantom{a^\prime  x \gamma_4\left(\uppsi^\prime_{p^\prime}\right) tr {\chi} \uppsi^\prime_{p^\prime}}   &= \slashed{\mathcal{D}} \uppsi^\prime_{p^\prime}  + E_{3} \left[\uppsi_p\right]   \\ \label{bi4}
\slashed{\nabla}_4 \uppsi^\prime_{p^\prime} +  \gamma_4\left(\uppsi^\prime_{p^\prime}\right) tr {\chi} \uppsi^\prime_{p^\prime} &= \slashed{\mathcal{D}} \uppsi_p  + E_{4} \left[\uppsi^\prime_{p^\prime}\right]
\end{align}
where $\slashed{\mathcal{D}}$ is to be replaced by the angular operator appearing for the particular curvature component under consideration.\footnote{For instance, $\slashed{\mathcal{D}} \uppsi^\prime_{p^\prime} = -2 \slashed{\mathcal{D}}_2^\star \beta$ in the equation for $\slashed{\nabla}_3 \alpha$ or $\slashed{\mathcal{D}} \uppsi^\prime_{p^\prime}= -\slashed{curl} \beta$ in the equation for $\slashed{\nabla}_4 \sigma$, etc. Cf.~Section \ref{bisec}.} Also $\gamma_{4}\left(\uppsi^\prime_{p^\prime}\right):=\frac{p^\prime}{2}$ except for $\beta$, which has $\gamma_4 \left(\beta\right) = 2$, is the weight of the curvature component under consideration.  The error terms are of the following structure:
\begin{align}
E_{3} \left[\uppsi_p\right] &= f_1 \uppsi_p +  f_3 \Gamma_2 + \sum_{p_1+p_2\geq p}  \Gamma_{p_1} \psi_{p_2}  \, , \label{e3e} \\
E_{4} \left[\uppsi^\prime_{p^\prime}\right] &=  f_2 \uppsi^\prime_{p^\prime} +  f_3 \Gamma_{\min\left(p^\prime,2\right)} + \sum_{p_1+p_2\geq  p^\prime+\frac{3}{2}} \Gamma_{p_1} \psi_{p_2} \, . \label{e4e}
\end{align}
\end{proposition}

\begin{proof}
Direct inspection of all equations. 
\end{proof}

Again, the $tr \chi$-term in (\ref{bi4}) can be eliminated using the renormalization of Remark \ref{rem:reno}. We note that $E_3\left[\uppsi_p\right]$ typically (i.e.~except for 
the boxed terms in the equations for $\rho$ and $\sigma$ of Section \ref{bisec}) admits better decay than stated above. In particular, for $E_3\left[\alpha\right]$ the sum starts at $p+1$ while for $E_3 \left[\beta\right]$ the sum starts at $p+\frac{1}{2}$. In Section \ref{errorest}, we will actually revisit the individual expression for the curvature components when we prove estimates for the Bianchi equations. 
\subsection{Commutation} \label{sec:commute}
The fundamental goal of this section is to show that the structure of the equations observed in Propositions \ref{uneq} and \ref{uneq2} is preserved under appropriate commutations. Our formalism to commute is slightly non-standard and requires some explanation.

Let $\mathfrak{D} = \{ M \slashed{\nabla}_3, r \slashed{\nabla}_4, r \slashed{\nabla} \}$ denote a collection of operators acting on covariant $S^2_{u,v}$-tensors of rank $n$. We also define the subsets $\mathfrak{D}_{\nearrow} = \{ r \slashed{\nabla}_4, r \slashed{\nabla} \}$ and $\mathfrak{D}_{\nwarrow} = \{  M \slashed{\nabla}_3, r \slashed{\nabla} \}$. 
Clearly, the $M\slashed{\nabla}_3$ and $r\slashed{\nabla}_4$ operators do not change the nature of a tensor, while applying $\slashed{\nabla}$ changes a $\left(0,N\right)$-tensor into a $\left(0,N+1\right)$-tensor. In our scheme, we estimate higher angular derivatives of a curvature component $u_{A_1...A_N}$ by considering the evolution equation for the covariant $(N+1)$-tensor $\slashed{\nabla}_B u_{A_1...A_N}$. The point here is that the $L^2 \left(S^2_{u,v}\right)$-norm  of the latter is equivalent (modulo lower order terms of type $\|u\|^2$, which can be thought as having been estimated in the previous step) to the $L^2 \left(S^2_{u,v}\right)$-norm of all angular-derivatives applied to $u_{A_1...A_N}$, see (\ref{eq:tec}).\footnote{Alternatively, one could commute the null-decomposed equations by a basis of Schwarzschildean angular momentum operators (keeping the rank of the quantities). Indeed, in the spherically-symmetric case these operators would trivially commute, while in our approach there would still be a lower order error-term present. However, given only asymptotic-flatness, where commutation with angular momentum operators would also generate errors, the commutation formula for $r \slashed{\nabla}$ is in fact somewhat cleaner (as is manifest by the following Lemma \ref{comlem}) and the additional lower order error-term is easily 
controlled.\label{footnotecom}} 

The reason for introducing these operators is the simple 
\begin{quotation}
\noindent
{\it Commutation Principle: One expects to prove the same decay bounds for $\mathfrak{D} \Gamma_p$ as for $\Gamma_p$ and the same decay bounds for $\mathfrak{D}\psi_p$ as for $\psi_p$. } 
\end{quotation}

This principle serves as a useful guide to interpret the structure of the terms appearing in the equations.

Finally, a further remark about our schematic notation: If $\xi$ is a $S^2_{u,v}$-tensor of rank $n$, we denote by $\mathfrak{D}^k \xi$ any fixed $k$-tuple $\mathfrak{D}_k \mathfrak{D}_{k-1}...\mathfrak{D}_1 \xi$ of operators applied to $\xi$, where each $\mathfrak{D}_i \in  \{ M \slashed{\nabla}_3, r \slashed{\nabla}_4, r \slashed{\nabla} \}$.  Note that the rank of the tensor $\mathfrak{D}^k \xi$ depends on the number of angular operators in the $k$-tuple.

For instance, if $\mathfrak{D}^k$ contains $i$ angular operators, then $\left(\mathfrak{D}^k \underline{\alpha}\right)_{C_1...C_iAB}$ is a covariant $S^2_{u,v}$-tensor of rank $i+2$ which is symmetric and traceless with respect to its last two indices. Similarly, $\left(\mathfrak{D}^k \underline{\beta}\right)_{C_1...C_iB}$ is a covariant $S^2_{u,v}$-tensor of rank $i+1$. For the latter, the operators $\slashed{div}$, $\slashed{curl}$ are always defined with respect to the \emph{last} index. Similarly, the operator $\slashed{\mathcal{D}}^\star_2$ can be generalized as $\left(-2\slashed{\mathcal{D}}^\star_2\left(\mathfrak{D}^k \underline{\beta}\right)\right)_{C_1...C_iAB} $
\[ 
= \slashed{\nabla}_A \left(\mathfrak{D}^k \underline{\beta}\right)_{C_1...C_iB} + \slashed{\nabla}_B \left(\mathfrak{D}^k \underline{\beta}\right)_{C_1...C_iA} - \slashed{g}_{AB} \left[\slashed{\nabla}^D \left(\mathfrak{D}^k \underline{\beta}\right)_{C_1...C_iD}\right] \, ,
\]
mapping to a covariant $S^2_{u,v}$-tensor of rank $i+2$ which is symmetric and traceless with respect to its last two indices.

\subsubsection{The Commutation Lemma}
We recall Lemma 7.3.3 of \cite{ChristKlei}, stated and proven there for a general null-frame. In our  gauge, some of the terms in fact vanish; we have kept those below (but crossed them out) to allow direct comparison with \cite{ChristKlei}.
\begin{lemma} \label{commutelemma}  \label{comlem}
Let $\xi_{C_1 \ ... \  C_N}$ be a $\left(0,N\right)$-tensor on $S^2_{u,v}$. Then
\begin{align}
\left[ \slashed{\nabla}_4, \slashed{\nabla}_B \right] \xi_{A_1... A_k}  &=  F_{4 B A_1 ... A_k} - \left(\hat{\chi}_{B}^{\phantom{B}C} + \frac{1}{2} tr \chi \slashed{g}_{B}^{\phantom{B}C}  \right)  \slashed{\nabla}_C \xi_{A_1\ldots A_k} \, ,
\nonumber \\
\left[ \slashed{\nabla}_3, \slashed{\nabla}_B \right] \xi_{A_1... A_k}  &= F_{3 B A_1... A_k}- \left(\underline{\hat{\chi}}_{B}^{\phantom{B}C} +  \frac{1}{2} tr \underline{\chi} \slashed{g}_{B}^{\phantom{B}C}  \right) \slashed{\nabla}_C \xi_{A_1... A_k} \, ,
\nonumber \\
\left[ \slashed{\nabla}_3 , \slashed{\nabla}_4\right] \xi_{A_1... A_k} &= F_{3 4 A_1 ... A_k}  \,  \nonumber
\end{align}
where
\begin{align}
F_{4 B A_1 ... A_k} &= \xcancel{ \left(\underline{\eta}_B +\zeta_B\right)} \slashed{\nabla}_4 \xi_{A_1... A_k}  \nonumber \\ &+ \sum_{i=1}^k \Big( \chi_{A_iB}\underline{\eta}_C - \chi_{BC} \underline{\eta}_{A_i} + \epsilon_{A_i C} {}^\star \beta_B\Big) \xi^{\phantom{A_1 ... A_{i-1}} C}_{A_1 ... A_{i-1} \phantom{C} A_{i+1} \ldots A_k} \, , \nonumber 
\end{align}
\begin{align}
F_{3 B A_1 ... A_k} &=  \left(\eta_B -\zeta_B\right) \slashed{\nabla}_3 \xi_{A_1... A_k}  \nonumber \\ &+ \sum_{i=1}^k \Big( \underline{\chi}_{A_iB}{\eta}_C -  \underline{\chi}_{BC}\eta_{A_i} - \epsilon_{A_i C} {}^\star \underline{\beta}_B\Big) \xi^{\phantom{A_1 ... A_{i-1}} C}_{A_1 ... A_{i-1} \phantom{C} A_{i+1}  \ldots A_k} \, , \nonumber
\end{align}
\begin{align}
F_{3 4 A_1 ... A_k} &= \boxed{ \hat{\omega} \slashed{\nabla}_3 \xi_{A_1... A_k} } - \xcancel{ \hat{\underline{\omega}} } \cdot \slashed{\nabla}_4 \xi_{A_1... A_k} + \left(\eta^B - \underline{\eta}^B\right) \slashed{\nabla}_B \xi_{A_1... A_k} \nonumber \\  &+ 2 \sum_{i=1}^k \Big( \underline{\eta}_{A_i} \eta_C - \eta_{A_i} \underline{\eta}_C + \epsilon_{A_iC} \sigma \Big) \xi^{\phantom{A_1 ... A_{i-1}} C}_{A_1 ... A_{i-1} \phantom{C} A_{i+1}  \ldots A_k} \,  \nonumber
\end{align}
and we crossed out ($
\xcancel{A}$) terms which vanish in our choice of frame. In addition,
\begin{align}
\left[ \slashed{\nabla}_A, \slashed{\nabla}_B \right] \xi_{C_1...C_k} &= K \sum_{i=1}^k \left( \slashed{g}_{C_i B} \xi_{C_1...C_{i-1}AC_{i+1}...C_k} - \slashed{g}_{C_i A} \xi_{C_1...C_{i-1}BC_{i+1}...C_k}\right) \nonumber
\end{align}
with $K$ the Gauss curvature on the spheres $S^2_{u,v}$.
\end{lemma}
Note that in our schematic notation we have
\begin{align}
F_{4BA_1...A_k} &=  \left( f_1 \Gamma_2 + \Gamma_2 \Gamma_2 + \psi_\frac{7}{2}\right) \xi_{A_1...A_k} \,  ,
\nonumber \\
F_{3BA_1...A_k} &= \Gamma_2 \mathfrak{D}_\nwarrow u_{A_1...A_k} + \left( f_1 \Gamma_2 + \Gamma_1 \Gamma_2 + \psi_2\right) \xi_{A_1...A_k} \,  , \nonumber \\
F_{34A_1...A_k} &=  \left(f_2 + \Gamma_2 + \Gamma_3 \right) \mathfrak{D}_\nwarrow \xi_{A_1...A_k} +  \left( \Gamma_2 \Gamma_2 + \psi_3\right) \xi_{A_1...A_k} \,  ,
\nonumber \\
K &= 
f_2 + f_1 \Gamma_2 + \Gamma_1 \Gamma_2 + \Gamma_2 \Gamma_2 + \psi_3 \, , \nonumber
\end{align}
where we have written the Gauss curvature $K$ of (\ref{Gauss}) in schematic notation. In particular, the only ``linear" (in the sense that $\xi$ does not multiply a term which is zero in Schwarzschild) term in the $F$'s is the boxed term in $F_{34}$ whose presence (in particular the sign of its coefficient $\hat{\omega}$) will be seen as a manifestation of the redshift.
\subsubsection{The commuted Bianchi equations}
Using the Lemma we shall now establish that the structure of the Bianchi equations is preserved under ``commutation" with operators from $\mathfrak{D}$. Beware that the operator $\slashed{\nabla}$ actually changes the order of the tensor as explained at the beginning of Section \ref{sec:commute}. Note that we certainly have for $s \geq 1$
\[
\mathfrak{D}^s f_{p}  = f_p
\]
in our notation. 

\begin{proposition} \label{comsum}
For any positive integer $k$, the commuted equations for the Bianchi-pairs take the following form:
\begin{align}
\slashed{\nabla}_3 \left(\mathfrak{D}^k \uppsi_p\right)  \phantom{a^\prime  x \gamma_4\left(\uppsi^\prime_{p^\prime}\right) tr {\chi}  \left(\mathfrak{D}^k \uppsi^\prime_{p^\prime}\right)}   &= \slashed{\mathcal{D}}  \left(\mathfrak{D}^k\uppsi^\prime_{p^\prime}\right)  + E_{3} \left[\mathfrak{D}^k \uppsi_p\right]  \, ,   \label{cobi1}
 \\
\slashed{\nabla}_4  \left(\mathfrak{D}^k \uppsi^\prime_{p^\prime} \right) +  \gamma_4\left(\uppsi^\prime_{p^\prime}\right) tr {\chi}  \left(\mathfrak{D}^k \uppsi^\prime_{p^\prime} \right) &= \slashed{\mathcal{D}}  \left(\mathfrak{D}^k \uppsi_p \right)  + E_{4} \left[\mathfrak{D}^k \uppsi_{p^\prime}^\prime\right] \, , \label{cobi2}
\end{align}
with the error terms 
\begin{align} \label{3dirbu}
E_{3} \left[\mathfrak{D}^{k} \uppsi_p\right] = \mathfrak{D} \left(E_{3} \left[\mathfrak{D}^{k-1} \uppsi_p \right]  \right) + \left(f_1 + r \Gamma_2 \right) \left(\mathfrak{D}^{k} \uppsi_p + \boxed{\mathfrak{D}^{k} \uppsi^\prime_{p^\prime}} \ \right)  \nonumber \\
 + \left(f_1 + \Lambda_1 + \Lambda_2 \right) \left( \mathfrak{D}^{k-1} \uppsi_p + \boxed{\mathfrak{D}^{k-1}  \uppsi^\prime_{p^\prime}} \ \right) \, ,
 \end{align}
\begin{align} \label{4dirbu}
E_{4} \left[\mathfrak{D}^{k} \uppsi^\prime_{p^\prime}\right]=   \mathfrak{D} \left( E_{4} \left[\mathfrak{D}^{k-1}  \uppsi^\prime_{p^\prime} \right] \right) + \left[f_2 + \Gamma_2\right] \left(\mathfrak{D}^{k} \uppsi_p + \boxed{ \mathfrak{D}^{k} \uppsi^\prime_{p^\prime}} \ \right) \nonumber \\
+ \left(f_1+ \Lambda_1 \right) \mathfrak{D}^{k-1}  \uppsi_p +\left(f_2+ \Lambda_2 \right) \boxed{\mathfrak{D}^{k-1}  \uppsi^\prime_{p^\prime}}  + f_0 E_{4} \left[\mathfrak{D}^{k-1}  \uppsi^\prime_{p^\prime} \right] \, ,
\end{align}
where $\Lambda_1$ and $\Lambda_2$ are expressions of the form
\begin{align}
\Lambda_1 &=  r \psi_{\frac{7}{2}} + r \psi_2 + \psi_2 + f_1 \Gamma_1 + r\left(f_0 + \Gamma_1\right) \left[  r \Gamma_2 \left(f_1+\Gamma_1 + \Gamma_2\right)  \right] \, , \nonumber \\
\Lambda_2 &= \mathfrak{D}\Gamma_2 + f_0 \Gamma_2 + r \Gamma_1 \Gamma_2 + r \Gamma_2 \Gamma_2 + r \psi_{\frac{7}{2}} + \psi_3  \, .
\end{align}
\end{proposition}
To understand the structure of the errors (\ref{3dirbu}) and (\ref{4dirbu}) we observe: 
\begin{remark}
In view of the commutation principle of Section \ref{sec:commute},
 the radial decay bound one expects to prove for $E_{3} \left[\mathfrak{D}^{k} \uppsi_p\right] $ and $E_{4} \left[\mathfrak{D}^{k} \uppsi^\prime_{p^\prime}\right]$ is the same as that for $k=0$ in Proposition \ref{uneq2}. In particular, one expects to prove that $E_{4} \left[\mathfrak{D}^{k} \uppsi^\prime_{p^\prime}\right]$ decays at least as $r^{-p^\prime-\frac{3}{2}}$ for all $k$.
\end{remark}
\begin{corollary}
The only terms  of $E_{3} \left[\mathfrak{D}^{k} \uppsi_p\right] $ and $E_{4} \left[\mathfrak{D}^{k} \uppsi^\prime_{p^\prime}\right]$ which are not quadratic in decaying quantities are of the form (cf.~(\ref{e3e}) and (\ref{e4e})):
\begin{align} \label{lin3term}
E_3^{lin}  \left[\mathfrak{D}^{k} \uppsi_p\right]  &=   \sum_{i=0}^k  \left(f_1 \mathfrak{D}^{i} \uppsi_p + f_1 \mathfrak{D}^{i} \uppsi^\prime_{p^\prime}  \right) + \mathfrak{D}^{k}  \left(f_3 \Gamma_2\right) \, ,
 \\
E_4^{lin}  \left[\mathfrak{D}^{k} \uppsi^\prime_{p^\prime}\right] &=  \sum_{i=0}^{k}  \left(f_1 \mathfrak{D}^{i} \uppsi_p + f_2 \boxed{\mathfrak{D}^{i} \uppsi^\prime_{p^\prime}}  \right)+ \sum_{i=0}^{k} \mathfrak{D}^{i} \left(f_3 \Gamma_{min(p^\prime,2)}\right) \, . \label{lin4term}
\end{align}
\end{corollary}

\begin{remark} \label{rem:comstructure}
Very importantly, in (\ref{4dirbu}) and (\ref{lin4term}) the boxed terms $\mathfrak{D}^{i}  \uppsi^\prime_{p^\prime}$ only appear with coefficients decaying like $\frac{1}{r^2}$ or stronger, while all other curvature terms appear with $\frac{1}{r}$ coefficients. This is essential because the energy estimate for this Bianchi-pair will only provide control over $\mathfrak{D}^i \uppsi^\prime_{p^\prime}$ on constant $v$-slices. Hence spacetime terms of the form $\int_{\mathcal{M}} f_1 \|\mathfrak{D}^i \uppsi^\prime_{p^\prime}\|^2$ would lead to a logarithmic divergence when integrating in $v$ near infinity. 

In contrast to the above, from (\ref{3dirbu}) such $\int_{\mathcal{M}} f_1 \|\mathfrak{D}^i \uppsi^\prime_{p^\prime}\|^2$-terms indeed appear. However, here they are unproblematic, as in the energy estimate for this Bianchi-pair they will be multiplied by $\mathfrak{D}^i \uppsi_{p}$. The flux of the latter component is controlled on constant $u$ by the energy estimate. We can hence control the resulting spacetime errors by the flux-term on constant $u$, and then exploit the exponential decay in $u$ of this flux when integrating in $u$.
\end{remark}

\begin{proof}[Proof (Proposition \ref{comsum}).]
We will establish the statement for $k=1$. The general case follows by a simple induction.
Starting with the un-commuted Bianchi equation
\begin{align} \label{bic}
\slashed{\nabla}_4 \uppsi^\prime_{p^\prime} +  \gamma_4\left(\uppsi^\prime_{p^\prime}\right) tr {\chi} \uppsi^\prime_{p^\prime} &= \slashed{\mathcal{D}} \uppsi_p  + E_{4} \left[\uppsi^\prime_{p^\prime}\right] \, ,
\end{align}
we compute the commutator (abbreviating $\gamma_4 =  \gamma_4\left(\uppsi^\prime_{p^\prime}\right)$ for the moment)
\begin{align}
\mathfrak{C} = \left[ r \slashed{\nabla}_4 , \slashed{\nabla}_4  \right] \uppsi^\prime_{p^\prime} - \left[ r \slashed{\nabla}_4 , \slashed{\mathcal{D}}  \right] \uppsi_{p} +  r \cdot \gamma_4 \left(\slashed{\nabla}_4 tr \chi \right) \uppsi^\prime_{p^\prime} \nonumber \\ 
= - \frac{1}{2} r \ tr \chi_\circ  \left[ \slashed{\nabla}_4  \uppsi^\prime_{p^\prime} - \slashed{\mathcal{D}}  \uppsi_p \right]  - r\left[\left[\slashed{\nabla}_4 , \slashed{\mathcal{D}} \right] \uppsi_p + \frac{1}{2} tr \chi_\circ \slashed{\mathcal{D}} \uppsi_p \right] + r \ \gamma_4 \left(\slashed{\nabla}_4 tr \chi \right) \uppsi^\prime_{p^\prime}  \nonumber \\ 
= - \frac{1}{2} r \ tr \chi_\circ  \left( - \gamma_4  tr {\chi} \uppsi^\prime_{p^\prime} +  E_{4} \left[\uppsi^\prime_{p^\prime}\right]\right) - \left[ r F_4 \uppsi_p - \left( \hat{\chi} + \frac{1}{2} \left(tr \chi - tr \chi_\circ \right) \right)  r \slashed{\mathcal{D}}  \uppsi_p\right] \nonumber \\ +  r \ \gamma_4 \left( - \frac{1}{2} tr \chi \ tr \chi + \hat{\omega} tr \chi  - \|\hat{\chi}\|^2 \right) \uppsi^\prime_{p^\prime} \, ,\nonumber
\end{align}
where we used the commutation Lemma and reinserted the Bianchi equation (\ref{bic}) as well as the equation $\slashed{\nabla}_4 \left(tr \chi\right) + \frac{1}{2} tr \chi \ tr \chi = \hat{\omega} tr \chi  - \|\hat{\chi}\|^2$ from (\ref{trx4}). Note that there is a cancellation of the weakest decaying contribution in the first and last term. This is crucial because otherwise terms of the form $f_1 \uppsi^\prime_{p^\prime}$ would appear! In schematic notation we have using again Lemma \ref{comlem},
\begin{align}
\mathfrak{C}
 = \left(f_2 + \Gamma_2 + \Gamma_2 \Gamma_2 r \right) \uppsi_{p^\prime}^\prime +
 \Gamma_2 \mathfrak{D} \uppsi_p +  r \left[ f_1 \Gamma_2 + \Gamma_1 \Gamma_2 + \uppsi_\frac{7}{2} \right] \uppsi_p + f_0 \cdot E_{4} \left[\uppsi^\prime_{p^\prime}\right] \nonumber \, .
\end{align}
It follows that the $r\slashed{\nabla}_4$-commuted equation (\ref{bic}) finally reads
\begin{align}
\slashed{\nabla}_4 \left(r \slashed{\nabla}_4 \uppsi^\prime_{p^\prime} \right) +  \gamma_4\left(\uppsi^\prime_{p^\prime}\right) tr {\chi}\left(r \slashed{\nabla}_4 \uppsi^\prime_{p^\prime} \right) &= \slashed{\mathcal{D}} \left(r \slashed{\nabla}_4 \uppsi_p \right)  + {E}_{4}  \left[r \slashed{\nabla}_4 \uppsi^\prime_{p^\prime} \right]
\end{align}
with
\begin{align}
{E}_{4}  \left[r \slashed{\nabla}_4 \uppsi^\prime_{p^\prime} \right] = \mathfrak{D} \left(E_{4} \left[\uppsi^\prime_{p^\prime}\right] \right) -\mathfrak{C} = \mathfrak{D} \left(E_{4} \left[\uppsi^\prime_{p^\prime}\right] \right) + 
\Gamma_2 \mathfrak{D} \uppsi_p \nonumber \\ + f_0 E_{4} \left[\uppsi^\prime_{p^\prime}\right]  +  r \left[ f_1 \Gamma_2 + \Gamma_1 \Gamma_2 + \uppsi_\frac{7}{2} \right] \uppsi_p +  \left(f_2 + \Gamma_2 + \Gamma_2 \Gamma_2 r \right) \uppsi^\prime_{p^\prime} \, .
\end{align}
Note that according to the commutation principle, ${E}_{4}  \left[r \slashed{\nabla}_4 \uppsi^\prime_{p^\prime} \right]$ will decay at least as fast in $r$ as the original $E_{4} \left[\uppsi^\prime_{p^\prime}\right]$.
Commuting (\ref{bic}) with $\slashed{\nabla}_3$ we see that
\begin{align}
\slashed{\nabla}_4 \left( \slashed{\nabla}_3 \uppsi^\prime_{p^\prime} \right) +  \gamma_4\left(\uppsi^\prime_{p^\prime}\right) tr {\chi}\left(\slashed{\nabla}_3 \uppsi^\prime_{p^\prime} \right) &= \slashed{\mathcal{D}} \left( \slashed{\nabla}_3 \uppsi_p \right)  + {E}_{4}  \left[\slashed{\nabla}_3 \uppsi^\prime_{p^\prime} \right]
\end{align}
with
\begin{align}
{E}_{4}  \left[\slashed{\nabla}_3 \uppsi^\prime_{p^\prime} \right] =  \slashed{\nabla}_3 \left(E_{4} \left[\uppsi^\prime_{p^\prime}\right] \right) + 
\left(f_2 + \Gamma_2\right) \left[ \mathfrak{D} \uppsi_{p^\prime}^\prime + \mathfrak{D} \uppsi_p \right] \nonumber \\ 
 + \left(f_2 + \mathfrak{D} \Gamma_2 + \Gamma_2 \Gamma_2 + \uppsi_3 \right) \uppsi^\prime_{p^\prime} +  \left[ f_1 \Gamma_2 + \Gamma_1 \Gamma_2 + \uppsi_2 \right]  \uppsi_p \, .
\end{align}
Note that the $\mathfrak{D} \Gamma_2$-term arises from the derivative falling on $tr \chi$. Note also that in the first line $f_2=\hat{\omega}_0$ is precisely the term arising from the boxed term in the commutation Lemma \ref{comlem}. This is a \emph{manifestation of the redshift} at the level of commutation (cf.~Section 7.2 of \cite{mihalisnotes} and Section \ref{sec:redshiftrole}). Using Lemma \ref{comlem} to derive
\begin{align}
\left[ r \slashed{\nabla}_A , \slashed{\nabla}_4 \right] u = \Gamma_2 r \slashed{\nabla}_A u +  \left(  \Gamma_2  \right) \mathfrak{D} u  
+ r \left[ f_1 \Gamma_2 + \Gamma_1 \Gamma_2 + \uppsi_\frac{7}{2} \right] u \nonumber \\
= \Gamma_2  \cdot \mathfrak{D} u  
+ r \left[ f_1 \Gamma_2 + \Gamma_1 \Gamma_2 + \uppsi_\frac{7}{2} \right] u \, ,
\end{align}
we obtain for commutation with $r \slashed{\nabla}_A$ the expression
\begin{align}
\slashed{\nabla}_4 \left(r \slashed{\nabla} \uppsi^\prime_{p^\prime} \right) +  \gamma_4\left(\uppsi^\prime_{p^\prime}\right) tr {\chi}\left(r \slashed{\nabla} \uppsi^\prime_{p^\prime} \right) &= \slashed{\mathcal{D}} \left(r \slashed{\nabla} \uppsi_p \right)  + {E}_{4}  \left[r \slashed{\nabla} \uppsi^\prime_{p^\prime} \right] \, ,
\end{align}
where
\begin{align}
{E}_{4}  \left[r \slashed{\nabla} \uppsi^\prime_{p^\prime} \right] =r \slashed{\nabla} \left(E_{4} \left[\uppsi^\prime_{p^\prime}\right] \right) + 
\Gamma_2 \mathfrak{D} \uppsi_{p^\prime}^\prime +  \left( \mathfrak{D} \Gamma_2 +  f_0 \Gamma_2 + r \Gamma_1 \Gamma_2 + r \uppsi_\frac{7}{2} \right) \uppsi^\prime_{p^\prime} \nonumber \\  + r\left(f_0+\Gamma_1\right) \left[ f_{2}  + r \left(f_1+\Gamma_ 1 + \Gamma_2\right) \Gamma_2 +  \uppsi_3 \right]  \uppsi_p \, . \nonumber
\end{align}
Collecting terms we obtain the statement of Proposition \ref{comsum} in the $4$-direction.

We turn to commuting the other Bianchi equation
\begin{align} \label{bic2}
\slashed{\nabla}_3 \uppsi_p = \slashed{\mathcal{D}} \uppsi^\prime_{p^\prime}  + E_{3} \left[\uppsi_p\right]  \, .
\end{align}
Commuting (\ref{bic2}) with $\slashed{\nabla}_3$ yields
\begin{align}
\slashed{\nabla}_3 \left(\slashed{\nabla}_3 \uppsi_p \right) = \slashed{\mathcal{D}} \left(\slashed{\nabla}_3 \uppsi^\prime_{p^\prime}  \right) + E_{3} \left[\slashed{\nabla}_3 \uppsi_p\right]  
\end{align}
with
\begin{align}
E_{3} \left[\slashed{\nabla}_3 \uppsi_p\right] = \slashed{\nabla}_3  E_{3} \left[\uppsi_p\right]  + \left(f_2 + \Gamma_2 \right) \mathfrak{D} \uppsi_p + \left[f_1 \Gamma_2 + \Gamma_1 \Gamma_2 + \uppsi_2 \right] \uppsi_p \, .
\end{align}
Commuting  (\ref{bic2}) with $r \slashed{\nabla}_4$ yields
\begin{align}
\slashed{\nabla}_3 \left(r \slashed{\nabla}_4 \uppsi_p \right) = \slashed{\mathcal{D}} \left(r \slashed{\nabla}_4 \uppsi^\prime_{p^\prime}  \right) + E_{3} \left[r\slashed{\nabla}_4 \uppsi_p\right]  
\end{align}
with
\begin{align}
E_{3} \left[r\slashed{\nabla}_4 \uppsi_p\right] = r \slashed{\nabla}_4  E_{3} \left[\uppsi_p\right]  + \left(f_1 + r \Gamma_2 \right) \left[\mathfrak{D} \uppsi_p + \mathfrak{D} \uppsi^\prime_{p^\prime} \right] \nonumber \\
+ r \left[f_1 \Gamma_2 + \Gamma_1 \Gamma_2 + \uppsi_\frac{7}{2} \right]  \uppsi^\prime_{p^\prime}  + r \left[ \Gamma_2 \Gamma_2 + \uppsi_3 \right] \uppsi_p \, .
\end{align}
Finally, commutation with the weighted angular derivative $r \slashed{\nabla}$ yields
\begin{align}
\slashed{\nabla}_3 \left(r \slashed{\nabla} \uppsi_p \right) = \slashed{\mathcal{D}} \left(r \slashed{\nabla} \uppsi^\prime_{p^\prime}  \right) + E_{3} \left[r\slashed{\nabla} \uppsi_p\right]  
\end{align}
with
\begin{align}
E_{3} \left[r\slashed{\nabla} \uppsi_p\right] =&\, r \slashed{\nabla}  E_{3} \left[\uppsi_p\right]  + \left(f_1 + r \Gamma_2 \right) \left[\mathfrak{D} \uppsi_p  \right] \nonumber \\
&+ r \left(f_0+\Gamma_1\right) \left[ f_{2}  + \left(f_1+\Gamma_ 1\right) \Gamma_2 +  \uppsi_3 \right] \uppsi^\prime_{p^\prime} + r \left[ f_1 \Gamma_2 + \Gamma_1 \Gamma_2 + \uppsi_2 \right]  \uppsi_p \, .
\end{align}
\end{proof}

\subsubsection{The commuted null structure equations}
\begin{proposition} \label{prop:nscommute}
The commuted null-structure equations take the form 
\begin{align}
\slashed{\nabla}_3 \Big( \mathfrak{D}^k \overset {(3)}{\Gamma}_p \Big) \phantom{+ c\left[\Gamma_p\right] tr \chi  \left(\mathfrak{D}^k \Gamma_p \right)} = E_3  \Big[\mathfrak{D}^k \overset {(3)}{\Gamma}_p\Big] \, ,
\end{align}
\begin{align} \label{anomaly}
\slashed{\nabla}_4 \Big(\mathfrak{D}^k \overset {(4)}{\Gamma}_p \Big) + c[\overset {(4)}{\Gamma}_p] \ tr \chi  \Big(\mathfrak{D}^k \overset {(4)}{\Gamma}_p \Big) = E_4 \Big[ \mathfrak{D}^k \overset {(4)}{\Gamma}_p\Big] \, ,
\end{align}
with the right hand sides
\begin{equation}
\begin{split}
E_3  \Big[\mathfrak{D}^{k} \overset {(3)}{\Gamma}_p\Big] = \mathfrak{D} \Big( E_3  \Big[\mathfrak{D}^{k-1} \overset {(3)}{\Gamma}_p\Big] \Big) + \left(f_1 + \Gamma_1 + r \Gamma_2\right) \mathfrak{D}^k \overset {(3)}{\Gamma}_p \\
+ \left(\Gamma_2 + r \Gamma_1 \Gamma_2 + r \psi_{\geq \frac{7}{2}} \right)\mathfrak{D}^{k-1} \overset {(3)}{\Gamma}_p \, ,
\end{split}
\end{equation}
\begin{equation} \label{4dirgu}
\begin{split}
E_4 \Big[ \mathfrak{D}^k \overset {(4)}{\Gamma}_p \Big] = \mathfrak{D} \Big( E_4 \Big[ \mathfrak{D}^{k-1} \overset {(4)}{\Gamma}_p \Big] \Big) + \left(f_2 + \Gamma_2\right) \mathfrak{D}^k \overset {(4)}{\Gamma}_p  \\
+ \left(f_2 + \mathfrak{D}\Gamma_2 + \Gamma_2 + r \Gamma_1 \Gamma_2 + r \psi_{\geq \frac{7}{2}} \right) \mathfrak{D}^{k-1} \overset {(4)}{\Gamma}_p \, .
\end{split}
\end{equation}
We also recall $c\left[\hat{\chi}\right]=c \left[tr \chi - tr \chi_\circ\right] = 1$, $c\left[\eta\right] = \frac{1}{2}$, $c \left[\frac{\Omega_\circ^2}{\Omega^2} - 1 \right]=0$ and the expressions (\ref{basice3e4}), (\ref{eq:aux1}) from Proposition \ref{uneq}.
\end{proposition}

\begin{remark}
Note again that in (\ref{4dirgu}) the terms $\mathfrak{D}^k \overset {(4)}{\Gamma}_p$ and $\mathfrak{D}^{k-1} \overset {(4)}{\Gamma}_p$ only appear with terms decaying like $r^{-2}$, while in the $3$-direction $r^{-1}$ appears.
\end{remark}

\begin{proof}
For the equation in the $3$-direction, this is a straightforward application of Lemma \ref{comlem} using the formulae established in the previous section.
As for the Bianchi equations, in the $4$-direction the identity (\ref{trx4}) is crucial in deriving the correct commutation formula for the null-structure equations. Indeed, if $\Gamma_p$ satisfies
\begin{align}
\slashed{\nabla}_4 \Gamma_p + c_1 \ tr \chi \ \Gamma_p = RHS \, ,
\end{align}
then, in view of (\ref{trx4})
\begin{align}
\left[ r \slashed{\nabla}_4 , \slashed{\nabla}_4  + c_1 tr \chi \right] \Gamma_p = - \slashed{\nabla}_4 r \cdot \slashed{\nabla}_4 \Gamma_p + c_1 \ r \left( \slashed{\nabla}_4  tr \chi \right) \Gamma_p \nonumber \\
= - \frac{1}{2} r \ tr \chi_\circ \left( -c_1 \ tr \chi \Gamma_p + RHS \right) + c_1 r \ \left( - \frac{1}{2} tr \chi \ tr \chi + \hat{\omega} tr \chi  - \|\hat{\chi}\|^2\right) \Gamma_p \nonumber \\
= \left(f_2 + \Gamma_2 + r \Gamma_2 \Gamma_2 \right) \Gamma_p + f_0 \  RHS \, , \nonumber
\end{align}
\begin{align}
\left[  \slashed{\nabla}_3 , \slashed{\nabla}_4  + c_1 tr \chi \right] \Gamma_p =\left(f_2 + \Gamma_2\right) \mathfrak{D}\Gamma_p + \left(f_2 + \mathfrak{D} \Gamma_2 + \Gamma_2 \Gamma_2 - \rho_\circ + \psi_3 \right) \Gamma_p \, , \nonumber
\end{align}
\begin{align}
\left[  r \slashed{\nabla}_A , \slashed{\nabla}_4  + c_1 tr \chi \right] \Gamma_p =\left(f_2 + \Gamma_2\right) \mathfrak{D}\Gamma_p + \left(f_2 + \mathfrak{D} \Gamma_2 + \Gamma_2 + r \Gamma_1 \Gamma_2 + r \psi_\frac{7}{2} \right) \Gamma_p \, , \nonumber
\end{align}
where for the last identity we recall that 
\begin{align}
\left(r \slashed{\nabla}_A tr \chi \right) \Gamma_p = \left(r \slashed{\nabla}_A \left[tr \chi - tr \chi_\circ\right] \right) \Gamma_p = \left( \mathfrak{D} \Gamma_2 \right)  \Gamma_p \, .
\end{align}
Using these formulae and Lemma \ref{comlem}, the commutation formula follows.
\end{proof}
\subsubsection{The role of the redshift under commutation} \label{sec:redshiftrole}
A final remark about the general ``size'' and the signs of the terms on the right hand sides of Propositions \ref{comsum} and \ref{prop:nscommute} is in order. As observed in the proof of Proposition  \ref{comsum}, commutation of (\ref{bic}) with $\slashed{\nabla}_3$ produces a linear term of the form $\hat{\omega}_\circ \slashed{\nabla}_3 \uppsi^\prime_{p^\prime}$ with $\hat{\omega}_\circ=\frac{1}{2M}>0$ because of the boxed ``redshift''-term in Lemma \ref{comlem}. This implies that after $k$ commutations of (\ref{bic}) with $\slashed{\nabla}_3$, there will be a linear term of the form $k \cdot \hat{\omega}_\circ  \slashed{\nabla}_3^k \uppsi^\prime_{p^\prime}$ on the right hand side. In other words, the coefficient of this linear term grows linearly with the number of derivatives taken and it has the wrong sign when integrating towards the past (i.e.~it drives exponential growth). For the analysis, this amplified redshift implies that the exponential decay rate we can impose will depend on (i.e.~grow with) the number of commutations. 
\subsection{The auxiliary equations} \label{auxquant}
In addition to our collection $\Gamma$ of metric- and Ricci-coefficients (cf.~Section \ref{sec:Gamdef}), we introduce explicitly the collection $G=\{G_1,G_2\}$ with members
\[
G_1 = r \slashed{\nabla}_A \left( \slashed{g}_{BC} - \slashed{g}^\circ_{BC} \right) = - r \slashed{\nabla}_A \left(  \slashed{g}^\circ_{BC} \right)  \, ,
\]
\[
G_2 = r \slashed{\nabla}_A \left(tr \underline{\chi} - tr \underline{\chi}_\circ \right) = r \slashed{\nabla}_A \left( tr \underline{\chi}\right) \, .
\]
From (\ref{mei}) we derive
\[
\slashed{\nabla}_3 \left(r \slashed{\nabla}_A \slashed{g}^\circ_{BC} \right) = \sum_{p_1+p_2 \geq p} \left[f \left(\Gamma + \Gamma \Gamma \right)\right]_{p_1} \left(G_{p_2} + \Gamma_{p_2}\right) + 2 r \slashed{\nabla}_A \hat{\underline{\chi}}_{BC} + \Gamma_1 \cdot r \slashed{\nabla} \hat{\underline{\chi}}
\]
and from (\ref{trx3}) 
\[
\slashed{\nabla}_3 \left(r \slashed{\nabla}_A tr \chi\right) =  \sum_{p_1+p_2 \geq p} \left[f \left(\Gamma + \Gamma \Gamma \right)\right]_{p_1} \left(G_{p_2} + \Gamma_{p_2}\right)  - 2  \hat{\underline{\chi}} \cdot
 r \slashed{\nabla}  \hat{\underline{\chi}} \, .
\]
Note that
\begin{align}
\slashed{\nabla}_C \slashed{g}^\circ_{AB} = \left(\slashed{\Gamma}^D_{\phantom{D}CA} - (\slashed{\Gamma}^\circ)^D_{\phantom{D}CA} \right) \slashed{g}^\circ_{DB} +  \left(\slashed{\Gamma}^D_{\phantom{D}CB} - (\slashed{\Gamma}^\circ)^D_{\phantom{D}CB} \right) \slashed{g}^\circ_{DA}
\end{align}
and therefore
\begin{align} \label{chrel}
 c \| \slashed{\nabla}_C \slashed{g}^\circ_{AB} \|_{\slashed{g}^\circ}  \leq \|\slashed{\Gamma}^D_{\phantom{D}CA} - (\slashed{\Gamma}^\circ)^D_{\phantom{D}CA}  \|_{\slashed{g}^\circ} \leq C \| \slashed{\nabla}_C \slashed{g}^\circ_{AB} \|_{\slashed{g}^\circ}  \, .
\end{align}
In view of this, we will sometimes identify $G_1 \equiv r \left(\slashed{\Gamma}^D_{\phantom{D}CA} - (\slashed{\Gamma}^\circ)^D_{\phantom{D}CA}\right)$.  The reason for these auxiliary quantities will become apparent when we address in Section \ref{sec:convergence} the issue of convergence of the solutions. 

\section{The norms} \label{sec:norms}  \label{sec:normse}

Recall $\mathfrak{D} = \{ M\slashed{\nabla}_3, r \slashed{\nabla}_4, r \slashed{\nabla} \}$ and  $\mathfrak{D}_{\nearrow} = \{ r \slashed{\nabla}_4, r \slashed{\nabla} \}$ , $\mathfrak{D}_{\nwarrow} = \{  M\slashed{\nabla}_3, r \slashed{\nabla}\}$. Define the dimensionless weight
\begin{align} \label{wdefn}
w = \frac{r}{2M} \,.
\end{align}

Given a spacetime manifold $\left(\mathcal{M}\left(\tau_1,\tau_2,v_\infty\right), g\right)$ of the form (\ref{maing}), we now define norms for the metric and Ricci coefficients on the one hand, and norms for the curvature components on the other. The norms of the former will simply be $L^2$-norms on the spheres $S^2_{u,v}$ given by (\ref{l2gnorm}): 
\[
\| \mathfrak{D}^i \Gamma_p \|^2_{L^2\left(S^2_{u,v}\right)} = \int_{S^2_{u,v}} \| \mathfrak{D}^i \Gamma_p \|^2 \sqrt{\slashed{g}} d\theta^1 d\theta^2 \, .
\]
Note that this norm is expected to decay like $r^{-2p+2}$ near infinity. The norms for the curvature will be $L^2$-norms on the $3$-dimensional slices of constant $\tau$ and on null-hypersurfaces of constant $u$ and constant $v$ respectively. We define the $k^{th}$-order curvature flux
\begin{align}
F\left[\mathfrak{D}^k{\Psi}\right] \left( \{u\} \times [v_1,v_2] \right) = 2\sum_{i=0}^k \sum_{i-perms}
 \int_{v_1}^{v_2} d\bar{v} \Bigg[ \int_{S^2\left(u,v\right)} \sqrt{\slashed{g}}\nonumber \\ \left\{ w^5 \| \mathfrak{D}^i \alpha \|^2 + w^4 \|\mathfrak{D}^i\beta\|^2 + w^2 \left[ \|\mathfrak{D}^i\rho - \mathfrak{D}^i\rho_\circ\|^2 +\|\mathfrak{D}^i\sigma\|^2 \right] +2 \| \mathfrak{D}^i \underline{\beta}\|^2 \right\} \Bigg] \left(u,\bar{v}\right) \nonumber \, ,
\end{align}
where for each $i$ we sum over all possible combinations of $i^{th}$ derivatives. Here it is implicit that $  \{u\} \times [v_1,v_2] \times S^2_{u,v} \subset \mathcal{M}\left(\tau_1,\tau_2,v_\infty\right)$. Similarly,
\begin{align}
F\left[\mathfrak{D}^k{\Psi}\right]\left( \left[u_1,u_2\right] \times \{v \} \right) =  2\sum_{i=0}^k \sum_{i-perms}
 \int_{u_1}^{u_2} d\bar{u}\ \Omega^2 \left(\bar{u},v\right)\Bigg[ \int_{S^2\left(u,v\right)} \sqrt{\slashed{g}}   \nonumber \\ 
 \left\{ 2w^5 \|\mathfrak{D}^i\beta\|^2 + w^4\left[ \|\mathfrak{D}^i\rho - \mathfrak{D}^i\rho_\circ\|^2 +\|\mathfrak{D}^i\sigma\|^2 \right] + w^2 \| \mathfrak{D}^i \underline{\beta}\|^2 + \|\mathfrak{D}^i \underline{\alpha}\|^2 \right\} \Bigg] \, , \nonumber
\end{align}
for $ \left[u_1,u_2\right] \times \{v \} \times S^2 \subset \mathcal{M}\left(\tau_1,\tau_2,v_\infty\right)$, as well as the energy
\begin{align}
\mathcal{E} \left[\mathfrak{D}^k \Psi\right] \left(\tau,r_1,r_2\right) = \frac{1}{4}\sum_{i=0}^k \sum_{i-perms}
 \int_{r_1}^{r_2} d\bar{r} \int_{S^2\left(u\left(\tau,r\right),v\left(\tau,r\right))\right)} \sqrt{\slashed{g}} 
 \Big\{ w^5 \| \mathfrak{D}^i \alpha \|^2  \nonumber \\
+ w^4 \|\mathfrak{D}^i \beta\|^2 + w^{4-h} \left[ \|\mathfrak{D}^i\rho - \mathfrak{D}^i\rho_\circ\|^2 +\|\mathfrak{D}^i\sigma\|^2 \right] + w^{2-h} \|\mathfrak{D}^i \underline{\beta}\|^2 +w^{-h} \|\mathfrak{D}^i \underline{\alpha}\|^2  \Big\} \,  \nonumber
\end{align}
for $\tau_1 \leq \tau \leq \tau_2$. The factor $\frac{1}{4}$ will be convenient and can be traced back to the estimates in Section \ref{sec:slices}. Recall that $1<h \leq 2$ was fixed in Section \ref{sec:manifold}.
\section{Initial Data and Local Evolution} \label{sec:local}
Consider a $\Sigma_{\tau_0}$ for some $\tau_0$ large to be determined later and its associated region $\mathcal{M} \left(\tau_0,\tau\right)$ for any $\tau > \tau_0$ large. Let $u_\tau = \lim_{v \rightarrow \infty} u \left(\tau,v\right)$ be the $u$-value where the constant $\tau$ hypersurface ``intersects'' $v=\infty$. Similarly, we denote  $v_\tau = \lim_{u \rightarrow \infty} v \left(\tau,u\right)$ the $v$ value of intersection with the horizon. 

We will prescribe scattering data on the hypersurface $\tilde{\mathcal{C}}_{U=0}:=\{U=0\} \times \left[v_{\tau_0},\infty\right) \times S^2$  and on the limiting hypersurface $\tilde{\mathcal{C}}_{v=\infty}:= \left[u_{\tau_0},\infty\right) \times \{v=\infty\} \times S^2$. To make the latter precise, we will define below first a notion of data at null infinity, and second a notion of associated finite data on a hypersurface $\tilde{\mathcal{C}}_{v=v_\infty}:= \left[u_{\tau_0},\infty\right) \times \{v=v_\infty\} \times S^2$, which converges to the desired asymptotic data in the limit $v_\infty \rightarrow \infty$. 

Following \cite{formationofbh}, let us now fix two stereographic coordinate patches $\left(\theta^1, \theta^2\right)$ and $\left(\underline{\theta}^1, \underline{\theta}^2\right)$ on the sphere $S^2_{U=0,v_\tau}$ and Lie-transport these coordinates along $\tilde{\mathcal{C}}_{U=0}$ using $\partial_v \theta^A=0=\partial_v \underline{\theta}^A$, $A=1,2$. (As in \cite{formationofbh}, this choice of coordinates will be particularly convenient when prescribing the seed metric and changing from one chart to another.) The notation $\theta$ will always refer to these particular coordinates. Note that the $S^2_{U=0,v}$-tensor $\slashed{g}$ has components $\slashed{g}_{AB}=\slashed{g}_{\theta^A \theta^B}= \slashed{g} \left(\partial_{\theta^A}, \partial_{\theta^B} \right)$ in the associated coordinate basis. Recall also that $\gamma_{AB}$ denotes the standard metric on the unit sphere.

\subsection{Scattering data sets}
\begin{definition} \label{def:scatdat}
A scattering data set of exponential decay $\tilde{P} \in \mathbb{R}^+$  and exponential regularity $(I \geq 5,J \geq 3) \in \mathbb{N} \times \mathbb{N}$ is given by the following:
\begin{enumerate}
\item A smooth ``seed" prescribed on  $\tilde{\mathcal{C}}_{U=0}$, obtained by prescribing freely a smooth symmetric $S^2_{U=0,v}$-tensor density of weight $-1$, $\hat{\slashed{g}}^{dat\mathcal{H}}_{AB} \left(v, \theta^1, \theta^2\right)$, along $\tilde{\mathcal{C}}_{U=0}$, whose determinant is equal to $1$ and which satisfies the following pointwise bounds on the components in the coordinate basis: For any $j=0, ... , J$ we have in each coordinate patch on the sphere
\begin{align} \label{hoz1}
 \sum_{i=0}^I \sum_{i_1+i_2=i} \Big| \left(\partial_v\right)^j \left(\partial_{\theta^1} \right)^{i_1} \left(\partial_{\theta^2} \right)^{i_2} \left[ \hat{\slashed{g}}^{dat\mathcal{H}}_{AB} - \frac{\gamma_{AB}}{\sqrt{\gamma}} \right] \Big| \leq \tilde{P}^j e^{-{\tilde{P}}\frac{v}{4M}} \, .
\end{align}
\item A smooth ``seed" function ``at infinity'' obtained by prescribing freely a symmetric $2$-tensor $\hat{\slashed{g}}_{AB}^{dat \mathcal{I}} \left(u, \theta^1, \theta^2\right)$, along $ \left[u_{\tau_0}, \infty\right) \times S^2$, which satisfies $\gamma^{AB} \hat{\slashed{g}}_{AB}^{dat \mathcal{I}} = 0$ and for any $j=0, ..., J$ the estimate
\begin{align} \label{oest}
\sum_{i=0}^I  \sum_{i_1+i_2=i}  \Big| \left(\partial_u\right)^j  \left(\partial_{\theta^1} \right)^{i_1} \left(\partial_{\theta^2} \right)^{i_2} \left[ \hat{\slashed{g}}^{dat \mathcal{I}}_{AB} \right] \Big| \leq \tilde{P}^j  e^{-\tilde{P} \frac{u}{4M}} \, 
\end{align}
in each coordinate patch on the sphere.\footnote{It is implicit that we fixed stereographic coordinates on one sphere and Lie-transported them using $\partial_u \theta^A=0$.}
\end{enumerate}
A scattering initial data set will be denoted by the tuple $\left(\hat{\slashed{g}}^{dat\mathcal{H}},   \hat{\slashed{g}}^{dat\mathcal{I}}, \tilde{P},I, J\right)$.
\end{definition}

The reason that we specify a tensor \emph{density} along the horizon and a symmetric traceless tensor at infinity will become clear below. The point is that  $\hat{g}_{AB}^{data\mathcal{I}}$ is a limiting object from which we will construct a family of tensor densities defined on finite null hypersurfaces approaching null infinity $\mathcal{I}^+$.
\subsection{Associated finite scattering data sets}

We will now construct associated finite scattering data sets. 

Before we present the details, let us provide a brief overview. We will pick a slice $\Sigma_{\tau}$ for $\tau>\tau_0+1$ large and a constant $v$ hypersurface $v=v_{\infty}$. On $\Sigma_{\tau}$ the data will be assumed to be trivially Schwarzschild, while on the horizon we will cut-off the seed function to match it smoothly with the Schwarzschild data on the slice $\Sigma_\tau$. On $v=v_\infty$ we will specify again the conformal class of the metric with the deviation from the round metric given to first order in $\frac{1}{r}$ by the data $\hat{\slashed{g}}_{AB}^{dat\mathcal{I}}$.\footnote{The reason for the $r$-weight appearing in the approximate data can be explained as follows: To ensure that the ($r$-independent!) $\hat{\slashed{g}}^{dat\mathcal{I}}_{AB}$ of Definition \ref{def:scatdat} indeed encodes the non-trivial gravitational (radiation) field along null-infinity, the finite data on $v=v_\infty$ has to be $r$-weighted appropriately, so as to produce the correct limit towards null-infinity.}

\begin{figure}[h!]
\[
\input{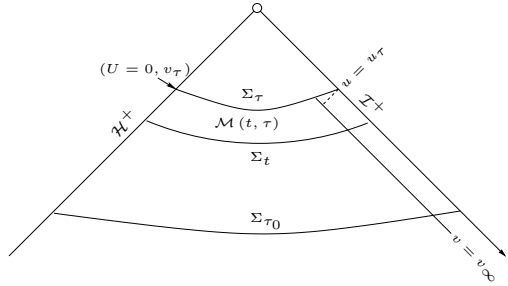}
\]
\caption{The scattering problem} \label{figure1}
\end{figure}

Let now $\Sigma_{\tau}$ be given and $\chi_{\tilde{x}} \left(x\right)$ be a smooth positive interpolating function which is equal to $1$ for $x\leq \tilde{x}-1$ and equal to zero for $x \geq \tilde{x}$. We define the following cut-off version of the data on null-infinity
\begin{align}
\hat{\slashed{g}}^{dat \mathcal{I}_\tau}_{AB} = \chi_{u_\tau} \left(u\right) \hat{\slashed{g}}^{dat \mathcal{I}}_{AB} \, .
\end{align}
Using $\hat{\slashed{g}}^{dat \mathcal{I}_\tau}_{AB}$, we now prescribe smoothly the conformal class of the metric along $\tilde{\mathcal{C}}_{v_\infty}$ by specifying a tensor-density of weight $-1$, $\hat{\slashed{g}}_{AB}^{dat(v_\infty)_\tau}$, having unit determinant and satisfying the following pointwise bounds on the components in coordinates: For $j=0,...,J$
\begin{align} \label{assodata}
\sum_{i=0}^I \Big| \left(\partial_u\right)^j \left(\partial_\theta\right)^i \left[\hat{\slashed{g}}_{AB}^{dat(v_\infty)_\tau} \left(u, \theta\right)  - \frac{\gamma_{AB}}{\sqrt{\gamma}} \left(\theta\right) - \frac{1}{ r\left(u,v_\infty\right)} \frac{1}{\sqrt{\gamma}} \hat{\slashed{g}}^{dat\mathcal{I}_\tau}_{AB} \left(u,\theta\right) \right] \Big| \nonumber \\
\leq \frac{1}{r^2\left(u,v_\infty\right)}\tilde{P}^j  e^{-\tilde{P} \frac{u}{4M}} 
 \ \ \ \ \ \ \textrm{for $u_{\tau_0} \leq u$} \, ,
\end{align}
where we have used the short hand notation $\left(\partial_\theta\right)^i$ for $\sum_{i_1+i_2=i}  \left(\partial_{\theta^1} \right)^{i_1} \left(\partial_{\theta^2} \right)^{i_2}$.

Let us explain (\ref{assodata}). Naively, we would simply define the tensor density  $\hat{\slashed{g}}_{AB}^{dat(v_\infty)_\tau}= \frac{\gamma_{AB}}{\sqrt{\gamma}} \left(\theta\right) + \frac{1}{ r\left(u,v_\infty\right)} \frac{1}{\sqrt{\gamma}} \hat{\slashed{g}}^{dat\mathcal{I}_\tau}_{AB}\left(u,\theta\right)$. However, this tensor density does not have determinant equal to 1. Using that $\gamma^{AB} \hat{\slashed{g}}^{dat\mathcal{I}}_{AB}=0$, one sees that the expression \emph{does} have unit determinant to order $\frac{1}{r}$ and hence that the correction to achieve unit determinant for $\hat{\slashed{g}}_{AB}^{dat(v_\infty)_\tau}$ is indeed of order $\frac{1}{r^2}$. This is what the error on the right hand side of (\ref{assodata}) accounts for. Existence of suitable $\hat{\slashed{g}}_{AB}^{dat(v_\infty)_\tau}$ satisfying (\ref{assodata}) follows from the implicit function theorem by choosing $v_\infty$ sufficiently large.

\begin{remark}
Note that the $\hat{\slashed{g}}_{AB}^{dat(v_\infty)_\tau}$ are not uniquely determined by the scattering data $\hat{\slashed{g}}^{dat\mathcal{I}_\tau}_{AB}$. However, they \emph{are} uniquely determined up to order $\frac{1}{r}$, which completely captures the non-trivial part of the data in the limit as $v_\infty \rightarrow \infty$ as we will see below.
\end{remark}

Similarly, on the horizon we define the following cut-off version of the data:
 For a given $\tau \geq \tau_0+1$ we define
\[
\hat{\slashed{g}}^{dat\mathcal{H}_\tau}_{AB} =  \chi_{v_\tau} \left(v\right) \cdot \hat{\slashed{g}}^{dat\mathcal{H}}_{AB} + \left(1- \chi_{v_\tau} \left(v\right) \right) \hat{\gamma}_{AB}  \ \ \ \ , \ \ \ \hat{\gamma}_{AB} = \frac{1}{\sqrt{\gamma}} \gamma_{AB}
\]
with $v_\tau$ determined by $\tau$ as explained above.\footnote{Recall that $\Sigma_\tau$ determines the quantities $u_{\tau}$ and $v_{\tau}$ as the coordinates of the sphere where $\Sigma_\tau$ intersects the horizon and null infinity respectively. See Figure \ref{figure1}.}  Allowing a different function $ \chi_{v_\tau} \left(v\right)$ for different components, we can arrange that the determinant of $\hat{\slashed{g}}^{dat\mathcal{H}_\tau}_{AB}$ always remains equal to one. 
Moreover, the $\hat{\slashed{g}}^{dat\mathcal{H}_\tau}_{AB}$ also satisfy (\ref{hoz1}) uniformly in $\tau$, provided we allow a constant factor in front of the exponential. 

We finally define a family of associated finite scattering data sets, $D_{\tau,v_\infty}$, as follows:
\begin{definition} \label{def:finitedata}
Given  $\tau \geq \tau_0+1$ and $v_\infty>\tau$, a smooth finite scattering data set associated with a scattering data set $\left(\hat{\slashed{g}}^{dat\mathcal{H}},   \hat{\slashed{g}}^{dat\mathcal{I}}, \tilde{P}, I, J\right)$ is denoted $D_{\tau, v_\infty}$ and consists of the following: 
\begin{itemize}
\item The hypersurface $\Sigma_\tau$, on which the data is exactly Schwarzschildean\footnote{that is to say the data induced on $\Sigma_\tau$ if the metric on $\mathcal{M}$ were exactly Schwarzschild of mass $M$}
\item the hypersurface $\tilde{\mathcal{C}}_{U=0}$ along which the $S^2_{U=0,v}$-tensor density \newline $\hat{\slashed{g}}^{dat\mathcal{H}_\tau}_{AB} \left(v, \theta^1, \theta^2\right)$ is prescribed and satisfies (\ref{hoz1})
\item the hypersurface $\tilde{\mathcal{C}}_{v=v_\infty}$ where $\hat{\slashed{g}}^{dat(v_\infty)_\tau}_{AB} \left(u,\theta^1,\theta^2\right)$ is prescribed and satisfies (\ref{assodata}). 
\end{itemize}
\end{definition}
In the remainder of this section, we outline how to construct all geometric quantities from a smooth finite scattering data set. We follow Christodoulou \cite{formationofbh} closely. Figure \ref{figure1} may be helpful.

We will allow ourselves the following abuse of notation: We will write  $\Omega_\mathcal{K}$, $\eta$, $\beta$, etc.~to denote both the geometric quantities defined from the seed data and the gauge fixing  in the following section (which are functions on the horizon and the hypersurface $v=v_\infty$ only) \emph{and} the spacetime solution arising from this data.
\subsection{Fixing the gauge} \label{sec:gaugefix}
In order to construct the geometry of the cones $\tilde{\mathcal{C}}_{U=0}$ and $\tilde{\mathcal{C}}_{v=\infty}$  from the seeds $\hat{\slashed{g}}^{dat\mathcal{H}_\tau}_{AB}$ and $\hat{\slashed{g}}^{dat\left(v_\infty\right)_\tau}_{AB}$ (by which we simply mean constructing all $\Gamma$ and $\psi$ with some number of their tangential and transversal derivatives), we need to fix a gauge for the lapse $\Omega$ along the respective cones.

On the horizon, following \cite{formationofbh}, we set $\Omega_{\mathcal{K}}=2M \sqrt{e^{-1}}$ on $\tilde{\mathcal{C}}_{U=0}$, which implies $\hat{\omega}_{\mathcal{K}}=0$ there, cf.~(\ref{Kom}). Note that this corresponds to $\frac{\Omega^2_\circ}{\Omega^2}=1$ in Eddington-Finkelstein variables and hence $\hat{\omega}=\hat{\omega}_\circ$. Also $\eta = -\underline{\eta} = \zeta$ on $\tilde{\mathcal{C}}_{U=0}$.

Along the hypersurface $\tilde{\mathcal{C}}_{v=v_\infty}$ we set $\Omega^2=\Omega_\circ^2$ so that again $\eta=-\underline{\eta}=\zeta$ on $\tilde{\mathcal{C}}_{v=v_\infty}$. \newline

Finally, in accord with the fact that the data are supposed to be trivial for $v \geq v_\tau$ on the horizon and for $u \geq u_\tau$ on $v=v_\infty$, we impose the following:

\begin{enumerate}
\item On the spheres $S^2_{U=0,v}$ with $v \geq v_\tau$, we prescribe the full metric $\slashed{g}^{dat\mathcal{H}_\tau}_{AB}$ to equal the round metric on a sphere of radius $r=2M$.  That is to say, setting $\slashed{g}^{dat\mathcal{H}_\tau}_{AB}=\Phi^2  \hat{\slashed{g}}^{dat\mathcal{H}_\tau}_{AB}$ (i.e.~$\Phi^2 = \sqrt{\det \slashed{g}^{dat\mathcal{H}_\tau}_{AB} }$) along $\tilde{\mathcal{C}}_{U=0}$, we impose that for $i=0,1,\dots,I$,
\begin{align} \label{bcphi}
 \left[ \left(\partial_\theta\right)^i \log \left(\frac{\Phi^2}{4M^2 \sqrt{\gamma}} \right) \right] \left(U=0,v\right) = 0  \textrm{ \ \ \ for $v\geq v_\tau$ }\, .
\end{align}
In addition, we impose that $tr \chi = tr \chi_\circ=0$, $\zeta = 0$ and $tr \underline{\chi}= tr \underline{\chi}_\circ=-M^{-1}$ holds on these spheres.\footnote{It suffices to impose this on one sphere and to use the equations to infer it for all others.} With these definitions, the horizon is exactly Schwarzschildean for $v \geq v_\tau$. In particular, all $\Gamma$ and all $\psi$ vanish identically.
\item Similarly, for the spheres $S^2_{u,v_\infty}$ with $u \geq u_\tau$, we impose that their metric $\slashed{g}^{dat(v_\infty)_\tau}_{AB}=\Phi^2  \hat{\slashed{g}}^{dat(v_\infty)_\tau}_{AB}$ equals the round metric of radius $r\left(u,v_\infty\right)$:
\[
\left(\partial_\theta\right)^i  \log \left(\frac{\Phi^2}{r^2 \sqrt{\gamma}} \right) \left(u,v_\infty\right)= 0 \qquad \textrm{for $u\geq u_\tau$ and $i=0,\ldots I$} \, .
\]
Imposing also $tr \chi = tr \chi_\circ=0$, $\zeta = 0$ and $tr \underline{\chi}= tr \underline{\chi}_\circ=-M^{-1}$ one again concludes that the metric is exactly Schwarzschildean  for all $u \geq u_\tau$ along $v=v_\infty$.
\end{enumerate}

\subsection{Determining the geometry: The horizon} \label{sec:hozco}  \label{sec:redsrole}

Given the seeds $\hat{\slashed{g}}^{dat\mathcal{H}_\tau}_{AB}$ and $\hat{\slashed{g}}^{dat\left(v_
\infty\right)_\tau}_{AB}$ and the gauge-choice for $\Omega$ made in section \ref{sec:gaugefix}, we can construct and derive estimates for all geometric quantities on the respective cones using the Einstein (constraint) equations on the cones. By the choice of gauge, only ODEs will have to be solved in this process. 

\begin{proposition} \label{prop:construction}
Given a smooth finite data set $D_{\tau,v_\infty}$ associated with a scattering data set $\left(\hat{\slashed{g}}^{dat\mathcal{H}},   \hat{\slashed{g}}^{dat\mathcal{I}}, \tilde{P}, I, J\right)$ as in Definition \ref{def:finitedata}, and given the gauge of Section \ref{sec:gaugefix}, there exists a unique smooth solution for all $\Gamma$ and $\psi$ along $\tilde{\mathcal{C}}_{U=0}$, such that the Bianchi equations (\ref{bi4}) and the null-structure equations (\ref{ns4}), (\ref{4uhatc}), (\ref{4utrc}) and (\ref{Gauss})-(\ref{elliptic2}) hold along $\tilde{\mathcal{C}}_{U=0}$. Moreover, any $Q \in \{\Gamma, \psi \}$ satisfies the estimate
\begin{align}  \label{beisp}
\sum_{i=0}^{\tilde{I}}    \Big|  \left(\partial_v\right)^j \left(\partial_\theta\right)^i \ \left[Q \right] \left(v,\theta\right) \Big| \leq C \cdot \tilde{P}^{2+j} e^{-\tilde{P}\frac{v}{4M}} \ \ \ \textrm{for $j=0, ... , \tilde{J}$} 
\end{align}
along  $\tilde{\mathcal{C}}_{U=0}$ for some $\tilde{I}$ and $\tilde{J}$ (depending on $I$ and $J$) and a constant $C$ which can be made uniform\footnote{i.e.~not depending on the cut-off $v_\tau$ of the finite data set $D_{\tau,v_\infty}$}, by chosing $\tilde{P}$ sufficiently large (depending on $M$, $I$ and $J$). Finally, the $\tilde{I}$ and $\tilde{J}$ can be read off for each quantity $Q$ in terms of $I$ and $J$ from the table below.
\end{proposition}
\begin{tabular}{ | c || c | c | c | c | c |  c | c | c | c | c| c|}
  \hline                        
  $\phantom{X}$ & $\hat{\slashed{g}}^{dataH}- \frac{\gamma}{\sqrt{\gamma}}$ & $\log \frac{\Phi^2}{4M^2 \sqrt{\gamma}}$ & $\hat{\chi}$, $tr \chi - tr \chi_\circ$ & $\zeta$, $\eta$, $\underline{\eta}$ & $tr \underline{\chi}-tr \underline{\chi}_\circ$, $\underline{\hat{\chi}}$  \\ \hline
  $\tilde{I}$ & $I$ & $I$ & $I$ & $I-1$ & $I-2$  \\ \hline
  $\tilde{J}$ & $J$ & $J+1$ & $J-1$ & $J$ & $J$ \\ \hline  
\end{tabular}
 $\phantom{X}$ \\ \\ \\ \phantom{x..}
\begin{tabular}{ | c || c | c | c | c | c |  c | c | c | c | c| c|}
  \hline                        
  $\phantom{X}$ &  $\alpha$ & $\beta$ & $\rho-\rho_\circ, \sigma$ & $\underline{\beta}$ & $\underline{\alpha}$ \\ \hline
  $\tilde{I}$ & $I$ & $I-1$ & $I-2$ & $I-3$ & $I-4$  \\ \hline
  $\tilde{J}$ & $J-2$ & $J-1$ & $J-1$ & $J$ & $J$  \\
  \hline  
\end{tabular}

\begin{remark}
The proof of the above proposition as well as that of Proposition \ref{prop:construction1a} will provide a first taste (purely at the level of the data) of the obstructions resulting from the red-shift effect on the horizon. Further important obstructions will appear when trying to propagate the decay imposed on the data to the spacetime, as we will see in Section \ref{sec:bootstrap}. Cf.~also the remarks following Lemma \ref{comlem}.
\end{remark}

\begin{proof}
One derives the following equation for the conformal factor $\Phi$ along $\tilde{\mathcal{C}}_{U=0}$:
\begin{align} \label{phiode}
\frac{\partial^2 \Phi}{\partial v^2} - \frac{1}{2M} \partial_v \Phi + \frac{\Phi}{8} \left(\hat{\slashed{g}}^{-1}_{dat\mathcal{H}_t}\right)^{AC}\left(\hat{\slashed{g}}^{-1}_{dat\mathcal{H}_t}\right)^{BD} \left[ \partial_v \hat{\slashed{g}}^{dat\mathcal{H}_t}_{AB} \right]\left[ \partial_v \hat{\slashed{g}}^{dat\mathcal{H}_t}_{CD} \right] = 0 ,
\end{align}
where $\left[ \hat{\slashed{g}}^{-1}_{dat\mathcal{H}_t} \right]^{AB} \hat{\slashed{g}}^{dat\mathcal{H}_t}_{BC}= \delta^{A}_{\phantom{A}C}$ defines the components of $\hat{\slashed{g}}^{-1}_{dat\mathcal{H}_t}$. Using the boundary condition (\ref{bcphi}) and that $\Phi_{,v} = 0$ on $S^2_{U=0,v_\tau}$, one can integrate the ODE (\ref{phiode}) and obtain:
\begin{align}  \label{cosa}
\sum_{i=0}^I    \Big|  \left(\partial_v\right)^j \left(\partial_\theta\right)^i \ \left[\log \frac{\Phi^2}{4M^2 \sqrt{\gamma}} \right] \Big| \leq C \tilde{P}^j e^{-\tilde{P}\frac{v}{4M} \cdot 2} \ \ \ \textrm{for $j=0, ... , J+1$ .} 
\end{align}
From the equations
\begin{align} \label{trxdet}
tr \chi = \frac{2}{\Phi} \frac{\partial \Phi}{\partial v} \qquad , \qquad
\hat{\chi}_{AB} = \frac{1}{2} \Phi^2 \partial_v \hat{\slashed{g}}^{dat\mathcal{H}_t}_{AB} \, 
\end{align}
one reads off
\begin{align}  \label{cosa2}
\sum_{i=0}^I    \Big|  \left(\partial_v\right)^j \left(\partial_\theta\right)^i \ \left[ tr \chi \right] \Big| \leq C \tilde{P}^{j+1} e^{-\tilde{P}\frac{v}{4M} \cdot 2} \ \ \ \textrm{for $j=0, ... , J$} \ \ \ ;  \nonumber \\
\sum_{i=0}^I    \Big|  \left(\partial_v\right)^j \left(\partial_\theta\right)^i \ \left[ \hat{\chi}_{AB} \right] \Big| \leq C \tilde{P}^{j+1} e^{-\tilde{P}\frac{v}{4M}} \ \ \ \textrm{for $j=0, ... , J-1$.} 
\end{align}
Next one obtains an estimate for the components of $\eta$ using the fact that it vanishes on $S^2_{U=0,v_\tau}$ and the evolution equation 
\[
\partial_v \left( \eta_A\right) = - tr \chi \eta_A  + \left(\slashed{g}^{-1}\right)^{BC} \slashed{\nabla}_C \hat{{\chi}}_{BA} - \frac{1}{2} \slashed{\nabla}_A tr {\chi}  \, ,
\]
which is derived by combining (\ref{etaeq}) with (\ref{elliptic2}) and recalling $\eta=-\underline{\eta}=\zeta$ on $\tilde{\mathcal{C}}_{U=0}$. Integrating this ODE for $\eta$, one derives the exponential decay bound (\ref{beisp}) for $\eta$ from the previous bounds. Combining equations (\ref{4utrc}) and the Gauss equation (\ref{Gauss}) (observing that the Gauss curvature can be expressed in terms of quantities we already have estimates on), we obtain an ODE for $tr \underline{\chi} - tr \underline{\chi}_\circ$.
Similarly, for  $\hat{\underline{\chi}}$, which is governed by equation (\ref{4uhatc}), we obtain
\[
\partial_v \left[ \hat{\underline{\chi}}_{AB} \right] + \hat{\omega}_\circ \hat{\underline{\chi}}_{AB} = -2\slashed{\mathcal{D}}^\star \underline{\eta} - \frac{1}{2} tr \underline{\chi}_\circ \hat{\chi} + ... \, ,
\]
where ``..." indicates quadratic terms. We can write the two resulting ODEs as
\begin{align}
\partial_v \left[ e^{\frac{v}{2M}} \left(tr \underline{\chi} - tr \underline{\chi}_\circ\right) \right] = e^{\frac{v}{2M}} \Big[ 2 \slashed{div} \underline{\eta} - 2\left(K-K_\circ\right) +2 \underline{\eta} \cdot \underline{\eta} \nonumber \\ - tr \underline{\chi}_\circ \left[ \left(\hat{\omega}-\hat{\omega}_\circ \right) +  \left(tr \chi - tr \chi_\circ\right)\right] + \left(tr \underline{\chi} - tr \underline{\chi}_\circ\right) \left(\left(tr \chi - tr \chi_\circ \right) +  \left(\hat{\omega} - \hat{\omega}_0\right) \right) \Big] \nonumber \, ,
\end{align}
\begin{align} \label{ode2}
\partial_v \left[ e^{\frac{v}{2M}} \hat{\underline{\chi}}_{AB} \right] = e^{\frac{v}{2M}}\Big[ -2\slashed{\mathcal{D}}^\star \underline{\eta} - \frac{1}{2} tr \underline{\chi}_\circ \hat{\chi}  - \frac{1}{2} \left( tr \underline{\chi}- tr \underline{\chi}_\circ \right) \hat{\chi}+ ...  \Big]\, ,
\end{align}
and from those -- using the exponential estimates already available -- one proves the exponential decay bound (\ref{beisp}) for both $tr \underline{\chi} - tr \underline{\chi}_\circ$ and $\hat{\underline{\chi}}$. Note that the integrating factor in the above ODEs blows up exponentially as $v \rightarrow \infty$; this forces the exponential decay rate imposed on the quantities in the square brackets on the right hand side to be \emph{stronger} than the surface gravity $\kappa=\hat{\omega}_0 \left(\mathcal{H}\right)=\frac{1}{2M}$, i.e.~$\tilde{P}>1$, in order for the right hand side of the above ODEs to be integrable in absolute value from infinity. This is the celebrated redshift effect, which since we are integrating backwards appears as a blueshift.

With the constraint on the constant $\tilde{P}$ imposed by the surface gravity of the horizon being understood, the estimate (\ref{beisp}) is also obtained for all curvature components using the Bianchi- and null-structure equations, as is carried out in complete detail in \cite{formationofbh, LukRod}.
\end{proof}

\begin{remark}
It would be interesting to find systematic ways of constructing data with seed functions not decaying exponentially by exploiting cancellations in the round bracket on the right hand side of (\ref{ode2}). However, it should be clear that even if this can be achieved (for instance, trivially, by imposing the data on the horizon to be Schwarzschild), it is not automatic that one can solve backwards given non-trivial polynomially decaying data on $\mathcal{I}^+$, precisely because of the aforementioned additional obstructions when propagating the decay.
\end{remark}

In addition to (\ref{beisp}), we can also define and estimate the \emph{transversal} derivatives\footnote{We record the slight abuse of language: The objects $\slashed{\nabla}^k_3 Q$ are defined by imposing the commuted null-structure equations. Only once these objects are propagated into the spacetime do they acquire their interpretation as \emph{transversal} derivatives.} $\slashed{\nabla}^k_3 Q$ of any quantity $Q$ using the commuted null-structure and Bianchi equations and the fact that $Q=0$ on $S^2_{U=0,v_\tau}$. This leads to the estimates (\ref{beispiel}) below, which we have conveniently stated for the $\slashed{g}_{\circ}$-norms.\footnote{Note that in view of $|\slashed{g}^\circ_{AB}| \leq C r^2$ and $|\slashed{g}_\circ^{AB}| \leq C r^{-2}$, lifting and lowering indices in general introduces $r$-weights for the components of an $S^2_{u,v}$-tensor. On the horizon, however, $r$ is bounded and hence (\ref{beispiel}) also holds pointwise for the individual components of $Q$. This will be different when we discuss the cone $v=v_\infty$.}

\begin{proposition} \label{prop:construction1a}
Under the assumptions of Proposition \ref{prop:construction}, there exist $\hat{I}$, $\hat{J}$ and $\hat{K}$ depending only on $I$ and $J$ such that for any $k=0,...,\hat{K}$
\begin{align}  \label{beispiel}
\sum_{i=0}^{\hat{I}}    \Big\| \slashed{\nabla}^k_3 \left(r \slashed{\nabla}_4\right)^j \left(r \slashed{\nabla}\right)^i \ \left[Q \right] \left(v,\theta\right) \Big\|_{\slashed{g}_\circ} \leq C_{\hat{I},\hat{J}, \hat{K}} \cdot \tilde{P}^{2+j} e^{-\tilde{P}\frac{v}{4M}} \ \ \ \textrm{for $j=0, ... , \hat{J}$} 
\end{align}
holds along $\tilde{\mathcal{C}}_{U=0}$. The constant $C_{\hat{I},\hat{J}, \hat{K}}$ can be made uniform (in $v_\tau$) by choosing $\tilde{P}$ large depending on $M$ and on $\hat{K}$. Finally, $\hat{I}$, $\hat{J}$ and $\hat{K}$ can all be made large by choosing $I$, $J$ and $\tilde{P}$ large.
\end{proposition}

\begin{proof}
This follows by patiently integrating the ODEs with trivial (future) data arising from the commuted null-structure and Bianchi equations. While this is straightforward, we point out an important structure: By Lemma \ref{comlem}, every commutation of a $\slashed{\nabla}_4$-equation with $\slashed{\nabla}_3$ produces a positive \emph{linear} term $\omega_\circ \slashed{\nabla}_3$. Indeed, ignoring lower order terms, if 
\[
\slashed{\nabla}_4 Q = RHS
\]
then
\[
\slashed{\nabla}_4 \left(\slashed{\nabla}_3^n Q\right) + n \cdot \omega_\circ  \left(\slashed{\nabla}_3^n Q\right)  =  \slashed{\nabla}_3^n \left( RHS \right)
\]
and hence
\[
\partial_v \left( e^{n \frac{ v}{2M}}  \slashed{\nabla}_3^n Q\right) =  e^{n \frac{ v}{2M}} \cdot \left(\slashed{\nabla}_3^n RHS\right) \, ,
\]
which means that the right hand side needs to decay stronger in $v$ the more transversal derivatives are taken. This explains why $\tilde{P}$ has to be chosen large depending on $\hat{K}$ to make the constant in (\ref{beispiel}) uniform. This ``amplified redshift" under commutation was observed for the wave equation in \cite{mihalisnotes} and is seen here as an ``amplified blueshift" since we are integrating towards the past.
\end{proof}

\subsection{Determining the geometry: The hypersurface $v=v_\infty$} \label{sec:radfields}

For the characteristic data along the surface $v=v_\infty$ we can follow an analogous procedure. However, in view of the limiting procedure applied later ($v_\infty \rightarrow \infty$), weights in the variable $r$ are now important.

\begin{proposition} \label{prop:construction2}
Given a smooth finite data set $D_{\tau,v_\infty}$ associated with a scattering data set $\left(\hat{\slashed{g}}^{dat\mathcal{H}},   \hat{\slashed{g}}^{dat\mathcal{I}}, \tilde{P}, I, J\right)$ as in Definition \ref{def:finitedata}, and given the gauge of Section \ref{sec:gaugefix}, there exists a unique smooth solution for all $\Gamma$ and $\psi$ along $\tilde{\mathcal{C}}_{v=v_\infty}$ such that the Bianchi equations (\ref{bi3}) and the null-structure equations (\ref{ns3}), (\ref{3hatc}), (\ref{3trc}) and (\ref{Gauss})-(\ref{elliptic2}) hold along $\tilde{\mathcal{C}}_{v=v_\infty}$. Moreover, (using commutation) one can define  uniquely and smoothly the transversal derivatives of any $\Gamma$ and $\psi$ along $\tilde{\mathcal{C}}_{v=v_\infty}$ such that in addition the remaining Bianchi (\ref{bi4}) and null-structure equations (\ref{ns4}), (\ref{4uhatc}), (\ref{4utrc}) hold. Finally, there exists $\hat{I}$, $\hat{J}$ and $\hat{K}$ (depending only on $I$ and $J$) such that any $Q_p \in \{\Gamma_p, \psi_p \}$ satisfies the estimate
\begin{align}  \label{beisp2}
\sum_{i=0}^{\hat{I}} r^p   \Big\| \left(r \slashed{\nabla}_4\right)^k \slashed{\nabla}_3^j \left( r \slashed{\nabla}^i \right) \left[ Q_p \right] \left(v,\theta\right) \Big\|_{\slashed{g}_\circ} \leq C \cdot \tilde{P}^{2+j} e^{-\tilde{P}\frac{u}{4M}} 
\end{align}
for $j=0, ... , \hat{J}$ and $k=0, ... , \hat{K}$ along  $\tilde{\mathcal{C}}_{v=v_\infty}$. The constant $C$ depends only on $M$ and $\hat{I}$, $\hat{J}$ and $\hat{K}$.
\end{proposition}

\begin{proof}
The proof mimics that of Propositions \ref{prop:construction} and \ref{prop:construction1a}. We present the estimates for the first few quantities to illustrate the use of $r$-weights. The analogue of (\ref{phiode}) reads
\begin{align}
\frac{1}{\Omega^2} \partial_u \left(r^2\frac{1}{\Omega^2} \partial_u \left[\frac{\Phi}{r} \right]  \right)= \nonumber \\
 - \frac{\Phi}{r} \frac{r^2}{8} \left(\hat{\slashed{g}}^{-1}_{dat(v_\infty)_\tau}\right)^{AC}\left(\hat{\slashed{g}}^{-1}_{dat(v_\infty)_\tau}\right)^{BD} \left[ \frac{1}{\Omega^2} \partial_u \hat{\slashed{g}}^{dat(v_\infty)_\tau}_{AB} \right]\left[ \frac{1}{\Omega^2} \partial_u \hat{\slashed{g}}^{dat(v_\infty)_\tau}_{CD} \right] 
\end{align}
which, in view of $\Phi^2 = r^2 \sqrt{\gamma}$ on $S^2_{u_\tau,v_\infty}$ leads to the estimate
\begin{align} 
\sum_{i=0}^I    \Big| r^2(u,v_\infty) \left(\partial_u\right)^j \left(\partial_\theta\right)^i \ \left[\log \frac{\Phi^2}{r^2 \sqrt{\gamma}} \right] \Big| \leq C \tilde{P}^j e^{-\tilde{P}\frac{u}{4M} \cdot 2} \ \ \ \textrm{for $j=0, ... , J+1$.} \nonumber
\end{align}
Furthermore, in view of
\begin{align} \label{polk}
 \underline{\hat{\chi}}_{AB} = \frac{1 }{2} \Phi^2 \frac{1}{\Omega^2} \partial_u \hat{\slashed{g}}^{dat(v_\infty)_\tau}_{AB} \sim \frac{1}{2} r \partial_u \hat{\slashed{g}}^{dat\mathcal{I}_\tau}_{AB} + \mathcal{O}\left(1\right) \, ,
\end{align}
we have 
\begin{align}
r \| \underline{\hat{\chi}} \|_{\slashed{g}^\circ} &= \sqrt{\frac{1}{4} r^4 \left(\slashed{g}^\circ\right)^{AB} \left(\slashed{g}^\circ\right)^{CD} \partial_u \hat{\slashed{g}}^{dat\mathcal{I}_\tau}_{AC}  \partial_u \hat{\slashed{g}}^{dat\mathcal{I}_\tau}_{BD} } + \mathcal{O}\left(r^{-\frac{1}{2}}\right) \nonumber \\ 
&= \frac{1}{2} \| \partial_u \hat{\slashed{g}}^{dat\mathcal{I}_\tau}\|_\gamma + \mathcal{O}\left(r^{-\frac{1}{2}}\right) \nonumber \, ,
\end{align}
which leads to (\ref{beisp2}) and moreover shows that it is indeed $\hat{\slashed{g}}^{dat\mathcal{I}_\tau}_{AB}$ which incorporates the non-trivial radiative information in the free-data $\hat{\slashed{g}}^{dat(v_\infty)_\tau}_{AB}$. Note also that one could replace $\slashed{g}^\circ$ by $\slashed{g}$ here. Next, for $\underline{\eta}_A$ we find
\[
\frac{1}{\Omega^2} \partial_u \left(\Phi^2 \underline{\eta}_A\right) = \Phi^2 \left[ (\slashed{g}^{-1})^{BC} \slashed{\nabla}_B \hat{\underline{\chi}}_{CA} - \frac{1}{2} \slashed{\nabla}_A \left(tr \underline{\chi} - tr \underline{\chi}_\circ \right)\right] \, .
\]
In our coordinate frame the right hand side is like $\sim r$ (the dominant term coming from $\underline{\hat{\chi}}$), which after integration and exploiting the exponential decay in $u$, yields $\frac{1}{r}$ decay for $\underline{\eta}_A$ and hence $\frac{1}{r^2}$-decay for the norm $\|\underline{\eta}\|_{\slashed{g}^\circ}$. The other quantities are obtained in complete analogy to the case of the horizon, thereby obtaining (\ref{beisp2}) for all $\Gamma_p$ and $\psi_p$ with appropriate $r$-weights.
\end{proof}
\subsection{Radiation fields} \label{sec:radfields2}
The proof of Proposition \ref{prop:construction2} (cf.~(\ref{polk})) suggests that one can isolate the dominant term in a $\frac{1}{r}$ asymptotic expansion for each $\Gamma_p$ and $\psi_p$ and that this term depends only on the scattering data. This leads one to define the following quantities:
\begin{align}
\hat{\underline{\chi}}^{\mathcal{I}}_{AB} \left(u,v,\theta\right)=  \frac{1}{2} r \left(u,v\right) \cdot \partial_u \hat{\slashed{g}}^{dat\mathcal{I}}_{AB} \left(u,\theta\right)
\end{align}
\begin{align}
\left(tr \underline{\chi}- tr \underline{\chi}_\circ\right)^{\mathcal{I}}  \left(u,v,\theta\right) = \frac{1}{r^2 \left(u,v\right)} \int_u^\infty r^2\|\hat{\underline{\chi}}^{\mathcal{I}} \|_{\slashed{g}^\circ}^2 \left(\bar{u},v,\theta \right) d\bar{u}
\end{align}
\begin{align}
\underline{\alpha}_{AB}^{\mathcal{I}} \left(u,v,\theta\right) = - \frac{1}{2} r \left(u,v\right) \cdot \partial_u \partial_u \hat{\slashed{g}}^{dat\mathcal{I}}_{AB} \left(u,\theta\right)
\end{align}
\begin{align}
\underline{\eta}_A^\mathcal{I} \left(u,v,\theta\right)  = -{\eta}_A^\mathcal{I} \left(u,v,\theta\right)  = +\frac{1}{2} r^{-1}  \left(u,v\right) \left(\gamma^{-1}\right)^{BC} \slashed{\nabla}^\circ_B  \hat{\slashed{g}}^{dat\mathcal{I}}_{CA} \left(u,\theta\right)
\end{align}
\begin{align} \label{memory}
\left(tr \chi - tr \chi_\circ \right)^\mathcal{I} \left(u,v,\theta\right) = 0 \textrm{ \ \ \ \ \ as (\ref{3trc}) shows it decays like $r^{-3}$}
\end{align}
\begin{align}
\hat{{\chi}}^{\mathcal{I}}_{AB} \left(u,v,\theta\right) =  -\frac{1}{2}\hat{\slashed{g}}^{dat\mathcal{I}}_{AB} \left(u,\theta\right)
\end{align}
\begin{align}
\underline{\beta}_A^\mathcal{I} \left(u,v,\theta\right) = \slashed{div}^\circ \hat{\underline{\chi}}^{\mathcal{I}}_{AB} \left(u,v,\theta\right)
\end{align}
\begin{align}
\left(\rho-\rho_\circ\right)^\mathcal{I} \left(u,v,\theta\right) = -\slashed{div}^\circ \underline{\eta}^\mathcal{I} + \frac{1}{2} \left(\hat{\chi}^\mathcal{I},\hat{\underline{\chi}}^\mathcal{I}\right) - \frac{1}{4} tr \chi_\circ \left(tr \underline{\chi} - tr \underline{\chi}_0\right)^\mathcal{I}
\end{align}
\begin{align}
\sigma^\mathcal{I}  \left(u,v,\theta\right) = curl^\circ \left({\eta}^\mathcal{I} \left(u,v,\theta\right)\right) -\frac{1}{2}\hat{{\chi}}^{\mathcal{I}} \wedge_\circ \hat{\underline{\chi}}^{\mathcal{I}}
\end{align}
\begin{align}
\beta^\mathcal{I}_A = 0 \qquad \textrm{as (\ref{elliptic2}) shows that it decays like $r^{-4}$}
\end{align}
\begin{align}
\alpha^\mathcal{I}_{AB} = 0 \qquad \textrm{as Bianchi shows that it decays like $r^{-5}$}
\end{align}
\begin{align}
\left(\hat{\omega} - \hat{\omega}_\circ\right)^\mathcal{I} \left(u,v,\theta\right) = -2 \int_u^\infty \left(\rho-\rho_\circ\right)^\mathcal{I} \left(\bar{u},v,\theta\right) d\bar{u}
\end{align}
\begin{align}
(b^\mathcal{I})^A = -\int_u^\infty 4(\eta^\mathcal{I})^A \left(\bar{u},v,\theta\right) d\bar{u}
\end{align}
\begin{align} \label{gclose}
\left(\slashed{g}_{AB}- r^2 \gamma_{AB}\right)^\mathcal{I} = r \slashed{g}_{AB}^{dat\mathcal{I}}
\end{align}
Here all operations of differentiation, contraction and index raising are defined with respect to the round metric $\slashed{g}^\circ = r^2 \gamma$. Note also that their dependence on $v$ is through the function $r\left(u,v\right)$ only.  Similarly, we make the same definition for the cut-off data, i.e.
\[
\hat{\underline{\chi}}^{\mathcal{I}_\tau}_{AB} \left(u,v,\theta\right)=  \frac{1}{2} r \left(u,v\right) \cdot \partial_u \hat{\slashed{g}}^{dat\mathcal{I}_\tau}_{AB} \left(u,\theta\right) \ \ \ \ \ \ etc.
\]
The point of isolating the radiation fields is that we have the following Proposition, which follows from the computations encountered in the proof of Proposition \ref{prop:construction2}.

\begin{proposition} \label{prop:asymptotic}
For each $\Gamma_p$ and $\psi_p$ we have along $v=v_{\infty}$:
\begin{align}
 \Gamma_p = \Gamma_p^{\mathcal{I}_\tau} + Err\left[\Gamma_p\right] \qquad \qquad \psi_p = \psi_p^{\mathcal{I}_\tau} + Err\left[\psi_p\right]
\end{align}
where $\Gamma_p^{\mathcal{I}_\tau},\psi_p^{\mathcal{I}_\tau}$ are $S^2_{u,v}$-tensors determined entirely by the (cut-off) scattering data $\hat{\slashed{g}}^{dat\mathcal{I}_\tau}$. Moreover, along $v=v_\infty$ the estimates
\begin{align}
w^{2p} \| \psi_p^{\mathcal{I}_\tau} \|_{\slashed{g}^\circ}^2 \leq C \tilde{P}^n e^{-\tilde{P}\frac{u}{2M}} \qquad \textrm{and} \qquad  w^{2p+\frac{1}{2}} \| Err\left[\psi_p\right] \|_{\slashed{g}^\circ}^2 &\leq C \tilde{P}^n e^{-\tilde{P}\frac{u}{2M}} \nonumber
\end{align}
hold \underline{uniformly in $\tau$ and $v_\infty$} for constants $C$ and some $n \leq 3$. Note that the error exhibits improved decay in $r$.\footnote{We could replace $\slashed{g}^\circ$ by $\slashed{g}$ in the above estimates in view of (\ref{gclose}).}  Finally,
\begin{align}
w^{2p} \| \psi_p^{\mathcal{I}_\tau} - \psi_p^\mathcal{I} \|_{\slashed{g}^\circ}^2 \leq C \tilde{P}^n e^{-\tilde{P}\frac{u\left(\tau,v_\infty\right)}{2M}}  \, .
\end{align}
The same three estimates hold with $\psi_p$ replaced by $\Gamma_p$ and $\psi_p^{\mathcal{I}_\tau}$ replaced by $\Gamma_p^{\mathcal{I}_\tau}$.
\end{proposition}

\begin{remark}
These asymptotics are sufficient for our purposes but could be refined. In particular, capturing the $r^{-3}$-asymptotics of (\ref{memory}) from (\ref{3trc}) is related to the celebrated Christodoulou memory effect \cite{Christodouloumemory}.
\end{remark}

\begin{remark}
The radiation fields defined in this section should be directly compared with those in Chapter 17 of \cite{ChristKlei}. 
\end{remark}

\subsection{Estimates for geometric quantities}
Taking into account that $\tau \sim 2v$ along the horizon and $\tau \sim 2u$ along null-infinity, Propositions \ref{prop:construction}, \ref{prop:construction1a} and \ref{prop:construction2} imply the following
\begin{proposition} \label{prop:databounds}
Given any $P \in \mathbb{R}^+$ and $s \in \mathbb{N}$ we can choose parameters $\tilde{P} \in \mathbb{R}^+, I \in \mathbb{N}, J\in \mathbb{N}$ such that for any scattering data set $\left(\hat{\slashed{g}}^{dat\mathcal{H}},   \hat{\slashed{g}}^{dat\mathcal{I}}, \tilde{P}, I, J\right)$, the geometry of an associated finite scattering data set $D_{\tau,v_\infty}$ satisfies the following estimates uniformly in both $\tau$ and $v_\infty$:
\begin{itemize} 
\item On the hypersurface $\tilde{\mathcal{C}}_{U=0}$
\begin{align} \label{hoz3}
\sum_{i=0}^s \sum_{i-perms} \|\mathfrak{D}^i \Gamma_p \|^2_{L^2 \left(S^2_{U=0,v}\right)} \leq  e^{-P \frac{\tau\left(U=0,v\right)}{2M}}  
\end{align}
holds for any $v_{\tau_0} \leq v \leq \infty$ and
\begin{align} \label{hoz4}
F\left[\mathfrak{D}^s \Psi \right] \left(\{U=0\} \times \left[v,v_\tau \right] \right) &\leq  e^{-P \frac{\tau\left(U=0,v\right)}{2M}} \, ,  \\
\sum_{i=0}^{s-1} \sum_{i-perms} \|\mathfrak{D}^i \underline{\alpha} \|^2_{L^4 \left(S^2_{U=0,v}\right)} &\leq  e^{-P \frac{\tau \left(U=0,v\right)}{2M}}. \label{auxab}
\end{align}

\item On the hypersurface $\tilde{\mathcal{C}}_{v=v_\infty}$
\begin{align} \label{inf3}
\sum_{i=0}^s \sum_{i-perms} w^{2p-2} \|\mathfrak{D}^i \Gamma_p \|^2_{L^2 \left(S^2_{u,v=v_\infty}\right)} \leq  e^{-P \frac{\tau\left(u,v_\infty\right)}{2M}} 
\end{align}
holds for any $u_{\tau_0} \leq u \leq \infty$ and 
\begin{align} \label{inf4}
F\left[\mathfrak{D}^s \Psi \right] \left( \left[u,u_\tau \right) \times \{v_\infty\} \right) &\leq  e^{-P \frac{\tau\left(u,v_\infty\right)}{2M}} \, ,  \\
\sum_{i=0}^{s-1} \sum_{i-perms} w^6 \|\mathfrak{D}^i {\alpha} \|^2_{L^4 \left(S^2_{U=0,v}\right)} &\leq e^{-P \frac{\tau\left(u,v_\infty\right)}{2M}}. \label{auxa}
\end{align}
\end{itemize}
\end{proposition}
Note that the larger one chooses $P$ and $s$, the larger one has to choose $\tilde{P}$ and $I,J$. We emphasize that the estimates of Proposition \ref{prop:databounds} hold \emph{uniformly} in $\tau$ and $v_\infty$ for any associated finite scattering data set $D_{\tau,v_\infty}$.

\begin{definition} \label{def:Padm}
Given $P \in \mathbb{R}^+$ and $s \in \mathbb{N}$, we will call a scattering data set $\left(\hat{\slashed{g}}^{dat\mathcal{H}},   \hat{\slashed{g}}^{dat\mathcal{I}}, \tilde{P},I,J\right)$
\underline{$P_s$-admissible} if $I$, $J$ and $\tilde{P}$ are sufficiently large such that the estimates of Proposition \ref{prop:databounds} hold for any associated finite data set $D_{\tau,v_\infty}$.
\end{definition}

\subsection{Local well-posedness}
Let us already remark that from \cite{ChoquetBruhat,Rendall, formationofbh} together with domain of dependence arguments we have:
\begin{theorem} \label{theolocal}
Any finite scattering data set $D_{\tau,v_\infty}$ with $\tau > \tau_0$ determines a unique smooth solution, i.e.~a smooth
Lorentzian metric $g$ expressed as $(\ref{maing})$ satisfying the Einstein vacuum equations (\ref{vacEq}) in $\mathcal{M} \left(t,\tau\right) \cap \{v \leq v_\infty\}$ for some $t < \tau$.
\end{theorem}
We will review this argument in Step 1 of Section \ref{sec:logic}.

\section{The Main Theorems} \label{sec:maintheorem}

\subsection{The full existence theorem}
We begin with a more precise version of the main theorem stated in the introduction specialised to the Schwarzschild case. Recall Definition \ref{def:Padm}.
\begin{theorem} \label{theo:full}
Given any integer $s\geq 3$ and an $M>0$, there exists a $P>0$ such that any $P_s$-admissible scattering data set  $\left(\hat{\slashed{g}}^{dat\mathcal{H}},   \hat{\slashed{g}}^{dat\mathcal{I}}, \tilde{P}, I, J\right)$ gives rise to a spacetime $\left(\mathcal{M} \left(\tau_0,\infty\right),g\right)$ for some $\tau_0<\infty$ with the following properties:
\begin{itemize}
\item The metric $g$ can be globally expressed in the double null-coordinates (\ref{maing}). 

\item The metric converges exponentially in $\tau$ to Schwarzschild with mass $M$ in the following sense: All metric- and Ricci-coefficients $\Gamma_p$ as well as the derivatives $\mathfrak{D}^k\Gamma_p$ for $k\leq s-1$ are in $L^\infty_u L^\infty_v L^2 \left(S^2_{u,v}\right)$ and satisfy for $k=0,1,...,s-1$ the estimate
\[
 \int_{S^2\left(u,v\right)} \| \mathfrak{D}^k \Gamma_p \|^2 \sqrt{\slashed{g}} \ d\theta_1 d\theta_2 \leq
 C_s \cdot e^{-P\frac{\tau}{2M}} \left(\frac{2M}{r}\right)^{2p-2} \left(u,v\right) \, .
\]
The curvature components $\Psi$ satisfy for $k=0,1,...,s-1$ the estimate (\ref{decl2}).
\item We have ${\rm Ric}(g)=0$.\footnote{For $s=3$, this equation is satisfied in $L^\infty_u L^\infty_v L^2 \left(S^2_{u,v}\right)$ with respect to the double-null coordinates.}

\item The metric $g$ realizes the given scattering data on the horizon $\mathcal{H}^+$ and (in the limit) at null-infinity $\mathcal{I}^+$. Both $\mathcal{H}^+$ and $\mathcal{I}^+$ are manifestly 
complete (cf.~\cite{Chrmil}).
\end{itemize}
\end{theorem}

We will prove Theorem \ref{theo:full} for $s=3$. This is the minimum regularity to close the problem via naive energy estimates. As the proof will reveal, once the estimates have closed, one easily obtains the theorem for arbitrary $s$ by further commutation. The last item of Theorem \ref{theo:full} is made precise by the estimate (\ref{realizedata}) of Theorem \ref{theo2}.

\begin{remark} \label{rem:Pdep}
The constant $P$ is expected to grow at least linearly with $s$ because the commuted redshift puts certain obstructions on the decay rate as explained in Sections \ref{sec:redshiftrole} and \ref{sec:hozco}. However, an additional structure in the Bianchi equation may allow one to show that $P$ is independent of $s$ and hence the existence of $C^\infty$-solutions. This will be investigated in future work.
\end{remark}

Theorem \ref{theo:full} will be proven by an analysis of the solutions arising from approximate scattering data associated with the scattering data.

\subsection{The ``finite" existence theorem}
Recall the region $\mathcal{M} \left(\tau_1,\tau_2, v_{\infty} \right)$ defined in (\ref{finregion}).

\begin{theorem} \label{theo1}
Given any integer $s \geq 3$ and $M>0$, there exists a $P>0$ such that the following statement holds:

 For any $P_s$-admissible scattering data set  $\left(\hat{\slashed{g}}^{dat\mathcal{H}},   \hat{\slashed{g}}^{dat\mathcal{I}}, \tilde{P}, I, J\right)$, there exists a $\tau_0$ such that any approximate scattering data set $D_{\tau,v_\infty}$ with $\tau > \tau_0$ gives rise to a unique smooth Lorentzian metric $g_\tau$ satisfying the vacuum Einstein equations in all of $\mathcal{M} \left(\tau_0,\tau,v_{\infty}\right)$. The metric $g_\tau$ can be globally expressed in the double null-coordinates (\ref{maing}) with the associated geometric quantities $\left(\Gamma,\psi\right)$ of $g_\tau$ satisfying the uniform estimates (\ref{decl}), (\ref{decl2}) and (\ref{decl3}) of Section \ref{sec:logic} in $\mathcal{M} \left(\tau_0,\tau,v_{\infty}\right)$.
\end{theorem}
Theorem \ref{theo1} will be proven in Section \ref{sec:bootstrap}.
\subsection{The convergence theorem}
Given Theorem \ref{theo1}, pick a monotonically increasing sequence $\tau_n \rightarrow \infty$
and consider the associated sequence of solutions $\left(\Gamma,\psi\right)_n$ arising from the scattering data set $D_{\tau_n, (v_\infty)_n}$ where $\left(v_\infty\right)_n = (\tau_n)^2$.

For $s$ sufficiently large, given the uniform bounds of Theorem \ref{theo1}, one may apply the Arzela-Ascoli Theorem to conclude the existence of a convergent subsequence and hence the existence of a limiting solution in all of $\mathcal{M}\left(\tau_0,\infty\right)$. This is the way Christodoulou proceeds in the last chapter of \cite{formationofbh} and indeed this already provides a satisfactory version of Theorem \ref{theo:full}.

As discussed already in Section~\ref{difrences},
with potential applications in mind, we will establish a stronger statement below, namely that the above sequence is Cauchy in a suitable space and converges to an ``ultimately Schwarzschildean" spacetime \cite{Holzegelspin2} in $\mathcal{M}\left(\tau_0,\infty\right)$:

\begin{theorem} \label{theo2}
The sequence of solutions $\left(\Gamma,\psi\right)_n$ converges for $\tau_n \rightarrow \infty$ to a vacuum spacetime $\left(\mathcal{M} \left(\tau_0,\infty\right), g_{lim} \right)$ with associated geometric quantities $\left(\Gamma,\psi\right)_{lim}$ such that for $k=0, \ldots, s-1$ we have for any $\left(u,v,\theta_1,\theta_2\right) \in \mathcal{M} \left(\tau_0,\infty\right)$:
\begin{equation}
w^{2p-2} \int_{S^2_{u,v}} \| \mathfrak{D}^k \left(\Gamma_p\right)_{n} - \mathfrak{D}^k \left( \Gamma_p \right)_{lim} \|^2 \sqrt{\slashed{g}} d\theta_1 d\theta_2 \rightarrow 0 \, .
\end{equation}
Moreover, the limiting $\left(\psi,\Gamma\right)_{lim}$ satisfy the estimates (\ref{decl}), (\ref{decl2}) and (\ref{decl3}) in $\mathcal{M} \left(\tau_0,\infty\right)$ for $k=0,1,\ldots, s-1$.  Finally, we have for any fixed $u$
\begin{equation} \label{realizedata}
\lim_{v \rightarrow \infty} w^{2p-2} \int_{S^2_{u,v}} \|  \mathfrak{D}^k \left( \Gamma_p \right)_{lim} - \mathfrak{D}^k \left(\Gamma_p\right)^\mathcal{I} \|^2 \sqrt{\slashed{g}} d\theta_1 d\theta_2 = 0 \, .
\end{equation}
\end{theorem}
In particular, Theorem \ref{theo2} can be applied to provide a uniqueness statement \emph{within the class of exponentially decaying solutions} (see however the discussion in Section \ref{sec:uniq}).
 Theorem \ref{theo2} is addressed in Section \ref{sec:convergence}. Note that Theorems \ref{theo1} and \ref{theo2} imply Theorem \ref{theo:full}, with the last statement of Theorem \ref{theo:full} being understood as the statement (\ref{realizedata}).

\section{Proof of Theorem \ref{theo1}} \label{sec:bootstrap}

\subsection{The logic of the proof} \label{sec:logic}
The proof of Theorem \ref{theo1} proceeds by a continuity argument. Let us define the following dimensionless constant:
\begin{align} \label{Cmaxdef}
C_{max} = 1+ 2^s \max_{i \leq s} \Big| r \ \mathfrak{D}^i \left( tr \chi_\circ \right) \Big|=  1+2^s \max_{i \leq s} \Big| 2r \ \mathfrak{D}^i \left(\frac{1-\frac{2M}{r}}{r} \right) \Big|
\end{align}
where we maximize over all expressions that can arise from applying up to $s$ derivatives from the collection $\mathfrak{D}$ to $ tr \chi_\circ$. Recall that $s$ denotes the number of frame-derivatives that we are commuting with and that we set $s=3$ for the remainder of the proof.
\begin{definition} \label{def:Adef}
Let $\mathcal{A} \subset\left[\tau_0,\tau_f\right) $ denote the set of real numbers $t \in \left[\tau_0,\tau_f\right)$ such that
\begin{enumerate}
\item There exists a smooth solution (\ref{maing}) of the vacuum Einstein equations (\ref{vacEq}) in canonical double-null coordinates $\left(u,v,\theta^1, \theta^2\right)$ in $\mathcal{M} \left(t,\tau_f,v_\infty\right)$ which realizes the data prescribed on the horizon $U=0$ and on $v=v_\infty$.
\item The following estimates hold in $\mathcal{M} \left(t,\tau_f,v_\infty\right)$ for the geometric quantities $\Gamma_p$ and $\psi_p$ of the solution. For $k=0,1,2,3$:
\begin{enumerate}
\item  Bootstrap Assumption on the Ricci-coefficients:
\begin{align} \label{decl}
  \int_{S^2\left(u,v\right)} \| \mathfrak{D}^k \Gamma_p \|^2 \sqrt{\slashed{g}} \ d\theta_1 d\theta_2 \leq
 C^2_{max} \cdot e^{-P\frac{\tau}{2M}} \left(\frac{2M}{r}\right)^{2p-2} \left(u,v\right) \, .
\end{align}
\item Bootstrap Assumption on Curvature\footnote{Recall that $\mathcal{M} \left(t,\tau_f,v_\infty\right)$ determines the quantities $u_{fut}$ and $v_{fut}$ in (\ref{decl2}) via definition (\ref{ufutdef}).}:
\begin{align} \label{decl2}
F\left[\mathfrak{D}^k \Psi\right] \left( \{u\} \times [v,v_{fut}] \right) + F\left[\mathfrak{D}^k\Psi\right] \left( \left[u,u_{fut}\right] \times \{v \} \right)
\leq 4 \cdot  e^{-P\frac{\tau}{2M}}
\end{align}
\begin{align} \label{decl3}
\mathcal{E}\left[\mathfrak{D}^k \Psi\right] \left(\tau,r_1=2M,r_2=r \left(\tau,v_\infty\right)\right) \leq 4 \cdot e^{-P\frac{\tau}{2M}} \, \, .
\end{align}
\item Auxiliary Bootstrap assumption: For $i=0$ and $i=1$,
\begin{align} \label{decl4}
\| \mathfrak{D}^i \slashed{g}_{AB} - \mathfrak{D}^i \slashed{g}_{AB}^\circ \|_{L^\infty \left(S^2\left(u,v\right)\right)} + \| r^2 K - 1 \|_{L^\infty \left(S^2\left(u,v\right)\right)}  \leq \epsilon \frac{2M}{r \left(u,v\right)}
\end{align}
where $K$ denotes the Gauss curvature of $\left(S^2\left(u,v\right), \slashed{g} \right)$ and $\epsilon=1/100$.
\end{enumerate}
\end{enumerate}
\end{definition}

\begin{remark} \label{rem:sc}
Note that the norms in (\ref{decl2}) and (\ref{decl3}) already incorporate the radial weights for the different $\psi_p$, so that the right hand side of (\ref{decl2}) does not contain any $r$-weights. Bootstrap assumption (2c) ensures uniform control over the Sobolev constants and the isoperimetric constant on $\left(S^2\left(u,v\right), \slashed{g} \right)$.
\end{remark}

Our goal is to show that $\mathcal{A} = \left[\tau_0,\tau_f\right)$. This will be achieved by first showing that there exists an $\epsilon>0$ such that $\left[\tau_f-\epsilon,\tau_f\right) \subset \mathcal{A}$ and then showing that if $\left[t,\tau_f\right) \subset \mathcal{A}$ for some $t<\tau_f$, then there exists an $\epsilon>0$ such that $\left[t-\epsilon,\tau_f\right) \subset \mathcal{A}$. 

For the second step, the hardest part of the analysis is to establish the following
\begin{proposition}[Improving the bootstrap assumptions] \label{prop:explain}
The constants $P$ and $\tau_0$ can be chosen sufficiently large such that the following statement holds: If $t \in \mathcal{A}$, then the estimates (\ref{decl})--(\ref{decl4}) actually hold in $\mathcal{M} \left(t,\tau_f,v_\infty\right)$ with a factor $\frac{3}{4}$ on their right hand sides.
\end{proposition}
The above Proposition will follow from Propositions \ref{prop:improve1} and  \ref{prop:improve2} of Section \ref{sec:curvimprove}. Let us finally explain how the above remarks and Proposition \ref{prop:explain} imply Theorem \ref{theo1}.

\emph{Step 1.} The fact that $\mathcal{A}$ is non-empty can be inferred as follows. By domain of dependence, in the region $\mathcal{M}\left(\tau_0,\tau_f\right) \cap \{u \geq u_\tau \} \cap \{v \geq v_\tau\}$ Schwarzschild is a development of the data on $\Sigma_{\tau_f}$. This induces smooth Schwarzschildean data on the null-boundary components $u=u_\tau$ and $v=v_\tau$ respectively of the latter region. By Rendall's theorem \cite{Rendall}, a smooth solution to (\ref{vacEq}) of the two resulting characteristic problems (with data on $\mathcal{H}^+$ and $v=v_\tau$, and $u=u_\tau$ and $v=v_\infty$ respectively) exists and can be defined on a manifold whose future boundary contains a neighborhood of the spheres $S^2_{U=0,v_\tau}$ (in $\mathcal{H}^+ \cup \{v=v_\tau\}$) and $S^2_{u_\tau,v_\infty}$ (in $\{v=v_\tau\} \cup \{u=u_\tau\}$) respectively. One can express that solution in the canonical double-null coordinates using the implicit function theorem (possibly shrinking the neighborhood) so that it exists in the small shaded regions below
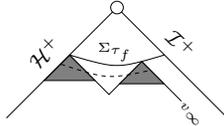
\[
\input{step1.pstex_t}
\]
Finally, for sufficiently small $\epsilon$ the bootstrap assumptions hold in  $\mathcal{M}\left(\tau_f-\epsilon,\tau_f\right)$ by continuity. See Chapter 16.3 of \cite{formationofbh}, where the transformation from harmonic coordinates to double-null coordinates is carried out in complete detail in a related setting.

\emph{Step 2.} To proceed with the continuity argument given Proposition \ref{prop:explain}, we assume that $\left[t,\tau_f\right) \subset \mathcal{A}$ for some $t<\tau_f$. From the estimates (\ref{decl}), (\ref{decl2}) and (\ref{decl3}) in that region, we can obtain higher regularity bounds by further commutation, the equations being essentially linear at this stage (cf.~Chapter 16.2 of \cite{formationofbh}). In this way we obtain uniform bounds for all derivatives in $\mathcal{M}\left(t,\tau_f\right)$ and conclude that the solution extends smoothly to $\Sigma_t$ producing a smooth initial data set for the Einstein equations on $\Sigma_t$.
Using the compactness of $\Sigma_t \cap \{v\leq v_\infty\}$ and applying Choquet-Bruhat \cite{ChoquetBruhat}, we obtain a smooth solution in $\mathcal{M}\left(t-\epsilon,\tau_f\right) \cap \{u > u_t \} \cap \{v > v_t\}$ for some $\epsilon>0$, which moreover extends smoothly to the boundary hypersurfaces $u=u_t$ and $v=v_t$ where it induces smooth characteristic data. Using again Rendall's theorem in a neighbourhood of the spheres $S^2_{U=0,v_t}$ and $S^2_{u_t,v_\infty}$ we obtain a smooth solution in $\mathcal{M}\left(t-\epsilon,\tau_f\right)$ for some $\epsilon>0$. Applying the implicit function theorem (shrinking $\epsilon$ if necessary), we express this solution in double-null coordinates and attach it smoothly to the solution in  $\mathcal{M}\left(t,\tau_f\right)$ (cf.~Chapter 16.3 of \cite{formationofbh}). 
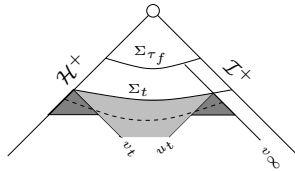
\[
\input{step2.pstex_t}
\]
Finally, by Proposition \ref{prop:explain} and continuity, the bootstrap assumptions continue to hold in this region (perhaps by shrinking $\epsilon>0$ once more).

\subsection{Improving the auxiliary bootstrap assumptions} \label{sec:booot}
We will now begin the proof of Proposition \ref{prop:explain}.  For any given $\tau_0$ large, let $\tau_f > \tau_0$ and consider the (past) development of the finite data set $D_{\tau_f,v_\infty}$. We define the bootstrap region
\begin{align}
\mathcal{B}  &= \mathcal{M} \left(\overline{\tau},\tau_f,v_\infty\right)   \
\end{align}
for $\overline{\tau} = \inf \mathcal{A}$ defined in Definition \ref{def:Adef}.
\subsubsection{Sobolev inequalities on $S^2_{u,v}$} \label{sec:sob1}

\begin{lemma} \label{lem:sob}
Consider the Riemannian manifold $\left(S^2_{\left(u,v\right)}, \slashed{g} \right)$ with $\slashed{g}$ satisfying the bootstrap assumptions (\ref{decl4}). Then we have for any integrable function $f$ whose derivative is integrable in $S^2_{u,v}$, the isoperimetric inequality
\begin{align} \label{isoperi}
\int_{S^2_{(u,v)}} \left(f-\bar{f}\right)^2 \sqrt{\slashed{g}} d\theta^1 d\theta^2 \leq c_{iso}\left(\int_{S^2_{(u,v)}} |r\slashed{\nabla} f|  \sqrt{\slashed{g}} d\theta^1 d\theta^2 \right)^2 \, .
\end{align}
Moreover, for any $S^2_{(u,v)}$-tensor $\xi$ the Sobolev inequalties
\begin{align}
\left[ \frac{1}{r^2} \int_{S^2_{(u,v)}} |\xi|^4 \sqrt{\slashed{g}} d\theta^1 d\theta^2 \right]^\frac{1}{2} \leq c_{sob} \left[\frac{1}{r^2} \int_{S^2_{(u,v)}} \left( |r \slashed{\nabla} \xi|^2 + |\xi|^2 \right) \sqrt{\slashed{g}} d\theta^1 d\theta^2 \right] \, , \nonumber
\end{align}
\begin{align}
\sup_{S^2_{(u,v)}} |\xi|^2 \leq c_{sob}  \left[\frac{1}{r^2} \int_{S^2_{(u,v)}} \left( |r \slashed{\nabla} \xi|^4 + |\xi|^4 \right) \sqrt{\slashed{g}} d\theta^1 d\theta^2 \right]^\frac{1}{2} \nonumber
\end{align}
hold. Here $c_{iso}$ and $c_{sob}$ are dimensionless constants whose value is close to their value for the round metric.
\end{lemma}
\begin{proof}
This is standard. For instance, see Section 5.2  \cite{formationofbh}, where the above is proved in a more intricate setting with less control on $\slashed{g}$ than assumed above.
\end{proof}
Combining the last two estimates of Lemma \ref{lem:sob} yields an $L^\infty$ estimate for $\xi$ in terms of the $H^2$ norm of $\xi$.  

In conjunction with Lemma \ref{lem:sob}, the bootstrap assumption (\ref{decl}) implies that for sufficiently large $\tau_0$, we have 
for $k=2,3$ and fixed $u,v$ the estimate
\begin{align} \label{l4r}
\left[\frac{1}{r^{2}} \int_{S^2\left(u,v\right)} \| \mathfrak{D}^{k-1} \Gamma_p \|^4 \sqrt{\slashed{g}} \ d\theta_1 d\theta_2 \right]^\frac{1}{2} \leq  \frac{c_{sob}}{r^2} \cdot \left( \textrm{RHS of (\ref{decl})} \right)
\end{align}
and also the pointwise bound
\begin{align} \label{l4r2}
 \| \mathfrak{D}^{k-2} \Gamma_p \|^2_{L^\infty\left(S^2\left(u,v\right)\right)} := \sup_{\theta \in S^2\left(u,v\right)} \| \mathfrak{D}^{k-2} \Gamma_p \|^2 \leq \frac{c_{sob}}{r^2} \cdot \left( \textrm{RHS of (\ref{decl})} \right) \, .
\end{align}
By choosing $\tau_0$ sufficiently large, (\ref{l4r2}) immediately improves the metric-part of the auxiliary bootstrap assumption (\ref{decl4}).
\subsubsection{Sobolev inequalities on null-cones} \label{sec:sob2}
The control of the null-fluxes (\ref{decl2}) provides estimates on the $2$-spheres foliating the cones via Sobolev inequalities. The latter are derived along the lines of Chapter 10 of \cite{formationofbh} and we now provide a suitable version below. Recall the dimensionless weight $w=r/2M$ from (\ref{wdefn}).

\begin{lemma} \label{lem:ncs}
For any $\left(u,v,\theta_1,\theta_2\right) \in \mathcal{B}$ we have the following Sobolev inequalities along null-cones lying in $\mathcal{B}$:
\begin{align} \label{sobv}
\sup_{v\leq v_{fut}} \| w^{q} \xi \|_{L^4\left(S^2_{u,v}\right)}  \leq & c_{sob} \Big[\| w^{q} \xi \|_{L^4\left(S^2_{u,v_{fut}}\right)}  \\
&+ \Big\{ \int_{v}^{v_{fut}} \int_{S^2_{u,\bar{v}}} d\bar{v} w^{2q-2}  d\mu_{\slashed{g}} \left( \|\mathfrak{D}_{\nearrow} \xi\|^2 + \|\xi\|^2\right) \Big\}^\frac{1}{2} \Big] \nonumber
\end{align}
and
\begin{align} \label{sobu}
\sup_{u\leq u_{fut}} \| w^{q-\frac{1}{2}} \xi \|_{L^4\left(S^2_{u,v}\right)} \leq & c_{sob} \Big[ \| w^{q-\frac{1}{2}} \xi \|_{L^4\left(S^2_{u_{fut},v}\right)} \\ & + \Big\{ \int_{u}^{u_{fut}} \int_{S^2_{\bar{u},v}} d\bar{u} d\mu_{\slashed{g}} w^{2q-2}  \left( \|\mathfrak{D}_{\nwarrow} \xi\|^2 + \|\xi\|^2\right)  \Big\}^\frac{1}{2} \Big] \, . \nonumber
\end{align}
Note that only tangential derivatives to the cone appear on the right hand side.
\end{lemma}

\begin{proof}
Applying the isoperimetric inequality (\ref{isoperi}) on $\left(S^2_{u,v}, \slashed{g}\right)$ to $|\xi|^3$, we derive
\[
\int_{S^2_{(u,v)}} \|\xi \|^6 d\mu_{\slashed{g}} \leq  \frac{c}{r^2} \left[\int_{S^2_{(u,v)}} \|\xi \|^4 d\mu_{\slashed{g}} \right] \left[\int_{S^2_{(u,v)}} \left( \|r \slashed{\nabla} \xi \|^2 + \|\xi\|^2 \right) d\mu_{\slashed{g}} \right] 
\]
for a uniform (in $u$ and $v$) constant $c$ (cf.~Remark \ref{rem:sc}) and any $S^2_{(u,v)}$-tensor $\xi$. From the above we obtain
\begin{align}
{4M^2} \int_{v}^{v_{fut}} d\bar{v} \int_{S^2_{(u,\bar{v})}}  w^{6q} \|\xi \|^6 d\mu_{\slashed{g}} \nonumber \\
\leq c \left[\sup_v \int_{S^2_{(u,v)}} w^{4q} \|\xi \|^4 d\mu_{\slashed{g}} \right] \left[\int_{v}^{v_{fut}} d\bar{v} \int_{S^2_{(u,\bar{v})}} w^{2q-2} \left( \|r \slashed{\nabla} \xi \|^2 + \|\xi\|^2 \right) d\mu_{\slashed{g}} \right] \nonumber
\end{align}
and
\begin{align}
{4M^2} \int_{u}^{u_{fut}} d\bar{u} \int_{S^2_{(\bar{u},v)}}  w^{6q-2} \|\xi \|^6 d\mu_{\slashed{g}} \nonumber \\
\leq c \left[\sup_u \int_{S^2_{(u,v)}} w^{4q-2} \|\xi \|^4 d\mu_{\slashed{g}} \right] \left[\int_{u}^{u_{fut}} d\bar{u} \int_{S^2_{(\bar{u},v)}} w^{2q-2} \left( \|r \slashed{\nabla} \xi \|^2 + \|\xi\|^2 \right) d\mu_{\slashed{g}} \right]. \nonumber
\end{align}
For the first, we note
\[
\frac{d}{dv} \int_{S^2_{(u,v)}} w^{4q} \|\xi \|^4 d\mu_{\slashed{g}} =  \int_{S^2_{(u,v)}} w^{4q} \left[ 4\|\xi\|^2 \xi \cdot \slashed{\nabla}_4 \xi + \|\xi\|^4 \left( \frac{4q}{r} \slashed{\nabla}_4 r  + tr \chi \right) \right]  d\mu_{\slashed{g}} \, ,
\]
which implies (upon using the pointwise bound for $tr \chi$ available through (\ref{l4r2}))
\begin{align}
\int_{S^2_{(u,v)}} w^{4q} \|\xi \|^4 d\mu_{\slashed{g}}  \leq \int_{S^2_{(u,v_{fut})}} w^{4q} \|\xi \|^4 d\mu_{\slashed{g}} \nonumber \\
+ c \sqrt{\int_{v}^{v_{fut}} d\bar{v} \int_{S^2_{(u,\bar{v})}} w^{2q-2} \left( \|r \slashed{\nabla}_4 \xi \|^2 + \|\xi\|^2 \right) d\mu_{\slashed{g}}} \sqrt{\int_{v}^{v_{fut}} d\bar{v} \int_{S^2_{(u,\bar{v})}}  w^{6q} \|\xi \|^6 d\mu_{\slashed{g}}} \nonumber .
\end{align}
Similarly, in the $u$-direction one has
\begin{align}
\int_{S^2_{(u,v)}} w^{4q-2} \|\xi \|^4 d\mu_{\slashed{g}}  \leq \int_{S^2_{(u_{fut},v)}} w^{4q-2} \|\xi \|^4 d\mu_{\slashed{g}} \nonumber \\
+ c\sqrt{\int_{u}^{u_{fut}} d\bar{u} \int_{S^2_{(\bar{u},v)}} w^{2q-2} \left( \| \slashed{\nabla}_3 \xi \|^2 + \|\xi\|^2 \right) d\mu_{\slashed{g}}} \sqrt{\int_{u}^{u_{fut}} d\bar{u} \int_{S^2_{(\bar{u},v)}}  w^{6q-2} \|\xi \|^6 d\mu_{\slashed{g}}} . \nonumber
\end{align}
Combining these inequalities yields the estimates of the Lemma.
\end{proof}

\begin{remark}
Besides the bootstrap assumption (\ref{decl4}), the above proof uses only a pointwise bound on $r \cdot tr \chi$ and $r \cdot tr \underline{\chi}$ available through (\ref{l4r2}).
\end{remark}

We now combine (\ref{sobu}) and (\ref{sobv}) with the bootstrap assumption on the fluxes (\ref{decl2}). Note that since all curvature components $\xi=\psi_p$ are trivial on $S^2_{U=0,v_\tau}$ we can apply (\ref{sobv}) on $\tilde{\mathcal{C}}_{U=0}$ to obtain an $L^4$ bound for $s-1$ derivatives of all $\psi_p$ \emph{except} $\underline{\alpha}$ from the $L^2$-curvature flux of $s$ derivatives along $\mathcal{\tilde{C}}_{U=0}$ . For the missing $\underline{\alpha}$ we have (\ref{auxab}) along $\mathcal{\tilde{C}}_{U=0}$. Similarly on $\tilde{\mathcal{C}}_{v=v_\infty}$  we can apply (\ref{sobu}), which together with (\ref{auxa}) finally yields $L^4$ bounds for $s-1$ derivatives of all $\psi_p$ on $\tilde{\mathcal{C}}_{v=v_\infty}$ taking into account that all $\xi$ vanish at $S^2_{u_\tau,v_\infty}$. Since the future cone from any sphere of the bootstrap region intersects $\tilde{\mathcal{C}}_{U=0} \cup \Sigma_\tau \cup \tilde{\mathcal{C}}_{v=v_\infty}$, we can control the $L^4$-norm of any $\psi_p$ from the fluxes and the data and hence in conjunction with the bootstrap assumptions (\ref{decl2}) this yields for $k=1,2,3$
\begin{align} \label{l4c}
\left[r^{2} w^{4p-4} \int_{S^2\left(u,v\right)}  \| \mathfrak{D}^{k-1}  \psi_p \|^4 \sqrt{\slashed{g}} \ d\theta_1 d\theta_2 \right]^\frac{1}{2} \leq c_{sob} \cdot \textrm{RHS of (\ref{decl2})} \, .
\end{align}
Applying the second inequality of Lemma \ref{lem:sob}, we also obtain the pointwise bound 
\begin{align}
w^{2p-2} \|  \mathfrak{D}^{k-2} \psi_p \|^2_{L^\infty\left(S^2\left(u,v\right)\right)} \leq \frac{c_{sob}}{r^2} \cdot \left( \textrm{RHS of (\ref{decl2})} \right) \, 
\end{align}
for $k=2,3$. In particular, this justifies the notation $\psi_p$ introduced in (\ref{stabmindec}).

As an immediate corollary, using equation (\ref{Gauss}) we improve (by choosing $\tau_0$ large) the bootstrap assumption (\ref{decl4}) on the Gauss curvature. We conclude that the auxiliary bootstrap assumption (2c) has already been improved.

\subsubsection{Geometry of the slices $\Sigma_\tau$}  \label{sec:slices}

Recall the slices $\Sigma_\tau$ defined at the end of Section \ref{sec:manifold}. We now obtain estimates for the normal and the induced volume element on $\Sigma_\tau$. Let us define the constant
\[
C_h := 2 \max_{r \in \left[9/4M,8M\right]} \Big[ \left(\frac{r}{2M}\right)^h \frac{1}{1-\frac{2M}{r}} \Big] \, .
\]

\begin{lemma}
The slices $\Sigma_\tau \cap \mathcal{B}$ satisfy
\[
C_h^{-1} \left(\frac{2M}{r}\right)^h \leq -g\left(\nabla \tau, \nabla \tau\right) \leq C_h \left(\frac{2M}{r}\right)^h \, .
\]
Similarly, we have in $(\tau, r, \theta^1, \theta^2)$-coordinates, the estimate 
\[
C_h^{-1} \left(\frac{2M}{r}\right)^{\frac{h}{2}} \sqrt{\slashed{g}} \leq \sqrt{g_{\Sigma_\tau}} \leq C_h  \left(\frac{2M}{r}\right)^{\frac{h}{2}}\sqrt{\slashed{g}}
\]
for the induced volume element. Denoting the unit normal to the $\Sigma_\tau$ by $n_{\Sigma_\tau}$ we also have
\begin{align} \label{useful}
\begin{split}
\frac{1}{4} \leq \sqrt{g_{\Sigma_\tau}} \cdot g \left(e_3, n_{\Sigma_\tau} \right) &\leq 9 \, , \\
\frac{1}{4}  \left(\frac{2M}{r}\right)^{h} \leq \sqrt{g_{\Sigma_\tau}} \cdot g \left(e_4, n_{\Sigma_\tau} \right) &\leq 16  \left(\frac{2M}{r}\right)^{h} \, .
\end{split}
\end{align}
\end{lemma}

\begin{proof}
It suffices to prove the above estimates (with slightly better constants) for Schwarzschild itself, since the bootstrap assumptions imply in particular uniform $C^1$-closeness estimates of the spacetime-metric (recall we already improved bootstrap assumption (\ref{decl4})). Moreover, one may first establish the above for $\tau_{piece}=u+v+f_{piece}$ instead of $\tau$ and constant $\frac{3}{4}C_h$ instead of $C_h$. Finally, one concludes that the same bounds hold with constant $C_h$ for the appropriately mollified $\tau=u+v+f_{smooth}$. To derive the estimates for $\tau_{piece}$, note that with $k\left(r\right) = \sqrt{1-\left(\frac{2M}{r}\right)^h\left(1-\frac{2M}{r}\right)}$ we have (recalling (\ref{fpiece})) the formula 
\begin{align}
d \left(u+v+f_{piece}\left(r^\star\right) \right) = \left[1- \lambda k\left(r\right) \right] du + \left[1+ \lambda  k\left(r\right) \right] dv 
\end{align}
where $\lambda = 1$ for $r < 9/4M$, $\lambda=0$ for $9/4M < r < 8M$ and $\lambda = -1$ for $r> 8M$. From this all the above bounds easily follow (with better constants).
\end{proof}
\subsection{Improving the assumptions for curvature} \label{sec:curvimprove}
In this section we improve bootstrap assumptions (\ref{decl2}) and (\ref{decl3}).
\subsubsection{The key proposition}
\begin{proposition} \label{prop:improve1}
For any point $\left(u,v, \theta_1,\theta_2\right)$ in the bootstrap region $\mathcal{B}$, we have, for any $k=0,1,2,3$ the estimate
\begin{align} 
F\left[\mathfrak{D}^k \psi\right] \left( \{u\} \times [v,v_{fut}] \right) + F\left[\mathfrak{D}^k\psi\right] \left( \left[u,u_{fut}\right] \times \{v \} \right) \nonumber \\
\leq \left(2+\frac{C_M}{P} + C_M e^{-\frac{P}{2M} \tau_0} \right) e^{-\frac{P}{2M}\tau\left(u,v\right)}
\end{align}
and also
\begin{align} \label{splo}
\mathcal{E}\left[\mathfrak{D}^k \psi\right] \left(\tau,2M,r\left(\tau,v_\infty\right)\right) \leq \left(2+\frac{C_M}{P} + C_M e^{-\frac{P}{2M} \tau_0} \right) e^{-\frac{P}{2M}\tau\left(u,v\right)} .
\end{align}
Moreover, these estimates hold independently of the size of $\tau_f > \tau_0$ defining the bootstrap region.
\end{proposition}

\begin{proof}
Let us assign the following weights to each Bianchi pair 
\begin{align}
\left(\alpha,\beta\right) \textrm{ \ \ has weight \ \ } & q = 5 \, , \nonumber \\
\left(\beta, \left[ \rho-\rho_\circ, \sigma \right] \right) \textrm{ \ \ has weight \ \ } & q = 4 \, , \nonumber \\
\left(\left[ \rho-\rho_\circ, \sigma \right] , \underline{\beta} \right) \textrm{ \ \ has weight \ \ } & q = 2 \, , \nonumber \\
\left(\underline{\beta}, \underline{\alpha} \right) \textrm{ \ \ has weight \ \ } & q = 0 \, .
\end{align}

In the following, we will use the short hand notation
\begin{align}
\hat{\rho} = \rho - \rho_\circ \ \ \ , \ \ \ \left(\hat{\rho},\sigma\right) = \left(\slashed{g} \hat{\rho} - \slashed{\epsilon} \sigma\right) \ \ \ , \ \ \ \|\mathfrak{D}^k \left(\hat{\rho},\sigma\right)\|^2 = \|\mathfrak{D}^k \hat{\rho}\|^2 + \|\mathfrak{D}^k \sigma\|^2.
\end{align}

Contracting the equations of each Bianchi pair with its weighted curvature component over all indices (i.e.~(\ref{cobi1}) with $w^q \mathfrak{D}^k\uppsi_p$ and (\ref{cobi2}) with $w^q \mathfrak{D}^k\uppsi^\prime_{p^\prime}$ respectively) we derive the following identities:
\begin{align} \label{pi1}
\nabla_a \left( \|\mathfrak{D}^k \alpha\|^2 w^5 \left[e_3\right]^a \right) + \nabla_a \left(2 \|\mathfrak{D}^k\beta\|^2 w^5 \left[e_4\right]^a \right) -4 \Omega^{-2} \slashed{\nabla}_A Q_1^A \nonumber \\
= f_1 w^5 \|\mathfrak{D}^k\alpha\|^2 + f_1 w^5 \|\mathfrak{D}^k\beta\|^2 -4 w^5 \left(\eta + \underline{\eta}\right) \cdot  \mathfrak{D}^k\alpha \cdot  \mathfrak{D}^k\beta \nonumber \\
+2 E_3\left[\mathfrak{D}^k\alpha\right] \cdot \mathfrak{D}^k\alpha \ w^5 + 4 E_4\left[\mathfrak{D}^k\beta\right] \cdot \mathfrak{D}^k\beta \ w^5 \, ,
\end{align}
\begin{align}  \label{pi2}
\nabla_a \left( \|\mathfrak{D}^k \beta\|^2 w^4 \left[e_3\right]^a \right) + \nabla_a \left( \|\mathfrak{D}^k\left(\hat{\rho},\sigma\right)\|^2 w^4 \left[e_4\right]^a \right) -2 \Omega^{-2} \slashed{\nabla}_A Q_2^A \nonumber \\
= f_1 w^4 \|\mathfrak{D}^k\beta\|^2 + \boxed{f_2} \ w^4 \|\mathfrak{D}^k\left(\hat{\rho},\sigma\right)\|^2 -2w^4 \left(\eta + \underline{\eta}\right) \cdot  \mathfrak{D}^k\beta \cdot  \mathfrak{D}^k\left(\hat{\rho},\sigma\right) \nonumber \\
+ 2E_3\left[\mathfrak{D}^k\beta\right] \cdot  \mathfrak{D}^k\beta \ w^4 +  2E_4\left[\mathfrak{D}^k\left(\hat{\rho},\sigma\right)\right] \cdot \mathfrak{D}^k\left(\hat{\rho},\sigma\right) \ w^4  \, ,
\end{align}
\begin{align}  \label{pi3}
\nabla_a \left( \|\mathfrak{D}^k \left(\hat{\rho},\sigma\right)\|^2 w^2 \left[e_3\right]^a \right) + \nabla_a \left( \|\mathfrak{D}^k \underline{\beta} \|^2 w^2 \left[e_4\right]^a \right) +2 \Omega^{-2} \slashed{\nabla}_A Q_3^A \nonumber \\
= f_1 w^2 \|\mathfrak{D}^k\left(\hat{\rho},\sigma\right)\|^2 + \boxed{f_2} \ w^2 \|\mathfrak{D}^k \underline{\beta}\|^2 + 2w^2 \left(\eta + \underline{\eta}\right) \cdot  \mathfrak{D}^k\left(\hat{\rho},\sigma\right) \cdot \mathfrak{D}^k \underline{\beta} \nonumber \\
+ 2E_3\left[\mathfrak{D}^k\left(\hat{\rho},\sigma\right)\right]  \cdot  \mathfrak{D}^k\left(\hat{\rho},\sigma\right) \ w^2 + 2 E_4\left[\mathfrak{D}^k \underline{\beta}\right] \cdot  \mathfrak{D}^k \underline{\beta} \ w^2 \, ,
\end{align}
\begin{align}   \label{pi4}
\nabla_a \left(2 \|\mathfrak{D}^k \underline{\beta} \|^2  \left[e_3\right]^a \right) + \nabla_a \left( \|\mathfrak{D}^k\underline{\alpha}\|^2  \left[e_4\right]^a \right) +4 \Omega^{-2} \slashed{\nabla}_A Q_4^A \nonumber \\
= f_1  \|\mathfrak{D}^k \underline{\beta}\|^2 + \boxed{f_2} \   \|\mathfrak{D}^k\underline{\alpha}\|^2 +  4\left(\eta + \underline{\eta}\right) \cdot  \mathfrak{D}^k \underline{\beta} \cdot   \mathfrak{D}^k\underline{\alpha} \nonumber \\
+ 4E_3\left[\mathfrak{D}^k \underline{\beta}\right]  \cdot  \mathfrak{D}^k\underline{\beta} \  + 2E_4\left[\mathfrak{D}^k\underline{\alpha}\right]\cdot  \mathfrak{D}^k\underline{\alpha}  \, ,
\end{align}
with 
\[
Q_1^A = w^5 \Omega^2 [\mathfrak{D}^k\alpha \cdot  \mathfrak{D}^k\beta]^A \qquad , \qquad
Q_4^A = \Omega^2 [\mathfrak{D}^k\underline{\alpha} \cdot  \mathfrak{D}^k\underline{\beta}]^A
\]
\[
Q_2^A = w^4 \Omega^2 [\mathfrak{D}^k\beta \cdot  \mathfrak{D}^k\left(\hat{\rho},\sigma\right)]^A \qquad , \qquad
Q_3^A = w^2 \Omega^2 [\mathfrak{D}^k\underline{\beta} \cdot  \mathfrak{D}^k\left(\hat{\rho},\sigma\right)]^A \, .
\]
Here the $\cdot$ denotes the obvious contraction of $S^2_{u,v}$-tensors: For instance, if $\mathfrak{D}^k$ contains $i\leq k$ angular operators, then (cf.~Section \ref{sec:commute})
\[
\left(\eta +\underline{\eta}\right) \cdot \mathfrak{D}^k\alpha \cdot  \mathfrak{D}^k\beta = \left(\eta +\underline{\eta}\right)^A \left(\mathfrak{D}^k\alpha\right)_{C_1...C_iAB} \left( \mathfrak{D}^k\beta \right)^{C_1...C_iB} \, .
\]
To derive (\ref{pi1})--(\ref{pi4}) recall ($e_3 = e_3^{\mathcal{E}\mathcal{F}}$ and $e_4 = e_4^{\mathcal{E}\mathcal{F}}$) 
\[
\nabla_a \left(e^{\mathcal{E}\mathcal{F}}_3\right)^a = tr \underline{\chi} \ \ \ \ \ , \ \ \ \ \ \nabla_a \left(e^{\mathcal{E}\mathcal{F}}_4\right)^a = tr \chi + \hat{\omega} \, ,
\]
as well as
\[
\left(e_3\right)^a \nabla_a r^n = n r^{n-1} e_3 \left(r\right) = \frac{n}{r} r^n \frac{-\Omega^2_\circ}{\Omega^2} 
\]
\[
\left(e_4\right)^a \nabla_a r^n = n r^{n-1} e_4 \left(r\right) = \frac{n}{r} r^n \frac{1}{2} r \cdot tr \chi_\circ 
\]
and apply Proposition \ref{comsum}.
\begin{remark} The crucial observation here is that -- due to the careful choice of weights -- the boxed term is $f_2$ and not $f_1$ in (\ref{pi2})--(\ref{pi4}).\footnote{This would also be the case for (\ref{pi1}), if we had chosen the weight $q=6$ instead of $q=5$ which would be necessary to obtain the characteristic peeling decay of curvature along null-infinity. However, the estimates of course close with less decay.} We illustrate this remark with one particular example, say the $\underline{\beta}$-term in equation (\ref{pi3}) and $k=0$: Collecting all terms which contribute on the right hand side, we find
\begin{align}
\| \underline{\beta}\|^2 r^2 \left[ -2 tr \chi  -2 \hat{\omega} + \frac{2}{r} \frac{1}{2} r \cdot tr \chi_\circ +  tr \chi \right]
=  \| \underline{\beta}\|^2 r^2 \left( f_2 + \Gamma_2\right) \, .
\end{align}
The cubic term $\| \underline{\beta}\|^2 w^2 \Gamma_2$ may be absorbed into the error $E_4\left[\mathfrak{D}^k \underline{\beta}\right]\cdot \mathfrak{D}^k \underline{\beta} \ r^2$.
\end{remark}
\vspace{.2cm}
\noindent After integration with respect to the spacetime volume form $2\Omega^2 \sqrt{\slashed{g}}du dv d\theta_1 d\theta_2$ and summation of (\ref{pi1})--(\ref{pi4}) in $\mathcal{M}\left(\tau,\tau_f,v_\infty\right)$ for $\tau\geq \bar{\tau}$ (and finally, summation over all $k$ and all permutations of derivatives of length $k$) we find:
\begin{enumerate}
\item The terms proportional to $\Omega^{-2} \slashed{\nabla}_A Q_i^A$ in the first lines
of $(\ref{pi1})$--$(\ref{pi4})$ vanish.
\item The other (boundary) terms in the first line will produce terms on $\Sigma_{\tau}$ and $\Sigma_{\tau_f}$ (denoted $T_1$ and $T_2$ respectively) as well as fluxes on $U=0$ and $v=v_\infty$. Clearly, $T_2 = 0$, since the metric is Schwarzschild there and all $\psi$ vanish. The fluxes produced satisfy
\begin{align}
F\left[\mathfrak{D}^k \psi\right] \left(\{U=0\} \times \left(v_\tau,v_{\tau_f} \right) \right) + F\left[\mathfrak{D}^k \psi\right] \left( \left(u_\tau,u_{\tau_f} \right) \times \{v_\infty\} \right) 
\leq 2 e^{-\frac{P}{2M} \tau} \nonumber
\end{align}
as follows from the initial conditions (\ref{hoz4}) and (\ref{inf4}). To compute the boundary term on $\Sigma_\tau$, $T_1$, we recall the estimate (\ref{useful}), which provides
\begin{align}
 T_1
\geq \sum_{i=0}^k \sum_{i-perms}
 \int_{2M}^{r\left(\tau,v_\infty\right)} dr \int_{S^2_{u,v}} \sqrt{\slashed{g}} d\theta^1 d\theta^2 \Bigg[ \nonumber \\
\| \mathfrak{D}^i \alpha \|^2 \frac{1}{4}w^5   
+ \frac{1}{4} w^{-h}  \|\mathfrak{D}^i \underline{\alpha}\|^2  
 +\| \mathfrak{D}^i \beta\|^2 \frac{1}{4}w^4 \nonumber \\ 
 + \| \mathfrak{D}^i \left(\hat{\rho}, \sigma\right) \|^2   \frac{1}{4} w^{4-h}
 + \|\mathfrak{D}^i \underline{\beta} \|^2   \frac{1}{4} w^{2-h} \Bigg] \geq 
\mathcal{E} \left[\mathfrak{D}^k \psi\right] \left(\tau,2M,r\left(\tau,v_\infty\right)\right) \, . \nonumber
\end{align}
\item The terms in the second line of (\ref{pi1})--(\ref{pi4}) can be estimated
\begin{align}
\int_{\mathcal{M} \left(\tau,\tau_f,v_\infty\right)} \textrm{$2^{nd}$-line} \leq C_M \int_{\tau}^{\tau_f} d\tilde{\tau} \ \mathcal{E} \left[\mathfrak{D}^k \psi\right] \left(\tilde{\tau}, 2M, r\left(v_\infty\right)\right) 
\leq \frac{C_M}{P}  e^{-\frac{P}{2M} \tau} \nonumber
\end{align}
with the second step following from the bootstrap assumption.
\item For the terms in the third line, we will prove in Section \ref{errorest}
\begin{align} \label{missing}
\int_{\mathcal{M} \left(\tau,\tau_f,v_\infty\right)} \textrm{$3^{rd}$-line}
\leq \left(\frac{C_M}{P} + C_M e^{-\frac{P}{2M}\tau_0} \right) e^{-\frac{P}{2M} \tau} \, . 
\end{align}
The second term on the right will arise from cubic errors, the first from linear spacetime errors.
\end{enumerate}
After proving (\ref{missing}) in Section \ref{errorest}, this provides the second estimate of the Proposition. To obtain the first, one repeats the proof, now integrating over 
\[
D\left(u,v\right) = J^+ \left(S^2_{u,v}\right) \cap \mathcal{M} \left(\tau,\tau_f,v_\infty\right) \, .
\]
The terms in the first line then produce the desired fluxes, while the spacetime error can be estimated as before (now using the part of the $\Sigma_\tau$ slices lying in $D\left(u,v\right)$).
\end{proof}

\subsubsection{Estimating spacetime errors} \label{errorest}
In this section, we prove the estimate (\ref{missing}) for the error-terms in the third line of (\ref{pi1})--(\ref{pi4}). For each Bianchi pair $\left(\uppsi_p,  \uppsi^\prime_{p^\prime} \right)$, these terms are of the following form:
\begin{align}
 \int_{\mathcal{M} \left(\tau,\tau_f,v_\infty\right)} E_3 \left[\mathfrak{D}^k \uppsi_p \right] \cdot \left(w^q \mathfrak{D}^k \uppsi_p \right) \Omega^2 \sqrt{\slashed{g}} du dv d\theta^1 d\theta^2
\end{align}
and
\begin{align}
 \int_{\mathcal{M} \left(\tau,\tau_f,v_\infty\right)} E_4 \left[\mathfrak{D}^k \uppsi^\prime_{p^\prime} \right] \cdot \left(w^q \mathfrak{D}^k \uppsi^\prime_{p^\prime} \right) \Omega^2 \sqrt{\slashed{g}} d\tau dr^\star d\theta^1 d\theta^2 \, ,
\end{align}
where $\mathcal{M}\left(\tau,\tau_f,v_\infty\right)$ is replaced by $D\left(u,v\right)$ in case the energy estimate is applied in the characteristic region $D\left(u,v\right)$. From Cauchy-Schwarz,
\begin{align}
 \int_{\mathcal{M} \left(\tau,\tau_f,v_\infty\right)} E_3 \left[\mathfrak{D}^k \uppsi_p \right] \cdot \left(w^q \mathfrak{D}^k \uppsi_p \right) \Omega^2 \sqrt{\slashed{g}} d\tau dr^\star d\theta^1 d\theta^2 \nonumber \\
\leq \int_{\mathcal{M} \left(\tau,\tau_f,v_\infty\right)} \left( \| E_3 \left[\mathfrak{D}^k \uppsi_p \right]\|^2 w^q + \| \mathfrak{D}^k \uppsi_p \|^2 w^q \right)  \sqrt{\slashed{g}} d\tau dr d\theta^1 d\theta^2  = \boxed{1A} + \boxed {1B} \nonumber
\end{align}
and
\begin{align}
 \int_{\mathcal{M} \left(\tau,\tau_f,v_\infty\right)} E_4 \left[\mathfrak{D}^k \uppsi^\prime_{p^\prime} \right] \cdot\left(w^q \mathfrak{D}^k \uppsi^\prime_{p^\prime} \right) \Omega^2 \sqrt{\slashed{g}} d\tau dr^\star d\theta^1 d\theta^2  \leq  
 \int_{\mathcal{M} \left(\tau,\tau_f,v_\infty\right)} \nonumber \\ \Big( w^h \| E_4 \left[\mathfrak{D}^k \uppsi^\prime_{p^\prime} \right]\|^2 w^q +w^{-h} \| \mathfrak{D}^k \uppsi^\prime_{p^\prime} \|^2 w^q \Big)   \sqrt{\slashed{g}} d\tau dr d\theta^1 d\theta^2 
= \boxed{2A} + \boxed {2B} \, . \nonumber
\end{align}
\begin{remark}
Recall Remark \ref{rem:comstructure}: It is precisely here that the important structure (the improved decay and special structure of $E_4 \left[\mathfrak{D}^k \uppsi^\prime_{p^\prime} \right]$) in the Bianchi equations is exploited. Otherwise, the term $\boxed{2A}$ could not be handled.
\end{remark}
For the terms $\boxed{1B}$ and $\boxed{2B}$ we see by inspection that
\begin{align}
 \int_{\mathcal{M} \left(\tau,\tau_f,v_\infty\right)} \boxed {1B} + \boxed {2B} \leq  C \int_{\tau}^{\tau_f} d\tilde{\tau} \ \mathcal{E} \left[\mathfrak{D}^k \psi\right] \left(\tilde{\tau}, 2M, r\left(v_\infty\right)\right) \leq \frac{C_M}{P} e^{-\frac{P}{2M} \tau}   \, . \nonumber
\end{align}
For the terms $\boxed{1A}$ and $\boxed{2A}$  we turn to the expressions for $E_3$ and $E_4$ collected in Proposition \ref{comsum} and the expressions (\ref{lin3term}) and (\ref{lin4term}). 

\paragraph{Linear terms.}
We first handle the two ``linear" contributions (\ref{lin3term}) and (\ref{lin4term}) and of those, we first handle the curvature terms. It is not hard to see\footnote{Recall Remark \ref{rem:comstructure} which is crucial here as it provides the necessary decay for the $\psi^\prime_{p^\prime}$-terms in $E_4$.} that
\begin{align}
 \int_{\mathcal{M}_\mathcal{B}\left(\tau,\tau_f\right)}  \| E_3 \left[\mathfrak{D}^k \uppsi_p \right]\|^2 w^q +  w^h \| E_4 \left[\mathfrak{D}^k \uppsi^\prime_{p^\prime} \right]\|^2 w^q \Bigg|_{lin,curv} \nonumber \\
\leq C \sum_{i=0}^k \int_{\mathcal{M}_\mathcal{B}\left(\tau,\tau_f\right)}  w^q \| \mathfrak{D}^i \uppsi_p \|^2 + w^q w^{-h} \|\mathfrak{D}^i \uppsi^\prime_{p^\prime}\|^2 \, ,
\end{align}
which can be estimated as previously.
Secondly, for the ``linear" $\Gamma$-term in (\ref{lin3term}) and (\ref{lin4term}) we have (say for $E_4$)
\[
\sum_{i=0}^k \int_{\mathcal{M}_\mathcal{B} \left(t,t_f\right)} w^q r^2 (f_3)^2 \|\mathfrak{D}^i \Gamma_{min(p^\prime,2)}\|^2  \Omega^2 \sqrt{\slashed{g}} du dv d\theta^1 d\theta^2 \leq \frac{C_M}{P} e^{-\frac{P}{2M} \tau}  \, ,
\]
which follows using the bootstrap assumption (\ref{decl}) on $\Gamma$ and taking into account the strong decay in $r$ which allows us to integrate in both $r$ and $\tau$.

\paragraph{Cubic terms.}
For the cubic (or higher) terms in $\boxed{1A}$ and $\boxed{2A}$ smallness will always arise from $\tau_0$ being large, as we can always estimate one of the terms pointwise and exploit that $\tau_0$ is large. This means that for those terms we only need to check 
\begin{itemize}
\item whether the \emph{decay} in $r$ is sufficient for $\boxed{1A}$ and $\boxed{2A}$ to be \emph{integrable}
\item whether the \emph{regularity} is sufficient to control all the terms of $\boxed{1A}$ and $\boxed{2A}$ from the (purely $L^2$) bootstrap assumptions via Sobolev embedding.
\end{itemize}
Let us start by computing the decay in $r$. The decay is sufficient, if the overall decay of the integrand (including the volume element) is at least $r^{-1-\delta}$. Recall also that taking derivatives of type $\mathfrak{D}$ does not change the $r$-decay of a quantity. Therefore, from Proposition \ref{comsum} we read off\footnote{Note that focusing on \emph{decay} we can simply think of $\Lambda_1 \sim \frac{\epsilon}{r}$ and $\Lambda_2 \sim \frac{\epsilon}{r^2}$.}
\begin{equation}
\begin{split}
\|E_3\left[\mathfrak{D}^k \alpha \right]\|^2 r^5 r^2 &\lesssim r^{-2\left(\frac{7}{2}+1\right)+5 + 2} = r^{-2} \, , \\
\|E_3\left[\mathfrak{D}^k \beta \right]\|^2 r^4 r^2 &\lesssim r^{-2 \cdot 4 +4 + 2}= r^{-2} \, , \\
\|E_3\left[\mathfrak{D}^k \left(\hat{\rho},\sigma\right) \right]\|^2 r^2 r^2 &\lesssim r^{-2 \cdot 3 +2 + 2} = r^{-2} \, , \\
\|E_3\left[\mathfrak{D}^k \underline{\beta} \right] \|^2 r^0 r^2 &\lesssim r^{-2 \cdot 2 +0 + 2} = r^{-2} \, .
\end{split}
\end{equation}
Turning to the $4$-direction and keeping in mind that all linear terms have already been dealt with, we have the following decay rates:
\begin{equation}
\begin{split}
\|E_4\left[\mathfrak{D}^k \beta \right]\|^2 r^5  r^2  r^h \lesssim r^{-2\left(\frac{7}{2}+2\right)+5 + 2 +h}= r^{-4+h} \, ,  \\
\|E_4\left[\mathfrak{D}^k \left(\hat{\rho},\sigma\right) \right]\|^2 r^4 r^2 r^h \lesssim r^{-2 \left(3+ 2\right) +4 + 2 + h} = r^{-4+h} \, ,  \\
\|E_4\left[\mathfrak{D}^k \underline{\beta}  \right]\|^2 r^2 r^2 r^h \lesssim r^{-2 \left(2+ 2\right) +2 + 2 +h}= r^{-4+h} \, ,  \\
\|E_4\left[\mathfrak{D}^k \underline{\alpha} \right]\|^2 r^0 r^2 r^h  \lesssim r^{-2 \left(1+2\right) +0 + 2 +h}= r^{-4+h} \, ,
\end{split}
\end{equation}
which provides sufficient $r$-decay. This establishes that -- at least in terms of $r$-decay -- all terms in $\boxed{1A}$ and $\boxed{2A}$ can be estimated.

In the second step, we address the regularity. Note that the error is of the form
\begin{equation}
\begin{split}
E_i \left[\psi_p\right] &= \psi \left(f + \Gamma\right) + f \Gamma \, , \\
E_i \left[\mathfrak {D}\psi_p\right] &= \mathfrak{D}\psi \left(f + \Gamma\right) + \psi \mathfrak{D}\Gamma + f \mathfrak{D} \Gamma + \textrm{l.o.t.} \, ,
 \\
E_i \left[\mathfrak {D}^2\psi_p\right] &= \mathfrak{D}^2\psi \left(f + \Gamma\right) + \mathfrak{D} \psi \mathfrak{D}\Gamma + f \mathfrak{D}^2 \Gamma + \textrm{l.o.t.} \, ,
 \\
E_i \left[\mathfrak {D}^3\psi_p\right] &= \mathfrak{D}^3\psi \left(f + \Gamma\right) + \mathfrak{D}^2 \psi \mathfrak{D}\Gamma +   \mathfrak{D} \psi \mathfrak{D}^2 \Gamma+ f \mathfrak{D}^3 \Gamma + \textrm{l.o.t.} 
\end{split}
\end{equation}
Clearly, we only need to look at the worst term $E_i \left[\mathfrak {D}^3\psi_p\right]$ for our regularity considerations. Since we control the curvature fluxes up to three derivatives in $L^2$ on the spacelike slices, and the $\Gamma$'s in $L^2$ on the spheres $S^2_{u,v}$, the first and the fourth term are easily controlled. 
For the second and third we note
\[
\| \mathfrak{D}^2 \psi \mathfrak{D}\Gamma \|^2_{L^2\left(S^2_{u,v}\right)} \lesssim \| \mathfrak{D}\Gamma \|^2_{L^\infty\left(S^2_{u,v}\right)}\| \mathfrak{D}^2 \psi \|^2_{L^2\left(S^2_{u,v}\right)} \lesssim e^{-\frac{P}{2M} \tau\left(u,v\right)}\| \mathfrak{D}^2 \psi \|^2_{L^2\left(S^2_{u,v}\right)}
\]
\[
\| \mathfrak{D} \psi \mathfrak{D}^2\Gamma \|^2_{L^2\left(S^2_{u,v}\right)} \lesssim \| \mathfrak{D}\psi \|^2_{L^\infty\left(S^2_{u,v}\right)}\| \mathfrak{D}^2 \Gamma \|^2_{L^2\left(S^2_{u,v}\right)} \lesssim e^{-\frac{P}{2M} \tau\left(u,v\right)}\| \mathfrak{D}^2 \Gamma \|^2_{L^2\left(S^2_{u,v}\right)}
\]
using the Sobolev inequalities of Sections \ref{sec:sob1} and \ref{sec:sob2}, in particular (\ref{l4r2}) and the equation below (\ref{l4c}). Hence after integration in $u$ and $v$, these terms can be estimated from the $L^2$-bootstrap assumptions.

\begin{remark}
It is possible to close the estimates with $s=2$, i.e.~commuting the equations only twice and using the $L^4$-estimates (\ref{l4r}) and (\ref{l4c}) on the error-term. While this would save one derivative, it would make the error-estimates more complicated which is why we work with $s\geq 3$.
\end{remark}

\subsection{Improving the assumptions on the Ricci-coefficients}  \label{sec:ricciimprove}
With the assumptions on the curvature fluxes improved, we can turn to improving the bootstrap assumption (\ref{decl}) on the Ricci-coefficients. The Ricci-coefficients are estimated via the transport equations along the null-directions. Before proving the key proposition, we derive two elementary Lemmas in Sections \ref{3dir} and \ref{4dir} revealing the ``gain'' obtained from integration in null-directions. Recall the definition of $u_{fut} \left(v\right)$ and $v_{fut}\left(u\right)$ in (\ref{ufutdef}).

\subsubsection{Integration in the $3$-direction} \label{3dir}
\begin{lemma} \label{lem:3dir}
Let $\Delta_p$ be a quantity satisfying 
\begin{equation} 
|\Delta_p \left(\hat{u}, \hat{v}, \theta^1, \theta^2 \right)| \leq e^{-P\frac{\tau}{2M}} \cdot \frac{\left(2M\right)^p}{r^p} \left(\hat{u},\hat{v}\right)
\end{equation}
for any spacetime point $\left(\hat{u},\hat{v}, \theta^1, \theta^2 \right) \in \mathcal{B}$ of our domain with $P\geq 2$. Then the estimate
\begin{equation*}
\int_{u}^{u_{fut}\left(v\right)}  |\Delta_p| \Omega^2 \left(\tilde{u},v,\theta^1, \theta^2 \right) d\tilde{u} \leq 4\cdot M \cdot \frac{1}{P}  e^{-\frac{P}{2M}\tau\left(u,v\right)} \frac{\left(2M\right)^p}{r^p}
\end{equation*}
holds for $\left(u,v, \theta^1, \theta^2 \right) \in \mathcal{B}$.
\end{lemma}
\begin{remark}
In applications, $\Delta_p$ will typically be an (appropriately $r$-weighted) $L^2\left(S^2 \left(u,v\right)\right)$ norm on the $\Gamma$'s.
\end{remark}
\begin{proof}
Note that along hypersurfaces of constant $v$ we have the global estimate $ \Omega^2 du \leq 2 d\tau$. Therefore, with $\tau_{(fut)}$ determined by $u_{(fut)}$ and $v$,
\begin{align}
\int_{u}^{u_{fut}(v)}  |\Delta_p| \Omega^2 d\tilde{u} \leq  \int_{\tau}^{\tau_{fut}} \left[e^{-P\frac{\tau}{2M}} \right] \cdot \frac{\left(2M\right)^p}{r^p} 2 d\tau
\nonumber \\
 = \int_{\tau}^{\tau_{fut}} -\frac{2M}{P} \partial_\tau \left[e^{-P\frac{\tau}{2M}} \right] \cdot \frac{\left(2M\right)^p}{r^p} 2 d\tau \nonumber \\
\leq \frac{4M}{P} \left[e^{-P\frac{\tau}{2M}} \right] \cdot \frac{\left(2M\right)^p}{r^p} \left(u,\hat{v}\right) + \int_{\tau}^{\tau_{fut}} -\frac{4M}{P} \left[e^{-P\frac{\tau}{2M}} \right] \cdot \frac{\left(2M\right)^p}{r^{p+1}} r_\tau \ d\tau \, .
\end{align}
Using that $|r_\tau|\leq 2$, $r \geq 2M$ we can absorb the last term by the second term in the first line to find
\[
\int_{u}^{u_{fut}\left(v\right)}  |\Delta_p| \Omega^2 du \leq M \frac{1}{P} \frac{4}{2-\frac{2}{P}}  \left[e^{-P\frac{\tau}{2M}} \right] \cdot \frac{\left(2M\right)^p}{r^p} \left(u,{v}\right) \, .
\] 

\end{proof}
\subsubsection{Integration in the $4$-direction} \label{4dir}
For the transport in the (backwards) four-direction we can derive a similar estimate; however, here we will ``lose'' one power of $r$ when integrating from null-infinity.
\begin{lemma} \label{lem:4dir}
Let $\Delta_p$ ($p\geq 0$) be a quantity satisfying 
\begin{equation} 
|\Delta_p \left(\hat{u}, \hat{v}, \theta^1, \theta^2 \right)| \leq e^{-P\frac{\tau}{2M}} \cdot \frac{\left(2M\right)^p}{r^p} \left(\hat{u},\hat{v}\right)
\end{equation}
in $\mathcal{B}$. Then, for $h \leq \tilde{h} \leq 2$, we have the estimate
\begin{equation} \label{vrefe}
\int_{v}^{v_{fut}\left(u\right)}  \frac{M^{\tilde{h}-1}}{r^{\tilde{h}}} \cdot \Delta_p dv \leq  2 \cdot 4^{h} \frac{1}{P} \frac{(2M)^{p}}{r^{p}} e^{-P\frac{\tau}{2M}}  \left(u,v\right) \, .
\end{equation}
Moreover, we have for $p > 0$ the estimate
\begin{equation} \label{vrefe2}
\int_{v}^{v_{fut}\left(u\right)}  \frac{1}{r} \cdot \Delta_p dv \leq  \left(\frac{2}{p} + \frac{4}{P} \right) \frac{(2M)^{p}}{r^{p}} e^{-P\frac{\tau}{2M}}  \left(u,v\right) \, .
\end{equation}
\end{lemma}
\begin{remark}
The estimate (\ref{vrefe2}) will become relevant for the anomalous boxed term in Proposition \ref{uneq}. 
\end{remark}
\begin{proof}
Note that along constant $u$ hypersurface we have the global estimate $dv \leq 2 \cdot 4^h \left( \frac{r}{2M}\right)^{h} d\tau$. Therefore,
\begin{align} \label{im1}
\int_{v}^{v_{fut}\left(u\right)} \frac{(2M)^{\tilde{h}-1}}{r^{\tilde{h}}} \Delta_p \left(u,\tilde{v}\right) d\tilde{v} \leq  2 \cdot 4^h  \int_{v}^{v_{fut}} \frac{(2M)^{p-1}}{r^{p}}  e^{-P\frac{\tau}{2M}} \left(u,\tilde{v}\right) d\tau \nonumber \\
\leq   2 \cdot 4^h  \frac{(2M)^{p-1}}{r^{p}} \left(u,v\right) \int_{v}^{v_{fut}}  e^{-P\frac{\tau}{2M}} d\tau \leq  2 \cdot 4^h \frac{1}{P} \frac{(2M)^{p}}{r^{p}} e^{-P\frac{\tau}{2M}}  \left(u,v\right)\, , 
\end{align}
which is (\ref{vrefe}). For the second statement we consider first the case that the $v$-value we are integrating up to, say $v^\star$, satisfies $r\left(u,v^\star\right) \geq 8M$. In this region $1 \geq 2r_v$ and since moreover $\tau$ increases in $v$ we obtain
 \begin{align} \label{pu}
\int_{v}^{v_{fut}} \frac{1}{r} \Delta_p \left(u,\tilde{v}\right) d\tilde{v} \leq  e^{-P\frac{\tau}{2M}} \left(u,v\right)  \int_{v^\star}^{v_{fut}} (2M)^p \frac{2r_v}{r^{p+1}}  dv
\end{align}
establishing (\ref{vrefe2}) for the region considered. In case that the $v$ value lies in the region $r\left(u,v^\star\right) \leq 8M$ we have to add to (\ref{pu}) an additional contribution. However, in this region $dv \leq 2 d\tau$ and one exploits the exponential decay as previously.
\end{proof}

\subsubsection{The key proposition}

\begin{proposition} \label{prop:improve2}
For any point $\left(u,v, \theta_1,\theta_2\right)$ in the bootstrap region $\mathcal{B}$, we have, for any $k=0,1,2,3$ the estimates
\begin{align} \label{yuh}
w^{2p-2} \| \mathfrak{D}^k \overset{(3)}{\Gamma}_p \|_{L^2\left(S^2\left(u,v\right)\right)} 
\leq \left(4+\frac{C_M}{P} + C_M e^{-\frac{P}{2M} \tau_0} \right) e^{-\frac{P}{2M}\tau\left(u,v\right)}
\end{align}
and 
\begin{align} \label{yah}
w^{2p-2} \| \mathfrak{D}^k \overset{(4)}{\Gamma}_p \|_{L^2\left(S^2\left(u,v\right)\right)} 
\leq  \left(4+\frac{C_M}{P} + C_M e^{-\frac{P}{2M} \tau_0} \right) e^{-\frac{P}{2M}\tau\left(u,v\right)} \nonumber \\
+ \left(\frac{2}{p} + \frac{4}{P}\right) \sqrt{C_{max}^2} \sqrt{4+ \frac{C_M}{P} + C_M e^{-\frac{P}{2M} \tau_0}} e^{-P\frac{\tau}{2M}} \, .
\end{align}
\end{proposition}

\begin{proof}
Note that for a scalar function $f$ we have the identities (cf.~Lemma \ref{lem:ncs})
\[
\partial_u \left[ \int_{S^2_{u,v}} f \sqrt{\slashed{g}} d\theta^1 d\theta^2 \right] =  \int_{S^2_{u,v}}  \left[ \slashed{\nabla}_3 f + tr \underline{\chi} f \right] \Omega^2 \sqrt{\slashed{g}} d\theta^1 d\theta^2 \, ,
\]
\[
\partial_v \left[ \int_{S^2_{u,v}} f \sqrt{\slashed{g}} d\theta^1 d\theta^2 \right] =  \int_{S^2_{u,v}}  \left[ \slashed{\nabla}_4 f + tr {\chi} f \right]  \sqrt{\slashed{g}} d\theta^1 d\theta^2 \, .
\]
From the first identity it follows from Proposition \ref{prop:nscommute} that
\begin{align} \label{udi} 
\partial_u \left[ \int_{S^2 \left(u,v\right)} w^{2p-2} \| \mathfrak{D}^k \overset{(3)}{\Gamma}_p\|^2 \sqrt{\slashed{g}} d\theta_1 d\theta_2\right] \nonumber \\ =  \int_{S^2 \left(u,v\right)}2  \left( E_3[\mathfrak{D}^k \overset{(3)}{\Gamma}_p] + f_1 \mathfrak{D}^k \overset{(3)}{\Gamma}_p\right) \cdot w^{2p-2} \mathfrak{D}^k \overset{(3)}{\Gamma}_p \ \Omega^2 \sqrt{\slashed{g}} d\theta_1 d\theta_2 \, .
\end{align}
Upon integration this yields
\begin{align}
w^{2p-2} \left(u,v\right) \| \mathfrak{D}^k \overset{(3)}{\Gamma}_p\|^2_{L^2\left(S\left(u,v\right)\right)} \leq w^{2p-2} \left(u_{fut},v\right) \| \mathfrak{D}^k \overset{(3)}{\Gamma}_p\|^2_{L^2\left(S\left(u_{fut},v\right)\right)}  \nonumber \\
+ \int_u^{u_{fut}} d\bar{u} \int_{S^2 \left(\bar{u},v\right)}2 \Big|  \left( E_3 [\mathfrak{D}^k \overset{(3)}{\Gamma}_p] + f_1 \mathfrak{D}^k \overset{(3)}{\Gamma}_p\right) \cdot w^{2p-2} \mathfrak{D}^k \overset{(3)}{\Gamma}_p\Big|  \Omega^2 \sqrt{\slashed{g}} d\theta_1 d\theta_2  \nonumber \, .
\end{align}
For the second term in the second line we can insert the bootstrap assumption and apply Lemma \ref{lem:3dir} which produces
\begin{align} \label{ste1}
\int_u^{u_{fut}} \int_{S^2 \left(u,v\right)}  \|\mathfrak{D}^k\overset{(3)}{\Gamma}_p\|^2 w^{2p-2}   \Omega^2 \sqrt{\slashed{g}} du d\theta_1 d\theta_2 \leq \frac{C_M}{P} e^{-P\frac{\tau}{2M}} \left(u,v\right)  \, .
\end{align}
Moreover, inspecting the various terms of $E_3 [\mathfrak{D}^k \overset{(3)}{\Gamma}_p]$ we see that the resulting integrand in (\ref{udi}) either (after applying Sobolev inequalities as in the previous section) satisfies the assumptions of Lemma \ref{lem:3dir}, or it is the curvature flux, which we already improved in the previous section.\footnote{We will provide more details on how to control the error when discussing the $4$-direction which is more difficult since careful track of the $r$-weights has to be kept.} Hence
\begin{align} 
\int_{S^2 \left(u,v\right)}  \|\mathfrak{D}^k \overset{(3)}{\Gamma}_p \|^2 w^{2p-2}  \sqrt{\slashed{g}} d\theta_1 d\theta_2 \leq \left(1 + 2 + \frac{C_M}{P} + C_M e^{-\frac{P}{2M}\tau_0} \right)  e^{-\frac{P}{2M} \tau} \, \nonumber
\end{align}
with the $1$ coming from the data and the other terms from the already improved curvature flux and the terms to which Lemma \ref{lem:3dir} has been applied. In particular, $C_M$ depends on the number of (linear) terms involved. This proves (\ref{yuh}).


For the $\mathfrak{D}^k\overset{(4)}{\Gamma}_p$ we have from Proposition \ref{prop:nscommute} and the renormalization of Remark \ref{rem:reno}
\begin{align} \label{fot}
\frac{1}{2} \partial_v  \left[ \int_{S^2 \left(u,v\right)} w^{-2} w^{4c[\overset{(4)}{\Gamma}_p]}  \| \mathfrak{D}^k \overset{(4)}{\Gamma}_p\|^2 \sqrt{\slashed{g}} d\theta_1 d\theta_2\right] = \nonumber \\
\int_{S^2 \left(u,v\right)} \left(\tilde{f}_2 + \Gamma_2\right) w^{-2} w^{4c[\overset{(4)}{\Gamma}_p]}  \| \mathfrak{D}^k \overset{(4)}{\Gamma}_p\|^2 \sqrt{\slashed{g}} d\theta_1 d\theta_2 \nonumber \\ + \int_{S^2 \left(u,v\right)}  E_4 [\mathfrak{D}^k \overset{(4)}{\Gamma}_p]  \ w^{-2} w^{4c[\overset{(4)}{\Gamma}_p]} \mathfrak{D}^k \overset{(4)}{\Gamma}_p \ \sqrt{\slashed{g}} d\theta_1 d\theta_2 \, .
\end{align}
Since $2c[\overset{(4)}{\Gamma}_p] \leq p$ for the $\overset{(4)}{\Gamma}_p$ involved, we can, after integration, apply Lemma \ref{lem:4dir} to the term in the second line. For the term in the third line we recall from Proposition \ref{prop:nscommute} and (\ref{eq:aux1}) that
\begin{align} \label{e4stru}
E_4 [\mathfrak{D}^k\overset{(4)}{\Gamma}_p] = \mathfrak{D}^k \left(\textrm{boxed in (\ref{eq:aux1})}\right) + \mathfrak{D}^k \left(f_1 \Gamma_3\right) \nonumber \\
+ \sum_{i=0}^k \left( f_2 \mathfrak{D}^i \Gamma_{p} \right) + \mathfrak{D}^k \psi_{p+\frac{3}{2}} + \textrm{q.t.} \, ,
\end{align}
with the last term denoting quadratic terms. The latter are very easy to handle and we will not treat them explicitly. (They exhibit sufficient decay in $r$ (as is manifest from the fact that $E_4$ gains a power of at least $\frac{3}{2}$ in radial decay) and will moreover always contain a smallness factor arising from $\tau_0$ large in view of their quadratic nature. See also the discussion at the end of Section \ref{sec:curvimprove}.) 

Let us  ignore the anomalous contribution from the boxed term as well as the term $\mathfrak{D}^k \left( f_1 \Gamma_3\right)$ in (\ref{e4stru}) for the moment. Then, for the sum-term in (\ref{e4stru}) we can apply Cauchy's inequality to the third line of (\ref{fot}) so as to produce a term as the one in the second line of (\ref{fot}) and the expression
\begin{align} \label{gt}
\int_{v}^{v_{fut}} dv \int_{S^2 \left(u,v\right)} \|  f_2 \mathfrak{D}^i\Gamma_p \|^2 w^{4c[\overset{(4)}{\Gamma}_p]}  \sqrt{\slashed{g}} d\theta_1 d\theta_2 \leq \frac{C_M}{P} e^{-P\frac{\tau}{2M}} \left(\frac{2M}{r}\right)^{2p-4c[\overset{(4)}{\Gamma}_p]} \, . 
\end{align}
This last inequality holds for all $i = 0, ... , k$ because Lemma \ref{lem:4dir} applies to the left hand side (note $2c[\overset{(4)}{\Gamma}_p] \leq p$ and $w^2 (f_2)^2 \sim w^{-2}$). On the other hand, for the curvature term in (\ref{e4stru}) applying again Cauchy's inequality to (\ref{fot}) yields
\begin{align} 
\int_{v}^{v_{fut}} dv \int_{S^2 \left(u,v\right)} \| \mathfrak{D}^k \psi_{p+\frac{3}{2}} \|^2 w^{4c\left[\Gamma_p\right]}  \sqrt{\slashed{g}} d\theta_1 d\theta_2 \nonumber \\
 \leq \left(2+ \frac{C_M}{P} +  C_M e^{-P\frac{\tau_0}{2M}}\right) e^{-P\frac{\tau}{2M}} \left(\frac{2M}{r}\right)^{2p-4c[\overset{(4)}{\Gamma}_p]}   \, .  \nonumber
\end{align}
Note that in this direction one can always take decaying $r$-weights out of the integral. 

We now turn to the contribution from the anomalous boxed term in (\ref{e4stru}) as well as the term $\mathfrak{D}^k \left( f_1 \Gamma_3\right)$ respectively (the latter appearing only for the $\mathfrak{D}^k \left(tr \chi - tr \chi_\circ\right)$ equation).

Explicitly, the three terms to consider are (cf.~Remark \ref{rem:anoterm}):
\begin{align} \label{boxest}
\int_v^{v_{fut}} dv \int_{S^2 \left(u,v\right)}  w^{4c\left[\eta\right]=2}  \left[ \mathfrak{D}^k \left(\frac{1}{2} tr \chi_\circ \underline{\eta} \right) \right] \cdot \left[ \mathfrak{D}^k \eta \right] w^{-2} \sqrt{\slashed{g}} d\theta_1 d\theta_2 \, \,
\end{align}
and
\begin{align} \label{boxest2}
\int_v^{v_{fut}} dv \int_{S^2 \left(u,v\right)} w^0 \frac{1}{2M} \left[ \mathfrak{D}^k \left(\frac{\Omega^2_\circ}{\Omega^2}-1 \right)  \right] \frac{1}{2M} \left[ \mathfrak{D}^k \left( \hat{\omega}-\hat{\omega}_\circ \right) \right] w^{-2} \sqrt{\slashed{g}} d\theta_1 d\theta_2
\end{align}
as well as 
\begin{align} \label{boxest3}
\int_v^{v_{fut}} dv \int_{S^2 \left(u,v\right)} w^{2} \left[ \mathfrak{D}^k \left( tr \chi_\circ \left(\hat{\omega}- \hat{\omega}_\circ \right) \right) \right] \left[ \mathfrak{D}^k \left(tr \chi - tr \chi_\circ\right) \right]\sqrt{\slashed{g}} d\theta_1 d\theta_2 \, .
\end{align}
Going back to our unifying notation, we have to estimate for $p_1+p_2=p$ (note that we either have $p_1=1$ or $p_1=0$) the expression
\begin{align}
 \int_{v}^{v_{fut}} dv \int_{S^2 \left(u,v\right)} \mathfrak{D}^k \left( f_{p_1}  \overset{(3)}{\Gamma}_{p_2} \right) \cdot w^{-2} w^{4c[\overset{(4)}{\Gamma}_p]} \mathfrak{D}^k \overset{(4)}{\Gamma}_p \ \sqrt{\slashed{g}} d\theta_1 d\theta_2 \nonumber \\
\leq \frac{C_{max}}{8} \sum_{8 \ terms}  \int_{v}^{v_{fut}} dv \int_{S^2 \left(u,v\right)} \frac{(2M)^{p_1-1}}{r^{p_1}} w^{2p-2}  \| \mathfrak{D}^i  \overset{(3)}{\Gamma}_{p_2} \| \| \mathfrak{D}^k \overset{(4)}{\Gamma}_p \| \sqrt{\slashed{g}} d\theta_1 d\theta_2 \nonumber
\end{align}
with $f_{p_1}$ as in (\ref{boxest})--(\ref{boxest3}).
This estimate follows from the fact that $2c[\overset{(4)}{\Gamma}_p] \leq p$ as well as the definition of the constant $C_{max}$. For $k=3$ the maximum number of terms that the round bracket can produce is eight and each comes with an $i \in \{0,1,2,3\}$.

Applying Cauchy-Schwarz to the previous expression yields
\begin{align}
 &\leq \frac{1}{8} \sum \sqrt{\int_{v}^{v_{fut}} dv \frac{1}{r} w^{2p-2} \| \mathfrak{D}^k \overset{(4)}{\Gamma}_p \|^2_{L^2\left(S^2(u,v)\right)}}\sqrt{\int_{v}^{v_{fut}} dv \frac{1}{r} w^{2p_2-2} \| \mathfrak{D}^i \overset{(3)}{\Gamma}_{p_2} \|^2_{L^2\left(S^2(u,v)\right)}} \nonumber \\
&\leq \left(\frac{2}{p} + \frac{C}{P}\right) \sqrt{C_{max}^2} \sqrt{4+ \frac{C_M}{P}+C_M e^{-\frac{P}{2M}\tau_0}} e^{-P\frac{\tau}{2M}} \, .
\end{align}
This follows after inserting the bootstrap assumption for the first root and the already established estimate (\ref{yuh}) for the second, and applying the estimate (\ref{vrefe2}) of Lemma \ref{lem:4dir} to both of these terms.
\end{proof}
\subsection{Closing the bootstrap}
Choosing $P$ and $\tau_0$ sufficiently large we can improve the bootstrap assumptions on curvature (\ref{decl2}) and (\ref{decl3}) by a factor of $\frac{3}{4}$ from Proposition \ref{prop:improve1}. For the Ricci-coefficients, Proposition \ref{prop:improve2} improves the bootstrap assumption (\ref{decl}) by a factor of $\frac{3}{4}$ provided that
\[
4+ \frac{C_M}{P} + C_M e^{-\frac{P}{2M}\tau_0} + 2 C_{max} \sqrt{4+ \frac{C_M}{P} + C_M e^{-\frac{P}{2M}\tau_0}} \leq \frac{3}{4} C_{max}^2
\]
which is indeed true for sufficiently large $P$ and $\tau_0$. Proposition \ref{prop:explain} is proven.

We finally note that once the bootstrap has closed we can improve the radial decay of some of the quantities further. In particular, using (\ref{3trc}) one shows that $tr \chi - tr \chi_\circ$ decays in fact like $r^{-3}$. Also, the full peeling decay for $\beta$ (i.e.~$r^{-4}$) and $\alpha$ (i.e.~$r^{-5}$) could be retrieved, as previously mentioned.
\section{Proof of Theorem \ref{theo2}: The convergence} \label{sec:convergence}
\subsection{Overview over the proof}
In this section we turn to the proof of Theorem \ref{theo2}. For this, we will need to consider differences of metric-, Ricci- and curvature components of two solutions $g_{\tau_1}$ and $g_{\tau_2}$ arising via Theorem \ref{theo1} from different data sets $D_{\tau_1,(v_\infty)_1}$ and $D_{\tau_2,(v_\infty)_2}$, $\tau_i \geq \tau_0+1$. 

Let $\tau_n > \tau_0$ be a monotonically increasing sequence of numbers with $\tau_n \rightarrow \infty$ and $\left(v_\infty\right)_n = (\tau_n)^2$. Let $\left(\Gamma,\psi\right)_n$ denote the solution arising from the approximate scattering data set $D_{\tau_n,\left(v_\infty\right)_n}$ via Theorem \ref{theo1}. The geometric situation is depicted in the figure below.
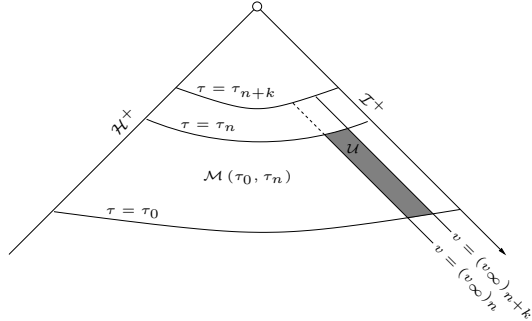
\[
\input{convergence.pstex_t}
\]
\\
Let $\mathcal{M}_{\tau,\bar{v}}:=  \mathcal{M}\left(\tau_0,\tau, \bar{v} \right)$. Clearly, for any $\tau>\tau_0$ and any $\bar{v}>v_{\tau_0}$, there is an $N$ such that the sequence of solutions $\left(\Gamma,\psi\right)_n$ for $n\geq N$ is defined in $\mathcal{M}_{\tau,\bar{v}}$. We claim that the sequence $\left(\Gamma,\psi\right)_n$ converges (in a suitable space) in $\mathcal{M}_{\tau,\bar{v}}$ for any $\tau,\bar{v}$. The convergence theorem will be established below along the following lines:
\begin{enumerate}[label=Step \arabic{enumi}.,ref=Step \arabic{enumi}]
\item Consider the solution  $\left(\Gamma, \psi\right)_{n+k}$ in the region $\mathcal{U}=\mathcal{M}\left(\tau_0,\tau_n\right) \cap \{ \left(v_\infty\right)_n \leq v\leq \left(v_\infty\right)_{n+k}\}$. Prove estimates for the solution which capture in a quantitative way that both curvature and Ricci-coefficients induced on $(v_\infty)_n$ are still close to the scattering data imposed on $v=(v_{\infty})_{n+k}$. This requires renormalization with the radiation fields of Section \ref{sec:radfields} (corresponding to subtracting the first term in an asymptotic expansion in $r^{-1}$ from each $\Gamma_p$ and $\psi_p$ so that the weighted \emph{difference} indeed goes to zero in the limit).

\item Derive the general equations (null-structure and Bianchi) for the differences of two solutions $\left(\Gamma,\psi\right)$ and $\left(\Gamma^\dagger,\psi^\dagger\right)$.

\item Use Step 2 to estimate the difference of two solutions $\left(\Gamma,\psi\right)_n$,$\left(\Gamma,\psi\right)_{n+k}$ in the region where they are both defined, that is in $\mathcal{M}\left(\tau_0,\tau_n\right) \cap \{v\leq \left(v_\infty\right)_n\}$. More precisely, prove $L^2\left(S^2\left(u,v\right)\right)$ estimates for the difference of $\Gamma$'s and curvature flux estimates for the difference of the $\psi$'s that depend solely on the ``data" on $\Sigma_{\tau_n}$, the horizon and the hypersurface $v=\left(v_\infty\right)_n$. Since all these go to zero as $n\rightarrow \infty$ (the latter by Step 1), convergence in (a weighted) $L^2 \left(S^2_{u,v}\right)$ follows. This step loses one derivative due to the quasi-linear nature of the equations, cf.~Remark \ref{rem:loss}. 

\item Either by commuting $s-1$ times and repeating the above steps, or by simply appealing to embedding theorems (recall that $s$ derivatives of the $\Gamma_n$ are uniformly bounded in $L^2_{S^2_{u,v}}$), one infers in particular that $s-1$ derivatives of the $\Gamma_n$ also converge in $L^2_{S^2_{u,v}}$. This yields the statement of the Theorem. Note that for $s=3$, this regularity implies that ${\rm Ric}\left(g\right)$ is indeed defined in $L^\infty_u L^\infty_v L^2_{S^2_{u,v}}$. Finally, Propositions \ref{prop:reno} and \ref{prop:reno2} yield (\ref{realizedata}).
\end{enumerate}

\subsection{Step 1: Controlling the flux on $v=(v_\infty)_n$}
In this section, we estimate the solution $\left(\Gamma, \psi\right)_{n+k}$ in the region $\mathcal{U}=\mathcal{M}\left(\tau_0,\tau_n\right) \cap \{ \left(v_\infty\right)_n \leq v\leq \left(v_\infty\right)_{n+k}\}$. Let us fix $0<\delta<1$ small.

\subsubsection{Estimates for curvature}
To obtain suitable estimates, we renormalize the Bianchi equations by subtracting from each curvature component its radiative term defined in Section \ref{sec:radfields2}. The quantity thereby obtained is expected to go to zero in the limit as $r \rightarrow \infty$ which makes it useful in the analysis. We present the renormalization process in detail for one Bianchi pair. For the other pairs, the computation is entirely analogous, with the exception of the $(\alpha, \beta)$-pair, which we will not renormalize.\footnote{Instead, the estimate for this pair will be applied with a slightly weaker weight gaining a convergence factor. Cf.~footnote \ref{foot:extra}.} We write
\begin{align}
\slashed{\nabla}_4 \left(\underline{\alpha} - \underline{\alpha}^{\mathcal{I}} \right) + \frac{1}{2} tr \chi \left(\underline{\alpha} -\underline{\alpha}^{\mathcal{I}} \right) + 2 \hat{\omega} \left(\underline{\alpha}  -\underline{\alpha}^{\mathcal{I}} \right) &=  2 \slashed{\mathcal{D}}_2^\star \left(\underline{\beta}-\underline{\beta}^{\mathcal{I}}  \right) \nonumber \\
+ E_4 \left[\underline{\alpha}\right] - \left( \slashed{\nabla}_4 \underline{\alpha}^{\mathcal{I}} + \frac{1}{2} tr \chi  \underline{\alpha}^{\mathcal{I}} \right) -2\hat{\omega} \underline{\alpha}^{\mathcal{I}}  + 2 \slashed{\mathcal{D}}_2^\star\underline{\beta}^{\mathcal{I}}  \, .
\end{align}
Note that the terms in the last line are all exponentially decaying in $\tau$, and decaying like $\frac{1}{r^3}$ because of a cancellation in the round bracket of the last line:
\[
\| \slashed{\nabla}_4 \underline{\alpha}^{\mathcal{I}} + \frac{1}{2} tr \chi  \underline{\alpha}^{\mathcal{I}} \| = \| \partial_v \left(-r \partial_u \partial_u \hat{\slashed{g}}^{dat\mathcal{I}}_{AB} \right) - \frac{1}{2} tr \chi  \underline{\alpha}^{\mathcal{I}}_{AB}\| + \frac{C}{r^3}e^{-P\frac{\tau}{2M}} \leq \frac{C}{r^3}e^{-P\frac{\tau}{2M}} \, .
\]
Similarly, we write the corresponding equation
\begin{align}
\slashed{\nabla}_3  \left(\underline{\beta} -\underline{\beta}^{\mathcal{I}}  \right) + 2 tr \underline{\chi}   \left(\underline{\beta} -\underline{\beta}^{\mathcal{I}} \right) = - \slashed{div}\left(\underline{\alpha}  - \underline{\alpha}^{\mathcal{I}} \right) \nonumber \\
 + E_3 \left[ \underline{\beta}\right] + \slashed{div}  \underline{\alpha}^{\mathcal{I}} - \left(\slashed{\nabla}_3 \underline{\beta}^{\mathcal{I}} + 2 tr \underline{\chi}\underline{\beta}^{\mathcal{I}} \right) \, .
\end{align}
This time there is no cancellation in the last line, which decays only like $r^{-2}$. This is not problematic, however, as the energy estimate contracts this equation with $\left(\underline{\beta} -\underline{\beta}^{\mathcal{I}}  \right)$ which decays (at least) like $r^{-2}$ as well. More precisely, applying the energy estimate in the region $\mathcal{U}$ to the Bianchi-pair $\left(\underline{\alpha},\underline{\beta}\right)$, we easily see that the spacetime-terms satisfy (the volume form being $\Omega^2 \sqrt{\slashed{g}} du dv d\theta_1 d\theta_2 \sim \sqrt{\slashed{g}} du dv d\theta_1 d\theta_2$ by Theorem \ref{theo1} and the fact that $r$ is large)
\[
\int_{\mathcal{U}} \Big\{ E_4 \left[\underline{\alpha}\right] - \left( \slashed{\nabla}_4 \underline{\alpha}^{\mathcal{I}} + \frac{1}{2} tr \chi  \underline{\alpha}^{\mathcal{I}} \right) -2\hat{\omega} \underline{\alpha}^{\mathcal{I}}  + 2 \slashed{\mathcal{D}}_2^\star\underline{\beta}^{\mathcal{I}} \Big\}\left(\underline{\alpha} - \underline{\alpha}^{\mathcal{I}}  \right)  \leq \frac{C}{\tau_n} e^{-\frac{P}{2M}\tau} \, ,
\]
\[
\int_{\mathcal{U}} \Big\{ E_3 \left[ \underline{\beta}\right] + \slashed{div}  \underline{\alpha}^{\mathcal{I}} - \left(\slashed{\nabla}_3 \underline{\beta}^{\mathcal{I}} + 2 tr \underline{\chi}\underline{\beta}^{\mathcal{I}} \right)\Big\} \left(\underline{\beta} -\underline{\beta}^{\mathcal{I}} \right)  \leq \frac{C}{\tau_n} e^{-\frac{P}{2M}\tau}
\]
using only the uniform estimates promised by Theorem \ref{theo2} and with the gain arising from the additional $r$-weight gained by the renormalization.

The equations for the other Bianchi pairs are renormalized and estimated analogously, cf.~(\ref{pi2}) and (\ref{pi3}). The only subtelty arises in the pair $\left(\alpha,\beta\right)$, which we will not renormalize.\footnote{The subtelty is caused by the fact that we are not imposing the full peeling decay, i.e.~we did not apply (\ref{pi1}) with weight $q=6$. \label{foot:extra}} Instead, we apply the estimate corresponding to (\ref{pi1}) with weight $5-\delta$. In view of 
\[
\int_{\mathcal{U}} f_1 w^{5-\delta} \|\mathfrak{D}^k \beta\|^2 \leq C \frac{1}{r^\delta\left(\tau_n,(v_\infty)_n\right)} e^{-\frac{P}{2M}\tau} \leq \frac{C}{\tau^{2\delta}_n} e^{-\frac{P}{2M}\tau}
\]
as $(v_\infty)_n = \tau_n^2$ we finally obtain (applying Proposition \ref{prop:asymptotic} to control the boundary terms on $\Sigma_{\tau_n} \cap \mathcal{U}$ and $\{v=(v_\infty)_n\} \cap \mathcal{U}$):

\begin{proposition} \label{prop:reno}
Let $\left(\Gamma, \psi\right)_{n+k}$ be a solution arising from Theorem \ref{theo1} and consider the solution in the region $\mathcal{U}$. Denoting the curvature components of the solution by $\beta \equiv \beta_{(n+k)}, ..., \underline{\alpha} \equiv \underline{\alpha}_{n+k}$ we have for for any $(v_\infty)_n \leq v \leq (v_\infty)_{n+k}$ and $u\left(\tau_0,v\right) \leq u \leq u\left(\tau_n, v\right)$ the estimate
\begin{align} \label{mdC}
\int_{u}^{u\left(\tau_n, v\right)} du \Big\{ \|\beta \|^2 w^{5-\delta} + w^4 \left(\left(\rho-\rho_\circ-\rho^{\mathcal{I}}\right)^2  +\left(\sigma-\sigma^{\mathcal{I}} \right)^2  \right) \nonumber \\
+ w^2  \| \underline{\beta} - \underline{\beta}^\mathcal{I} \|^2 +  \|\underline{\alpha} -\underline{\alpha}^{\mathcal{I}} \|^2 \Big\} \sqrt{\slashed{g}} d\theta^1 d\theta^2 \leq \frac{C}{\tau^{2\delta}_n} e^{-\frac{P}{2M}\tau\left(u,v\right)}
\end{align}
\end{proposition}
\begin{remark}
Note that the right hand side of (\ref{mdC}) goes to zero as $n \rightarrow \infty$. Note also that in view of  Proposition \ref{prop:asymptotic}, the same estimate holds replacing $\sigma^\mathcal{I}$ by $\sigma^{\mathcal{I}_{\tau_{n}}}$ etc.
\end{remark}

\subsubsection{Estimates for the Ricci coefficients}
We next turn to the Ricci coefficients $\Gamma_p$. 
\begin{proposition} \label{prop:reno2}
Under the assumptions of Proposition \ref{prop:reno}, we have in $\mathcal{U}$ the estimate
\begin{align} \label{mdG}
w^{2p-2}\| \left(\Gamma_p\right)_{n+k} -\Gamma_p^{\mathcal{I}} \|^2_{L^2_{\slashed{g}} \left(S^2(u,v)\right)} \leq  \frac{C}{\tau^{2\delta}_n} e^{-\frac{P}{2M}\tau\left(u,v\right)} \, .
\end{align}
In particular, the right hand sides goes to zero as $n\rightarrow \infty$.
\end{proposition}

\begin{proof}
Let us first consider the equations in the $3$-direction \emph{involving curvature}, i.e.~the equations for $\hat{\omega}-\hat{\omega}_\circ$, $\underline{\eta}$ and $\hat{\underline{\chi}}$. They are easily seen to be of the form
\[
\slashed{\nabla}_3 \left(\Gamma_p - \Gamma^\mathcal{I}_p\right) = \psi_p - \psi_p^\mathcal{I} + \mathcal{O}\left(\frac{e^{-\frac{P}{2M}\tau}}{r^{p+1}}\right)
\]
after renormalization. Applying Lemma \ref{lem:3dir} and using Proposition \ref{prop:reno} for the curvature flux-term that arises (and again Proposition \ref{prop:asymptotic} to control the boundary term), one establishes the estimate (\ref{mdG}) for these $\Gamma_p$. Next, considering the equation for $\left(tr \underline{\chi} - tr \underline{\chi}_\circ\right) - \left(tr \underline{\chi} - tr \underline{\chi}_\circ\right)^\mathcal{I}$, and $\slashed{g}_{AB} = \slashed{g}_{AB}^\circ - r \hat{\slashed{g}}^{dat\mathcal{I}}_{AB}$, we also obtain (\ref{mdG}) for those quantities, using the fact that we already obtained the estimate (\ref{mdG}) for $\hat{\underline{\chi}}- \hat{\underline{\chi}}^\mathcal{I}$ in the first step. Turning to the equations in the $4$-direction\footnote{Note that (\ref{beq}) is the only equation in the $3$-direction which has not yet been considered.} we see that
\[
\slashed{\nabla}_4 \left(\hat{\chi}-\hat{\chi}^\mathcal{I}\right) + tr \chi  \left(\hat{\chi}-\hat{\chi}^\mathcal{I}\right) =  \mathcal{O}\left(\frac{e^{-\frac{P}{2M}\tau}}{r^{\frac{7}{2}}}\right)
\]
\[
\slashed{\nabla}_4 \left(\eta -{\eta}^{\mathcal{I}}\right) = -\frac{3}{4} tr \chi  \left(\eta -{\eta}^{\mathcal{I}}\right) + \frac{1}{4} tr \chi \left(\underline{\eta} - \underline{\eta}^\mathcal{I}\right) + \frac{1}{4} tr \chi \left(\eta + \underline{\eta}\right) - \beta +  \mathcal{O}\left(\frac{e^{-\frac{P}{2M}\tau}}{r^{4}}\right)
\]
In view of $\|\eta + \underline{\eta}\|$ decaying like $r^{-3}$ and the $\left(\underline{\eta} -\underline{\eta}^{\mathcal{I}}\right)$-term already satisfying (\ref{mdG}), the estimate (\ref{mdG}) also follows for these quantities. Similarly for (\ref{4cofa}) and (\ref{trx4}). The latter does not need to be renormalized because the right hand side already decays more than two powers in $r$ better. Finally, from (\ref{beq}) one obtains the smallness estimate for the renormalized $b$ in view of the right hand side already satisfying (\ref{mdG}).
\end{proof}

\subsubsection{The final estimate}
The two previous propositions yield the important
\begin{corollary} \label{cor:diffestfinal}
Let $\left(\Gamma, \psi\right)_{n}$, $\left(\Gamma, \psi\right)_{n+k}$ be two solutions arising from Theorem \ref{theo1}. Then, on $v=\left(v_\infty\right)_n$, their difference satisfies for any $u\left(\tau_0, (v_\infty)_n\right) \leq u \leq u\left(\tau_n, (v_\infty)_n\right)$
\begin{align} 
\int_{u}^{u\left(\tau_n, (v_\infty)_n\right)} du \Big\{ \|\boldsymbol\beta\|^2 w^{5-\delta} + w^4 \left( \boldsymbol\rho^2 +\boldsymbol\sigma^2 \right)
+ w^2 \| \underline{\boldsymbol\beta} \|^2 +  \|\underline{\boldsymbol\alpha}\|^2 \Big\} \sqrt{\slashed{g}} d\theta^1 d\theta^2 \nonumber \\
\leq \frac{C}{\tau^{2\delta}_n} e^{-\frac{P}{2M}\tau\left(u,\left(v_\infty\right)_n\right)} \, , \nonumber
\end{align}
where $\boldsymbol \psi_p = \left(\psi_p\right)_{n+k} - \left(\psi_p\right)_{n}$ as well as
\begin{align}
w^{2p-2}\| \left(\Gamma_p\right)_{n+k} -  \left(\Gamma_p\right)_{n} \|^2_{L^2_{\slashed{g}} \left(S^2(u,v)\right)} \leq  \frac{C}{\tau^{2\delta}_n} e^{-\frac{P}{2M}\tau\left(u,\left(v_\infty\right)_n\right)} \,.
\end{align}
In particular, the right hand sides go to zero as $n \rightarrow \infty$.
\end{corollary}

\begin{proof}
Note that by construction of the approximate scattering data on $v=\left(v_\infty\right)_n$, the estimates (\ref{mdC}) and (\ref{mdG}) also hold for the the solution  $\left(\Gamma, \psi\right)_n$ on $v=\left(v_\infty\right)_n$. Then apply the triangle inequality.
\end{proof}
\subsection{Step 2: The equations for differences}
Let $g$ and $g^\dagger$ be two metrics arising from Theorem \ref{theo1}.
We define
\[
\boldsymbol\Gamma_p = \Gamma_p -\Gamma^\dagger_p \ \ \ \  and  \ \ \ \ \boldsymbol\psi_p = \psi_p - \psi^\dagger_p \ \ \ \  and  \ \ \ \ \boldsymbol G_p = G_p - G^\dagger_p
\]

\begin{remark} Note that with this definition $\hat{\underline{\boldsymbol\chi}}_{BC}=\hat{\chi}_{BC} - \hat{\chi}^\dagger_{BC}$ is not traceless with respect to $\slashed{g}$ but that the difference between this quantity and its $\slashed{g}$-traceless part is quadratic. Note also that $tr {\boldsymbol\chi}= tr \chi - tr {\chi}^\dagger$ is the difference of the traces as functions, \emph{not} the $\slashed{g}$-trace of the difference of the $\chi$-tensors. 
\end{remark}

When taking differences of the evolution equations, we will see in particular differences of the inverse metrics and differences of the Christoffel symbols $\slashed{\Gamma}$ on $S^2_{u,v}$.  Since $\slashed{g}$ and $\slashed{g}^\dagger$ are both invertible (in fact, close to the round metric), we have the matrix identity
\[
\slashed{g}^{-1} - \left(\slashed{g}^\dagger\right)^{-1} = - \slashed{g}^{-1} \left(\slashed{g} - \slashed{g}^\dagger \right)  \left(\slashed{g}^\dagger\right)^{-1}
\]
and hence
\[
\| \slashed{g}^{-1} - \left(\slashed{g}^\dagger\right)^{-1} \| \leq C\| \slashed{g} - \slashed{g}^\dagger \| \, .
\]
In view of this, we will allow ourselves to write $\boldsymbol \Gamma_1$ also to denote the difference $\slashed{g}^{-1} - \left(\slashed{g}^\dagger\right)^{-1}$. Similarly, in view of (\ref{chrel}) we will incorporate the difference $r \boldsymbol{\slashed{\Gamma}} = r \left(\slashed{\Gamma} - \slashed{\Gamma}^\dagger\right)$ into $\boldsymbol G_1$ when employing schematic notation.

Since both metrics $g$ and $g^\dagger$ are smooth and \emph{uniformly} close to Schwarzschild in the sense that by Theorem \ref{theo1} we have
\begin{align} \label{uni1}
\sum_{i=0}^k \sum_{i-perms}  \ w^{2p-2} \int_{S^2\left(u,v\right)} \| \mathfrak{D}^i \Gamma_p\|_{\slashed{g}}^2 \sqrt{\slashed{g}} d\theta^1 d\theta^2 \leq \epsilon \, 
\end{align}
\begin{align} \label{uni2}
F\left[\mathfrak{D}^k \Psi\right] \left( \{u\} \times [v,v_{fut}] \right) + F\left[\mathfrak{D}^k\psi\right] \left( \left[u,u_{fut}\right] \times \{v \}\right) \leq \epsilon \, 
\end{align}
for any $\left(u,v,\theta^1,\theta^2\right) \in \mathcal{M}\left(\tau_0,\tau_n\right) \cap \{v\leq \left(v_\infty\right)_n\} $ and 
for an $\epsilon$ that can be made small by choosing $\tau_0$ large, we have in particular uniform pointwise bounds on all quantities up to the level of three derivatives of the metric. With this in mind, we define $\lambda_{p_1}$ to denote a quantity which is schematically of the form\footnote{Here any contraction is taken with respect to the metric $g$. Note that the difference between contracting with $g$ and contracting with $\tilde{g}$ is small. One could also contract with $\tilde{g}$ and easily estimate the resulting errors.}
\begin{align}
\lambda_{p_1} = f_{p_1} +{\Gamma}^{(\dagger)}_{p_1} + (f \Gamma \cdot  {\Gamma}^{(\dagger)})_{p_1} + \psi_{p_1} + r^{p_1-1} \slashed{\Gamma} + \mathfrak{D} \Gamma_{p_1} \, .
\end{align}
In view of the above remark we have the uniform bound
\[
r^{p_1} \| \lambda_{p_1} \|_{L^\infty_{\slashed{g}}\left(S^2_{u,v}\right)} \leq C \, .
\]
The following proposition is the analogue of Proposition \ref{uneq} for differences:
\begin{proposition} \label{maind}
Let $g$ and $g^\dagger$ be two metrics arising from Theorem \ref{theo1}. Let $\Gamma_p$, $\Gamma^\dagger_p$ denote the respective collection of metric and Ricci coefficients (i.e.~$S^2_{u,v}-tensors$) and $\psi_p$, $\tilde{\psi}_p$ the collection of curvature components (also $S^2_{u,v}$-tensors). Then the differences $\boldsymbol\Gamma_p = \Gamma_p - \Gamma^\dagger_p$ and $\boldsymbol\psi_p = \psi_p - \psi^\dagger_p$ satisfy
\[
\slashed{\nabla}_3 \overset{(3)}{\boldsymbol\Gamma}_p =  \sum_{p_1+p_2\geq p} \lambda_{p_1} \boldsymbol\Gamma_{p_2} + \boldsymbol\psi_p \, ,
\]
\[
\slashed{\nabla}_4 \overset{(4)}{\boldsymbol\Gamma}_p + c [\overset{(4)}{\boldsymbol\Gamma}_p] tr \chi  \overset{(4)}{\boldsymbol\Gamma}_p =   \boxed{ \sum_{p_1+p_2\geq p+1} \lambda_{p_1} \overset{(3)}{\boldsymbol\Gamma}_{p_2}}  + \sum_{p_1+p_2\geq p+2} \lambda_{p_1}{\boldsymbol\Gamma}_{p_2} + \lambda_{p+2} \boldsymbol{\slashed{\Gamma}}+   \boldsymbol\psi_{p+\frac{3}{2}} \, .
\]
\end{proposition}

\begin{proof}
Let us drop the superscripts $(3)$ and $(4)$ during the proof.
Subtracting the two transport equations in the $3$-direction we obtain
\[
\slashed{\nabla}_3 \left(\Gamma_p -\Gamma^\dagger_p\right) = \left(\slashed{\nabla}^\dagger_3 - \slashed{\nabla}_3 \right)   \Gamma^\dagger_p + \sum_{p_1+p_2\geq p} \left(f_{p_1} + \Gamma^{(\dagger)}_{p_1}\right) \boldsymbol\Gamma_{p_2} + \boldsymbol\psi_p \, .
\]
For the first term, we have from Lemma \ref{teco} the expression
\[
\left(\slashed{\nabla}^\dagger_3 - \slashed{\nabla}_3 \right) \Gamma^\dagger_p = \left(\frac{1}{(\Omega^\dagger)^2} - \frac{1}{\Omega^2} \right) \cdot \partial_u \Gamma^\dagger_p + \left[ \left((g^\dagger)^{-1}\right) \underline{\chi}^\dagger - \left({g}^{-1}\right) {\underline{\chi}} \right] \Gamma^\dagger_p \, .
\]
Finally, we can insert back the transport equation for $\partial_u \Gamma^\dagger_p$ and write the square bracket as a sum of differences of $\Gamma$'s leading to the first equation of the Proposition. 

In the $4$-direction we proceed similarly. This time we are led to consider
\[
\left(\slashed{\nabla}^\dagger_4 - \slashed{\nabla}_4 \right)   \Gamma^\dagger_p = \left(b^\dagger -b \right) \partial_{\theta}  \Gamma^\dagger_p + \left[ \left((g^\dagger)^{-1}\right) {\chi}^\dagger - \left({g}^{-1}\right) {{\chi}} \right] \Gamma^\dagger_p + \left[ b^\dagger \slashed{\Gamma}^\dagger - b \slashed{\Gamma} \right] \Gamma^\dagger_p \, .
\]
Again, we can write the square-brackets as a sum of differences which yields the second equation of the Proposition.
\end{proof}

In view of the fact that the difference of the Christoffel-symbols appears in the equations for the difference, we will also need the equations for the difference of the auxiliary quantities $\boldsymbol{G}_i = G_i - G^\dagger_i$ defined in Section \ref{auxquant}.

\begin{proposition}
We have the equations
\begin{align} \label{auxd1}
\slashed{\nabla}_3 \boldsymbol{G}_1 = \slashed{\nabla}_3 \left(r\slashed{\nabla}_A \slashed{g}^\circ_{BC} - r \slashed{\nabla}^\dagger_A \slashed{g}^\circ_{BC}  \right) = \sum_{p_1+p_2\geq 1} \lambda_{p_1} \left(\boldsymbol\Gamma_{p_2} + \boldsymbol G_{p_2}\right) + 2 r \slashed{\nabla}_A\hat{\underline{ \boldsymbol\chi}}_{BC}
\end{align}
\begin{align} \label{auxd2}
\slashed{\nabla}_3 \boldsymbol{G}_2 = \slashed{\nabla}_3 \left(r\slashed{\nabla}_A tr \underline{\chi} - r{\slashed{\nabla}}_A tr \underline{\chi}^\dagger  \right) = \sum_{p_1+p_2\geq 2} \lambda_{p_1} \left(\boldsymbol\Gamma_{p_2} + \boldsymbol G_{p_2}\right) + 2 r \hat{\underline{\chi}} \slashed{\nabla}_A \hat{\underline{ \boldsymbol\chi}} \, . 
\end{align}
In addition, we have the elliptic equation for the difference
\begin{align} \label{ellipticd}
\left(\slashed{g}^{-1}\right)^{AB} \slashed{\nabla}_A \hat{\underline{\boldsymbol\chi}}_{BC} = \frac{1}{2} \slashed{\nabla}_A tr \underline{\boldsymbol\chi}  + \underline{\boldsymbol \beta} + \sum_{p_1+p_2\geq 2}  \lambda_{p_1} \boldsymbol \Gamma_{p_2} +  \hat{\underline{\chi}} \boldsymbol{\slashed{\Gamma}}
\end{align}

\end{proposition}
\begin{remark}
Note that in view of the estimate (\ref{chrel}) we identify $\boldsymbol{G_1} = r  \boldsymbol{\slashed{\Gamma}} $.
\end{remark}
\begin{proof}
This follows by taking the differences of the equations derived in Section \ref{auxquant} and using Lemma \ref{teco}, the computation proceeding in complete analogy to that of the previous proposition.
\end{proof}

Finally, we need the equations for curvature differences, i.e.~the analogue of Proposition \ref{uneq2}. For this we define (in analogy with $\lambda_{p_1}$) a quantity which depends on three derivatives of the metric. In particular, it depends on first derivatives of the curvature components: 
\[
\kappa_{p_1} = f_{p_1} + \Gamma_{p_1} +(f \Gamma \Gamma^{(\dagger)} )_{p_1} + \psi^\dagger_{p_1} +  \mathfrak{D} \psi^\dagger_{p_1}
\]
In any case, since both $g$ and $g^\prime$ are uniformly close to Schwarzschild  in the sense of (\ref{uni1}) and (\ref{uni2}), we have from Sobolev 
\[
r^{p_1} \| \kappa_{p_1} \|_{L^\infty_{\slashed{g}} \left(S^2_{u,v} \right)} \leq C \, .
\]

\begin{proposition} \label{prop:bianchidiff}
The Bianchi pairs of curvature differences $\left(\boldsymbol\uppsi_p,\boldsymbol\uppsi_{p^\prime}^\prime \right)$ satisfy
\[
\slashed{\nabla}_3 \boldsymbol\uppsi_p = \slashed{\mathcal{D}} \boldsymbol\uppsi^\prime_{p^\prime} + E_3 \left[\boldsymbol\psi_p\right] 
\]
with
\begin{align}
E_3 \left[\boldsymbol\uppsi_p\right] =  \sum_{p_1+p_2\geq p} \left( \lambda_{p_1} \boldsymbol\psi_{p_2} + \kappa_{p_1+1}  \boldsymbol\Gamma_{p_2} \right) + \lambda_{p^\prime}  \boldsymbol{\slashed{\Gamma}}
\end{align}
and 
\[
\slashed{\nabla}_4 \boldsymbol\uppsi^\prime_{p^\prime} + \gamma_4 \left(\boldsymbol\uppsi^\prime_{p^\prime}\right)  tr \chi \ \boldsymbol\uppsi^\prime_{p^\prime}  = \slashed{\mathcal{D}} \boldsymbol\uppsi_p + E_{4} \left[\boldsymbol\uppsi^\prime_{p^\prime}\right]
\]
with
\begin{align}
E_{4} \left[\boldsymbol\uppsi^\prime_{p^\prime}\right] =  \sum_{p_1+p_2\geq p^\prime + \frac{3}{2}} \left( \lambda_{p_1} \boldsymbol\psi_{p_2} + \kappa_{p_1}  \boldsymbol\Gamma_{p_2} \right) + \lambda_{p}  \boldsymbol{\slashed{\Gamma}}  \, .
\end{align}
\end{proposition}

\begin{remark} \label{rem:loss}
Note the loss of derivatives as the error depends on derivatives of curvature via the  $\kappa_{p_1}$.  Note also that in comparison with Proposition \ref{uneq2}, the $\boldsymbol{\slashed{\Gamma}}$-terms on the right hand side are new.
\end{remark}

\begin{proof}
For the $3$-equation we will have to estimate
\[
\left(\slashed{\nabla}^\dagger_3 - \slashed{\nabla}_3 \right)  \uppsi^\dagger_p =\left(\frac{1}{(\Omega^\dagger)^2} - \frac{1}{\Omega^2} \right) \cdot \partial_u \uppsi^\dagger_p + \left[ \left((g^\dagger)^{-1}\right) \underline{\chi}^\dagger - \left({g}^{-1}\right) {\underline{\chi}} \right] \uppsi^\dagger_p
\]
and (if the angular operator is $\slashed{div}$, for instance)
\[
\left[(g^\dagger)^{-1}\slashed{\nabla}^\dagger - {g}^{-1} {\slashed{\nabla}} \right] \uppsi_{p^\prime}^{\prime}= \left(\tilde{g}^{-1} - {g}^{-1} \right) \slashed{\nabla} \uppsi_{p^\prime}^{\prime} + {g}^{-1} \boldsymbol{\slashed{\Gamma}}  \uppsi_{p^\prime}^{\prime} \, .
\] 
The other angular operators are similar. In the $4$ direction we have to estimate
\[
\left(\slashed{\nabla}^\dagger_4 - \slashed{\nabla}_4 \right)  \uppsi_{p^\prime}^{\prime} = \left(b^\dagger -b \right) \partial_{\theta} \uppsi_{p^\prime}^{\prime} + \left[ \left((g^\dagger)^{-1}\right) {\chi}^\dagger - \left({g}^{-1}\right) {{\chi}} \right] \uppsi_{p^\prime}^{\prime} + \left[b^\dagger \slashed{\Gamma}^\dagger - b \slashed{\Gamma} \right] \uppsi_{p^\prime}^{\prime} \, ,
\]
which after writing the square brackets as sums of differences is of the desired form.
\end{proof}
\subsection{Step 3: Estimating differences} \label{sec:diffest}
We now redo the estimates proven for the individual solutions in Section \ref{sec:bootstrap}, for the differences of solutions. The idea is to bootstrap a slightly weaker exponential decay (replacing $P$ by $P/2$) and to use the gain as a smallness factor for the convergence.

\begin{proposition} \label{prop:conv}
Let $\left(\Gamma,\psi\right)_n$ be the (smooth) Einstein metric arising from the initial data set $D_{\tau_n,(v_\infty)_n}$ and $\left(\Gamma,\psi\right)_{n+k}$ be the (smooth) Einstein metric arising from $D_{\tau_{n+k},\left(v_\infty\right)_{n+k}}$. 
Then, the differences of the metric- and Ricci-coefficients satisfy for any $\left(u,v,\theta_1,\theta_2\right) \in \mathcal{M}\left(\tau_0,\infty\right) \cap \{v \leq (v_\infty)_n \}$ the estimates
\begin{align} \label{boot1}
w^{2p-2} \int_{S^2_{u,v}} \left[ \|\boldsymbol\Gamma_{p}\|^2 + \|\boldsymbol G_p\|^2 \right] \sqrt{\slashed{g}} d\theta_1 d\theta_2 \leq C  \frac{1}{\tau_n^{2\delta}} \cdot e^{-\frac{P}{4M}\tau\left(u,v\right)} \, ,
\end{align}
\begin{align} \label{boot2}
\int_{u}^{u_{fut}} d\bar{u} \ w^2 \int_{S^2_{\bar{u},v}} \left[ \| \slashed{\nabla} \boldsymbol{\hat{\underline{\chi}}} \|^2 \right] \sqrt{\slashed{g}} d\theta_1 d\theta_2 \leq  C  \frac{1}{\tau_n^{2\delta}}  \cdot e^{-\frac{P}{4M}\tau\left(u,v\right)}  \, .
\end{align}
The difference of curvature satisfies
\begin{align} \label{boot3}
F_\delta\left[\boldsymbol \psi\right] \left(\{u\} \times \left[v,v_{fut}\right] \right) + F_\delta\left[\boldsymbol \psi\right] \left( \left[u,u_{fut}\right] \times \{v\} \right) \leq C  \frac{1}{\tau_n^{2\delta}} \cdot e^{-\frac{P}{4M}\tau\left(u,v\right)}
\end{align}
and
\begin{align} \label{boot4}
\mathcal{E}_\delta\left[\boldsymbol \psi \right] \left(\tau, 2M,r\left(\tau,(v_\infty)_n\right) \right) \leq  C  \frac{1}{\tau_n^{2\delta}}  \cdot e^{-\frac{P}{4M}\tau\left(u,v\right)} \, 
\end{align}
for uniform constant $C$ depending only on $M$. Here $F_\delta$ and $\mathcal{E}_\delta$ are the familiar norms defined in Section \ref{sec:norms} except that the weight $w^5$ is replaced by $w^{5-\delta}$ everywhere. 
\end{proposition}

\begin{proof}
Note that in $\mathcal{M} \left(\tau_{n+k}, \infty, (v_\infty)_n \right)$ both metrics are exactly Schwarzschild and hence the statement of the proposition holds trivially.  The above estimates also hold in $\mathcal{M} \left(\tau_n, \tau_{n+k},(v_\infty)_n \right)$ as one of the solutions is exactly Schwarzschild there, and hence the estimates follow from the bounds already established for individual solutions in Theorem \ref{theo1}.

Finally, the above estimates also hold on the horizon, in view of the fact that the data vanishes up to $\tau_{n}-1$, and on $v=(v_\infty)_n$, the latter following from Corollary \ref{cor:diffestfinal}.  Let us now fix $C_0$, such that all of the above hold in $\mathcal{M} \left(\tau_n, \infty\right)$ as well as on the horizon and $v=(v_\infty)_n$. The difficulty is to extend the bounds to the region $\mathcal{M} \left(\tau_0, \tau_n,(v_{\infty})_n \right)$.
We define the bootstrap region
\begin{align} \label{ellipticco}
&\mathcal{B}  = \mathcal{M} \left(\overline{\tau},\tau_n\right) \cap \{ v \leq (v_{\infty})_n \} \ \ \ \textrm{with} \nonumber \\
\overline{\tau} &= \inf_{\tau_0 < t < \tau_n} \Big\{ \textrm{estimates (\ref{boot1})-(\ref{boot4}) hold in $\mathcal{M}\left(t, \tau_n,(v_{\infty})_n\right)$ with $C= C_{max}^2 C_0$}  \Big\} \nonumber \, .
\end{align}
The region is clearly closed and also non-empty by continuity of the above norms. We will now show that actually the estimates of the Proposition hold in $\mathcal{B}$ with $C<C_{max}^2 C_0$, which implies that $\mathcal{B}$ is open, hence $\mathcal{B}=\mathcal{M}\left(\tau_0, \tau_n\right) \cap \{ v \leq (v_{\infty})_n \}$ and therefore the Proposition. As may be anticipated by the reader, these estimates proceed in exactly the same way as for the quantities themselves. Therefore, we will sketch the remainder of the proof and focus on the only novel feature: The additional quantity $\boldsymbol{\slashed{\Gamma}}$ appearing in the equations of Proposition \ref{prop:bianchidiff}.

\emph{Step 1.} We improve (\ref{boot3}) and (\ref{boot4}) from the Bianchi equations of Proposition \ref{prop:bianchidiff} repeating the analysis of section \ref{sec:curvimprove}. The only difference is that we now apply (\ref{pi1}) with weight $5-\delta$ instead of $5$ and that we need to control the additional $\boldsymbol{\slashed{\Gamma}}$-term, which has not appeared previously. The latter term is easily seen to exhibit sufficient radial decay and is hence estimated using the bootstrap assumption (\ref{boot1}) (recall $\|r  \boldsymbol{\slashed{\Gamma}} \| = \|\boldsymbol G_1\|$). This gains a factor of $\frac{1}{P}$ after spacetime integration and hence improves the estimates (\ref{boot3}) and (\ref{boot4}).\footnote{A more detailed analysis reveals that the relevant term is actually cubic, so one could also exploit the fact that $\tau_0$ is large.}

\emph{Step 2.} With the assumptions on the curvature fluxes improved, we derive the following elliptic estimate from (\ref{ellipticd}):
\begin{align}
\int_{S^2_{u,v}}\| \slashed{\nabla} \boldsymbol{\hat{\underline{ \chi}}} \|^2 \leq \int_{S^2_{u,v}}   \Big[ \|\slashed{\nabla} tr \underline{\boldsymbol \chi}\|^2  + | \underline{\boldsymbol \beta} |^2 + \frac{C}{r^{4}} \|\boldsymbol G_1\|^2 +C \sum_{p_1+p_2\geq 2} r^{2p_1} \|\boldsymbol \Gamma_{p_2}\|^2  \Big] ,
\end{align}
where the area form $\sqrt{\slashed{g}}d\theta_1 d\theta_2$ is implicit. Upon integrating this in $u$ with appropriate $r$-weight to match the curvature flux (i.e.~``multiplying" with $\int du \Omega^2 w^2$), we see that in view of the fact that the $\boldsymbol {\underline{\beta}}$-curvature flux has been improved in \emph{Step 1} and the fact that integrating $\boldsymbol \Gamma$ and $\boldsymbol G$ in $u$ gains a $\frac{1}{P}$-factor via Lemma \ref{lem:3dir}, the estimate (\ref{boot2}) can be improved.

\emph{Step 3.} To improve the $\boldsymbol G_i$-parts of (\ref{boot1}) we integrate the transport equations (\ref{auxd1}) and (\ref{auxd2}) using Lemma \ref{lem:3dir}. Note that the last term in both of these equations has just been improved in \emph{Step 2}, while the other terms are handled precisely as before.

\emph{Step 4.} Finally, to improve the $\boldsymbol \Gamma_p$-part of (\ref{boot1}) we integrate the transport equations of Proposition \ref{maind} precisely as we did for the individual solutions. The anomalous terms are handled precisely as before.
\end{proof}

\begin{remark}
The coupling of hyperbolic and elliptic estimates in the above proof has its origin in the monumental proof of the stability of Minkowski space \cite{ChristKlei} and has been further developed in \cite{KlaiRod, Lydia, LukRod, formationofbh, l2bounded}.
\end{remark}

\section{Spacetimes asymptotically settling down to Kerr} \label{sec:thekerrcase}
In this section we describe the modifications required to treat the more general case of spacetimes asymptotically settling down to a Kerr metric. If we want to stick to our convenient set-up of integrating along characteristics from the horizon and null-infinity, the first step is to express the Kerr-metric in a double-null coordinate system (\ref{maing}). While this is already much more complicated than in the Schwarzschild case, the hard work has already been carried out by Pretorius and Israel \cite{Pretorius}. We briefly summarize their construction. 

\subsection{Step 1: Kerr-metric in double null coordinates}
In Boyer-Lindquist coordinates we have (away from the pole at $\theta = 0$)
\[
g_{Kerr} = \frac{\Sigma}{\Delta} dr^2 + \Sigma d\theta^2 +R^2 \sin^2 \theta d\phi^2 - \frac{4ma r \sin^2 \theta}{\Sigma} d\phi dt - \left(1-\frac{2mr}{\Sigma}\right) dt^2
\]
with
\[
\Sigma = r^2 + a^2 \cos^2 \theta \ \ \ , \ \ \ R^2 = r^2 +a^2 + \frac{2ma^2 r \sin^2 \theta}{\Sigma} \ \ \ , \ \ \ \Delta= r^2 -2mr +a^2 \, .
\]
In \cite{Pretorius}, the eikonal equation is solved to determine a tortoise coordinate $r_\star \left(r,\theta\right)$ and a coordinate $\lambda \left(r,\theta\right)$ such that in particular the hypersurfaces
\[
2u = t- r_\star \ \ \ , \ \ \ 2v = t+r_\star
\]
are characteristic. The relation between Boyer-Lindquist $\left(t,r,\theta,\phi\right)$-coordinates and the new $\left(t,r^\star,\lambda,\phi\right)$ or (the trivially related)  $\left(u,v,\lambda,\phi\right)$-coordinates is implicit. In the new coordinates, the Kerr-metric can be expressed as
\begin{align} \label{gken}
g_{Kerr} = -4\frac{\Delta}{R^2} du dv + \frac{L^2}{R^2} d\lambda^2 + R^2 \sin^2 \theta \left(d\phi - \omega_B \left(du + dv\right)\right)^2 \, ,
\end{align}
where $\omega_B = \frac{2mar}{\Sigma R^2}$, $L$ is a function of $\lambda, u$ and $v$. Finally, any $r$ and $\theta$ appearing in (\ref{gken}) is now to be interpreted as a function of $r^\star$ and $\lambda$: $r\left(r^\star=v-u,\lambda\right)$ and $\theta\left(r^\star=v-u,\lambda\right)$. All these relations are implicit but at least the asymptotic behavior (large $r$) of the metric quantities can be obtained. For Schwarzschild, one has the relations $L=\frac{r^2}{\sin 2\theta}$ and $\lambda = \sin^2 \theta$ recovering the familiar Eddington-Finkelstein form of the metric. To express the metric in the form (\ref{maing}), we do the coordinate transformation
\[
\phi = \tilde{\phi} + h\left(u,v,\lambda\right) \, ,
\]
where we require only that $\frac{\partial h}{\partial u} = \omega_B \left(u,v,\lambda\right)$ and $\frac{\partial h}{\partial v} = -\omega_B \left(u,v,\lambda\right)$.\footnote{In other words, $h$ depends only on $r_\star$ and $\lambda$ just as $\omega_B$ depends only on $r^\star$ and $\lambda$.} This allows us to express the metric as
\[
g_{Kerr} = -4\Omega_{Kerr}^2 du dv + \slashed{g}^{Kerr}_{CD} \left(d\theta^C - b_{Kerr}^C dV\right) \left(d\theta^D - b_{Kerr}^D dV\right)
\]
with $\theta^1 = \theta^\lambda =  \lambda$ and $\theta^2 = \theta^{\tilde{\phi}} = \tilde{\phi}$ and
\begin{align}
\Omega_{Kerr}^2 = \frac{\Delta}{R^2} \ \ \ , \ \ \ b_{Kerr}^\lambda = 0 \ \ \ , \ \ \ b_{Kerr}^{\tilde{\phi}} = 2\omega_B  \ \ \ , \ \ \ \slashed{g}^{Kerr}_{\tilde{\phi} \tilde{\phi}} = R^2 \sin^2 \theta\nonumber \\
\slashed{g}^{Kerr}_{\lambda \lambda} = \frac{L^2}{R^2} + \left(\frac{\partial h}{\partial \lambda}\right)^2 R^2 \sin^2 \theta \ \ \ , \ \ \ \slashed{g}^{Kerr}_{\lambda \tilde{\phi}} =2 \left(\frac{\partial h}{\partial \lambda}\right) R^2 \sin^2 \theta \, .
\end{align}
Note that $\omega_B \rightarrow \frac{a}{r_+^2 + a^2}$ becomes constant as $r \rightarrow r_+$, the largest root of $\Delta_-$.\footnote{Observe that now $b_{Kerr} \neq 0$ on the horizon $\mathcal{H}^+$, which differs from our previous choice of gauge in Section \ref{sec:gaugefix}. To achieve $b=0$ one could do the coordinate transformation $\tilde{\phi} = \tilde{\phi}^\prime + \frac{a}{r_+^2 +a^2}v$. However, this coordinate system is not very intuitive near infinity as the new $b$ does not decay at infinity. Consequently, one may do the local theory in the $b_{hoz}=0$-system and then change to the convenient system above in which $b_{hoz}= \left(0,\frac{a}{r_+^2 +a^2}\right)$.} We will refer to the above as the Eddington-Finkelstein form of the Kerr metric. From here, we could define the regular Kruskal coordinates as in \cite{Pretorius}, i.e.~set
\[
U = -e^{-\kappa u} \qquad \textrm{and} \qquad V = e^{\kappa v} \qquad \textrm{where} \qquad \kappa = \frac{\sqrt{m^2-a^2}}{r_+^2+a^2}
\]
is the surface gravity of the horizon. For the extremal case, where $\Delta_-$ has double zero at $r_+$ and $v-u \sim -\left(r-r_+\right)^{-1}$ near the horizon (instead of $v-u \sim \log \left(r-r_+\right)$ in the non-extremal case), one may set $U=\arctan u -\frac{\pi}{2}$ and $V=-\arccot v$ to obtain regular double null (Kruskal-type) coordinates.

\subsection{Step 2: Renormalisation} \label{sec:thekerrcasere}
One now proceeds as in the Schwarzschild case, i.e.~one fixes the differentiable manifold $\mathcal{M}_{M,a}$ of Kerr with its coordinate atlas $\left(u,v,\lambda,\tilde{\phi}\right)$ and equips $\mathcal{M}_{M,a}$ (or appropriate subsets thereof) with metrics $g$ of the form (\ref{maing}) arising from an appropriate initial value formulation with data on the horizon and at infinity. As before, one employs the regular frame $e_3 = \frac{1}{\Omega^2}\partial_u$, $e_4 = \partial_v + b^A \partial_A$, $e_1 = \partial_\lambda$, $e_2 = \partial_{\tilde{\phi}}$ to define the Ricci- and curvature components. The fact that $\Omega$, $\slashed{g}$ and $b$ are known for Kerr as functions of $u,v,\lambda$ then allows for comparison (``renormalization'') of the actual metric $g$ with the Kerr-metric we are converging to. 

We express both the Bianchi- and the null-structure equations, as equations for \emph{decaying} quantities. This is computationally more involved, since many more quantities are non-vanishing for Kerr -- especially since we are not working in the algebraically special frame of Kerr but in a frame arising from a double-null foliation! The key, however, is that we do not need explicit expressions for the values of $\alpha_{Kerr}$, $\beta_{Kerr}$ etc. It suffices to know that these quantities are uniformly bounded with all (regular frame) derivatives, satisfy the familiar radial decay (\ref{stabmindec}) at infinity and that they obey the Bianchi equations ``at the lowest order". For instance, using Lemma \ref{teco}, the Bianchi equation for $\beta$ may be written (ignoring quadratic terms for simplicity)
\begin{align}
\slashed{\nabla}_4 \left( \beta - \beta_{Kerr} \right) + 2 tr \chi  \left( \beta - \beta_{Kerr} \right) - \hat{\omega} \left( \beta - \beta_{Kerr} \right) = \slashed{div} \left(\alpha - \alpha_{Kerr} \right) \nonumber \\ +  \left(\underline{\eta} - \underline{\eta}_{Kerr} + 2 \zeta   - 2\underline{\zeta}_{Kerr} \right) \cdot \alpha_{Kerr} + \left(\underline{\eta}_{Kerr}+ 2 \zeta_{Kerr} \right) \cdot \left(\alpha - \alpha_{Kerr} \right)  \nonumber \\
- 2 \left(tr \chi - tr \chi_{Kerr} \right) \beta_{Kerr} + \left(\hat{\omega}-\hat{\omega}_{Kerr}\right)\beta_{Kerr} + \left(\slashed{g}^{-1} - \slashed{g}^{-1}_{Kerr} \right) \slashed{\nabla} \alpha_{Kerr}\nonumber \\ 
+ \slashed{g}^{-1}_{Kerr} \left(\slashed{\Gamma} - \slashed{\Gamma}_{Kerr} \right) \alpha_{Kerr}+ \left(b-b_{Kerr} \right) \slashed{\nabla} \beta_{Kerr} + \left(\slashed{g}^{-1} - \slashed{g}^{-1}_{Kerr} \right) \chi_{Kerr} \beta_{Kerr} \nonumber \\ + \slashed{g}^{-1}_{Kerr} \left(\chi - \chi_{Kerr} \right) \beta_{Kerr} + \left(b-b_{Kerr}\right) \slashed{\Gamma}_{Kerr} \beta_{Kerr} + b_{Kerr} \left(\slashed{\Gamma} - \slashed{\Gamma}_{Kerr}\right) \beta_{Kerr}  \nonumber \, .
\end{align}
All other Bianchi equations can be treated similarly. We observe that Proposition \ref{uneq2} on the structure of the Bianchi equations still holds with one notable modification: The terms involving the difference of the Christoffel-symbols are new and we comment on them in Section \ref{sec:Kerrest}.

Similarly, it can be checked that the null-structure equations are still of the schematic form of Proposition \ref{uneq}.

Finally, commutation with the frame derivatives also proceeds as before: Recall that nowhere did we rely on the fact that the spacetimes we constructed were almost spherically symmetric. The commutation process with the $\slashed{\nabla}$-operators required only keeping track of the radial decay at infinity, no ``almost Killing'' properties were being used. In view of this, the results of Section \ref{sec:commute} continue to hold because the additional angular momentum terms appearing for Kerr are seen to contribute at higher order in terms of decay in $r$ (in the asymptotically flat coordinate system that we are using) and hence not affect the cancellations of the main terms exploited in the results of Section \ref{sec:commute}.
\subsection{Step 3: Estimates} \label{sec:Kerrest}
To measure the exponential decay (and in particular, to state the bootstrap assumptions) we will need the analogue of the function $\tau$ for the Kerr metric. Such a function is defined in \cite{dafrodlargea}; in fact, we adapted the latter Kerr-construction for our Schwarzschild case when we defined (\ref{fpiece})!

The estimates now proceed as in the Schwarzschild case. In particular, as the null-condition near infinity has already been made manifest in Step 2, the curvature estimates of Proposition \ref{prop:improve1} continue to hold with the same radial weights. There are more error-terms but all of them are handled as in Section \ref{errorest}, with one exception: As observed in Step 2 above, the error will now contain the difference of the Christoffel-symbols on the spheres already at the level of the renormalized equations themselves (and not only at the level of differences as in Section \ref{sec:convergence}). In any case, we already resolved this additional difficulty when we estimated differences (the elliptic estimates proven in Section \ref{sec:diffest} can of course be proven also for the solutions themselves!): One simply bootstraps (\ref{boot1}) and (\ref{boot2}), and retrieves the latter via the elliptic estimate (\ref{ellipticco}). Finally, for the estimates on the $\Gamma_p$ one repeats Proposition \ref{prop:improve2} with the same weights. Again, the only additional difficulty are the new Christoffel terms which requires an additional estimate, which we already carried out in Section \ref{sec:diffest} at the level of differences.

\section{Robinson-Trautman metrics} \label{sec:RobTraut}

In this section we briefly review the family of Robinson-Trautman metrics and compare them with the metrics obtained from Theorem \ref{theo:full}. 

\subsection{Construction} \label{sec:rotr1}
Fix a metric $\overline{g}_0\left(\theta\right)$ on $S^2$ and define the $u$-dependent family of metrics
\[
\bar{g}\left(u,\theta\right) = e^{2\lambda\left(u,\theta\right)} \overline{g}_0\left(\theta\right) \, .
\]
Fix also a constant $M$ and a smooth seed function $\lambda_0\left(\theta\right)$ on $S^2$.
Then consider the following parabolic problem
\begin{equation} \label{Calabi}
\begin{cases} \phantom{rc}
\Delta_{\bar{g}} R_{\bar{g}} = \frac{1}{24M} \partial_u \lambda & \text{for $u>u_0$} \\ 
\lambda \left(u_0, \theta\right) = \lambda_0 \left(\theta\right) & \text{at $u=u_0$} \, 
\end{cases} 
\end{equation} 
with $R_{\bar{g}}$ the scalar curvature of $\bar{g}$.
The above equation may be written as $\bar{g}_{AB} \Delta_{\bar{g}} R_{\bar{g}} = \frac{1}{12M} \partial_u \bar{g}_{AB}$, and is known as the Calabi equation in the literature \cite{Calabi}.\footnote{While (\ref{Calabi}) and their associated Robinson-Trautman metrics below can be considered for other topologies, we focus here on the case of $S^2$.} In \cite{Chruscielexistence}, Chru\'sciel proved global existence and (exponential) convergence to the round metric as $u \rightarrow \infty$ for the parabolic problem (\ref{Calabi}).

As observed by Robinson and Trautman \cite{RobTraut}, from a solution to (\ref{Calabi}) one can in turn construct a solution of the Einstein vacuum equations (\ref{vacEq}). Let $\mathcal{M} =\left[u_0,\infty\right) \times \left(0,\infty\right) \times S^2$ and equip it with the metric
\begin{align} \label{RobTraut}
g = -\left( \frac{R_{\bar{g}}}{2} + \frac{r}{12M} \Delta_{\bar{g}} R_{\bar{g}} - \frac{2M}{r} \right) du^2 - 2 du dr + r^2  e^{2\lambda\left(u,\theta\right)} \left(\overline{g}_0\right)_{AB} dx^A dx^B  \, .
\end{align}
Then $\left(\mathcal{M},g\right)$ is Ricci-flat. It is also asymptotically flat at null-infinity, which can be seen, for instance, by passing to Bondi coordinates \cite{Winicour}. The vectorfield $\partial_r$ generates a shear-free congruence of null-geodesics which is hypersurface orthogonal (to the hypersurfaces of constant $u$). These facts make the metric algebraically special.

\subsection{Global Structure}
The global structure of Robinson-Trautman (RT) metrics (\ref{RobTraut}) has been understood by Chru\'sciel in \cite{Chruscielglobal} and is best illustrated by a Penrose diagram:
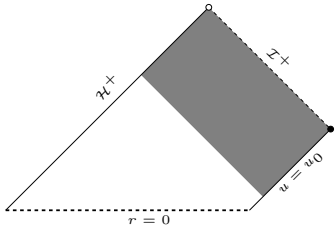
\[
\input{robtraut.pstex_t}
\]
For us, the shaded region is the most interesting. In general, there exists no Robinson-Trautman extension past the null hypersurface $u=u_0$, as this would involve solving a backwards heat equation.\footnote{Of course, one could extend the metric backwards \emph{outside} the Robinson-Trautman class provided suitable data is prescribed on null-infinity. Cf.~point 6 of Section \ref{sec:mainpoints}.} To investigate the question of whether the metric extends smoothly to the future of $\mathcal{H}^+$, i.e.~whether the horizon is smooth, Chru\'sciel  \cite{Chruscielnonsmooth} derived an asymptotic expansion of the metric towards $u \rightarrow \infty$  from an asymptotic expansion of the solution to the parabolic equation (\ref{Calabi}). After a coordinate change $\hat{u}=\exp\left(-u/4M\right)$, $\hat{v}=\exp \left((u+2r+4M\log(r-2M))/4M \right)$, the horizon is seen as the boundary $\hat{u}=0$ of a larger manifold, $\overline{\mathcal{M}}$. The results then are the following: 1) The metric on the horizon is Schwarzschild up to terms of order $\hat{u}^6$ and hence, in particular, there is no energy flux through the horizon for Robinson-Trautman metrics. 2)
 A careful expansion of the metric in powers of $\hat{u}$ will be completely regular up to some power of $\hat{u}$ but generically involve terms of the form $\hat{u}^n \log |\hat{u}|$ at some level. This in turn can be used to prove that generically the metric is not $C^m$-extendible through $\mathcal{H}^+$ for some large $m$.\footnote{At least not to a manifold $\overline{\mathcal{M}}$ smoothly foliated by spheres coinciding with the $u=const$, $r=const$ spheres in $\mathcal{M}$. See \cite{Chruscielnonsmooth}.} 3) For a special (non-generic) choice of initial data the level at which the log-terms appear can \emph{in principle}\footnote{Chru\'sciel \cite{Chruscielnonsmooth} shows existence of data making the first log-term vanish.} be pushed to larger $n$. However, the conjecture (formulated in \cite{Chruscielglobal}) is that the only RT-metric admitting a completely smooth extension through $\mathcal{H}^+$ is the Schwarzschild metric. 
 
\subsection{Remarks}
As seen in Section \ref{sec:rotr1}, Robinson-Trautman metrics 
are algebraically special and have less formal degrees of freedom than the solutions arising from Theorem \ref{theo:full}, as the former are completely parametrized by the metric on a single sphere, while the scattering data for the solutions of Theorem \ref{theo:full} admit the complete functional degrees of freedom along both the horizon (where the RT solutions are necessarily trivial) and null-infinity. Moreover, it is not clear whether there is an analogue of the Robinson-Trautman construction producing spacetimes asymptotically settling down to  Kerr.

Interestingly, the Robinson-Trautman class is in fact not included in the class of solutions we construct in Theorem \ref{theo:full}. This is simply because the exponential decay rate exhibited by the Robinson-Trautman solutions is ``borderline" (in the sense that it is the minimum decay rate expected from the blue-shift heuristics) while we have made no attempt to optimize our exponential decay rate $P$. It is an interesting problem to determine the minimal exponential decay rate of our solutions and to retrieve the RT solutions from our framework. See also Remark \ref{rem:Pdep} in this context.

\appendix \label{appendix}

\section{The remaining null-structure equations}
For completeness, we collect here those null-structure equations not listed in Sections \ref{sec:rc} and \ref{sec:mq}:

\begin{align} \label{4uhatc}
\slashed{\nabla}_4 \underline{\hat{\chi}} + \frac{1}{2} tr \chi \underline{\hat{\chi}} = -2 \slashed{\mathcal{D}}_2^\star \underline{\eta} - \hat{\omega} \underline{\hat{\chi}} - \frac{1}{2} tr \underline{\chi} \hat{\chi} + \underline{\eta} \hat{\otimes} \underline{\eta} \, ,
\end{align}
\begin{align} \label{4utrc}
\slashed{\nabla}_4 \left(  tr \underline{\chi}  - tr \underline{\chi}_\circ\right) =\left[ -\frac{1}{2} tr \chi - \hat{\omega} \right] \left(tr \underline{\chi} - tr \underline{\chi}_\circ\right) -\frac{1}{2} tr \underline{\chi}_\circ \left(tr \chi - tr \chi_\circ\right)   \nonumber \\ 
+ 2 \slashed{div} \underline{\eta} - tr \underline{\chi}_\circ \left(\hat{\omega} - \hat{\omega}_\circ \right) - \hat{\chi}\cdot \underline{\hat{\chi}} + 2 \underline{\eta} \cdot \underline{\eta} + 2 \left(\rho-\rho_\circ\right) \, ,
\end{align}
\begin{align} \label{3hatc}
\slashed{\nabla}_3 {\hat{\chi}} + \frac{1}{2} tr \underline{\chi} {\hat{\chi}} = -2 \slashed{\mathcal{D}}_2^\star {\eta} - \underline{\hat{\omega}} {\hat{\chi}} - \frac{1}{2} tr{\chi}  \underline{\hat{\chi}} + {\eta} \hat{\otimes} {\eta} \, ,
\end{align} 
\begin{align} \label{3trc}
\slashed{\nabla}_3 \left(  tr {\chi}  - tr  {\chi}_\circ\right) = -\frac{1}{2} tr \chi \left(tr \underline{\chi} - tr \underline{\chi}_\circ\right) -\frac{1}{2} tr \underline{\chi}_\circ \left(tr \chi - tr \chi_0\right) + 2 \slashed{div} \eta \nonumber \\ 
- \hat{\chi}\cdot \underline{\hat{\chi}} + 2 \eta \cdot \eta + 2 \left(\rho-\rho_\circ\right) \, ,
\end{align}
\begin{align} \label{Gauss}
K=- \frac{1}{4} tr \chi tr \underline{\chi} + \frac{1}{2} \hat{\chi} \cdot \hat{\underline{\chi}} - \rho \, ,
\end{align}
\begin{align}
\slashed{curl} \eta = - \frac{1}{2} \hat{\chi} \wedge \underline{\hat{\chi}} + \sigma
\ \ \ \ \ , \ \ \ \ \ 
\slashed{curl} \underline{\eta} = + \frac{1}{2} \hat{\chi} \wedge \underline{\hat{\chi}} - \sigma \, ,
\end{align}
\begin{align} \label{elliptic}
\left(\slashed{g}^{-1}\right)^{BC} \slashed{\nabla}_C \hat{\underline{\chi}}_{BA} = \frac{1}{2} \slashed{\nabla}_A tr \underline{\chi} - \left(\slashed{g}^{-1}\right)^{BC} \underline{\eta}_B \hat{\underline{\chi}}_{AC}  + \frac{1}{2} tr \underline{\chi} \  \underline{\eta}_A + \underline{\beta}_A \, ,
\end{align}
\begin{align} \label{elliptic2}
\left(\slashed{g}^{-1}\right)^{BC} \slashed{\nabla}_C \hat{{\chi}}_{BA} = \frac{1}{2} \slashed{\nabla}_A tr {\chi} - \left(\slashed{g}^{-1}\right)^{BC}{\eta}_B \hat{{\chi}}_{AC}  + \frac{1}{2} tr {\chi} \  {\eta}_A - {\beta}_A \, .
\end{align}

\bibliographystyle{hacm}
\bibliography{thesisrefs}

\end{document}